\documentclass[11pt]{article}

\usepackage{graphicx, mathtools, amssymb, latexsym, amsmath, amsfonts, amsthm, bbm, tikz}

\usetikzlibrary{calc, arrows, shapes}
\usetikzlibrary{matrix,positioning, backgrounds, patterns, decorations.markings}

\usepackage{xcolor}
\definecolor{DarkGreen}{rgb}{0.1,0.5,0.1}
\definecolor{DarkRed}{rgb}{0.5,0.1,0.1}
\definecolor{DarkBlue}{rgb}{0.1,0.1,0.5}
\definecolor{Black}{rgb}{0,0,0}

\usepackage[small]{caption}
\usepackage[pdftex]{hyperref}
\hypersetup{
    unicode=false,          
    pdftoolbar=true,        
    pdfmenubar=true,        
    pdffitwindow=false,      
    pdfnewwindow=true,      
    colorlinks=true,       
    linkcolor=DarkBlue,          
    citecolor=DarkGreen,        
    filecolor=DarkGreen,      
    urlcolor=DarkBlue,          
}

\usepackage{fancyhdr}

\mathtoolsset{showonlyrefs}

\newcommand\white[1]{{\color{white} #1}}

\theoremstyle{plain}
\newtheorem{theorem}{Theorem}[section]

\newtheorem{lemma}[theorem]{Lemma}
\newtheorem{claim}[theorem]{Claim}

\newtheorem{proposition}[theorem]{Proposition}
\newtheorem{fact}[theorem]{Fact}
\newtheorem{corollary}[theorem]{Corollary}

\theoremstyle{definition}
\newtheorem{definition}[theorem]{Definition}
\newtheorem{question}[theorem]{Question}

\newcommand\ncon{(\mathrm{N}\cap\mathrm{coN})\mathrm{TIME}}
\newcommand\nconplus{(\mathrm{N}\cap\mathrm{coN})^{+}\mathrm{TIME}}
\newcommand\TIME{\mathrm{TIME}}
\newcommand\ntime{\mathrm{NTIME}}

\newcommand\ceil[1]{\lceil#1\rceil}
\newcommand\floor[1]{\lfloor#1\rfloor}
\newcommand\defeq{\ensuremath{\stackrel{\rm def}{=}}} 

\newcommand\ind[1]{\text{ind}(#1)}

\newcommand\undirectedfigure[2]{
\node[anchor=east](y1)  at (0.98, 1) {$1$};
\node[anchor=east](y54) at (0.98, 1.25) {$5/4$};
\node[anchor=east](y43) at (0.98, 1.3333) {$4/3$};
\node[anchor=east](y32) at (0.98, 1.5) {$3/2$};
\node[anchor=east](y2)  at (0.98, 2) {2};
\coordinate (x1)  at (1,1.0);
\coordinate (x32) at (1.5,1.0);
\coordinate (x85) at (1.6,1.0);
\coordinate (x53) at (1.66666,1.0);
\coordinate (x74) at (7.0/4.0,1.0);
\coordinate (x158) at (15.0/8.0,1.0);
\coordinate (x2)  at (2,1.0);
\node[anchor=north] at (x1)   {$1$};
\node[anchor=north] at (x32)  {$\frac{3}{2}$};
\node[anchor=north] at (x53)  {$\frac{5}{3}$};
\node[anchor=north] at (x74)  {$\frac{7}{4}$};
\node[anchor=north] at (x2)   {$2$};
\coordinate (r2) at (1.5, 2);
\coordinate (r32) at (#1, 1.5);
\coordinate (r43) at (13.0/8.0, 1.3333);
\coordinate (r54) at (18.0/11.0, 1.25);
\coordinate (r65) at (23.0/14.0, 1.2);
\coordinate (r76) at (28.0/17.0, 1.16666);
\coordinate (r87) at (33.0/20.0, 1.125);
\coordinate (r98) at (38.0/23.0, 10.0/9.0);
\coordinate (r109) at (43.0/26.0, 1.1);
\coordinate (rNew) at (5.0/3.0, 1.5);
\coordinate (b2) at (1,2);
\coordinate (b32) at (1.5, 1.5);
\coordinate (b43) at (1.75, 1.3333);
\coordinate (b54) at (1.875, 1.25);
\coordinate (b65) at (1.9375, 1.2);
\coordinate (b76) at (1.96875, 1.16666);
\coordinate (b87) at (1.9844, 1.142857);
\coordinate (b98) at (1.992, 1.124);
\coordinate (b109) at (1.996, 1.111111);
\coordinate (b1110) at (1.998, 1.1);
\coordinate (b1615) at (1.99993, 1.066666);
\coordinate (b2120) at (1.999998, 1.05);
\coordinate (b1) at (2,1);

\draw[fill=red!20,draw=none] (1,1) rectangle (r2);
\draw[fill=red!30,draw=none] (1.5,1) rectangle (r32);
\draw[fill=blue!20,draw=none] (2,2) rectangle (b2);
\draw[fill=blue!20,draw=none] (2,2) rectangle (b32);
\draw[fill=blue!20,draw=none] (2.001,2) rectangle (b1);
\draw[blue] (y2)--(r2)--(b32)--(2.0,1.5);
\begin{pgfonlayer}{above}
\node[draw, fill=red!20, circle,scale=0.5] at (r2) {};
\node[draw, fill=red!20, circle,scale=0.5] at (r32) {};
\node[draw, fill=blue!20, circle,scale=0.5] at (b2) {};
\node[draw, fill=blue!20, circle,scale=0.5] at (b32) {};
\end{pgfonlayer}

\node[fill=red!20] (rw13) at (1.25, 1.5) {\color{black}\cite{RodittyW13}};
\node[fill=red!30, rotate=270] (brsvw18) at (1.56, 1.25) {\color{black}\cite{BackursRSWW18}};
\node[] (apsp) at (1.1, 2.05) {\color{black}APSP};
\node[] (bfs) at (2.07, 1.05) {\color{black}#2};
\node[align=left] (chechik) at (1.75, 1.75) {\color{black}\cite{RodittyW13}\\\cite{ChechikLRSTW14}};

\draw[blue] (2,1.5)--(b1);
\node[draw, fill=blue!30, circle,scale=0.5] at (b1) {};
}
\newcommand\figureaxes{
\draw[->] (0.98,1) to (2.05,1);
\draw[->] (1,0.98) to (1,2.05);
\node[anchor=west] at (2.03,1) {Apx};
\node[anchor=south] at (1,2.05) {$k$};
\node[anchor=east, rotate=90] at (0.81,1.75) {$\tilde O(m^k)$ time};
}

\newcommand\CGRfigure{
\draw[blue] (1.5,1.5)--(2,1.5);
\draw[fill=blue!30,draw=none] (2,1.50) rectangle (b43);
\draw[fill=blue!30,draw=none] (2,1.50) rectangle (b54);
\draw[fill=blue!30,draw=none] (2,1.50) rectangle (b65);
\draw[fill=blue!30,draw=none] (2,1.50) rectangle (b76);
\draw[fill=blue!30,draw=none] (2,1.50) rectangle (b1);
\draw[fill=blue!30,draw=none] (2,1.50) -- (b76) -- (b87) -- (b98) -- (b109) -- (b1110)--(b1615)--(b2120)--(b1)-- cycle;
\draw[blue] (1.75,1.5)--(b43)
             --(1.875,1.3333)--(b54)
             --(1.9375,1.25)--(b65)
             --(1.96875,1.2)--(b76) -- (b87) -- (b98) -- (b109) --(b1110)--(b1615)--(b2120)--(b1);
\node[draw, fill=blue!30, circle,scale=0.5] at (b43) {};
\node[draw, fill=blue!30, circle,scale=0.5] at (b54) {}; 
\node[draw, fill=blue!30, circle,scale=0.5] at (b65) {};
\node[draw, fill=blue!30, circle,scale=0.5] at (b76) {};
\node[draw, fill=blue!30, circle,scale=0.5] at (b87) {};
\node[draw, fill=blue!30, circle,scale=0.5] at (b98) {};
\node[draw, fill=blue!30, circle,scale=0.5] at (b1) {};
\node[] (cgr16) at (1.88, 1.41) {\color{black}\cite{CairoGR16}*};

}
\newcommand\unweightedresultfigure{
  \draw[fill=red!70,draw=none] (8/5,1) rectangle (5/3,3/2);
  \node[draw, fill=red!70, circle,scale=0.5] at (5/3,3/2) {};
  \node[rotate=270] (new) at (1.635, 1.25) {Thm~\ref{thm:53}};
}
\newcommand\directedunweightfigure{
\draw[fill=red!30,draw=none] (1.59,1) rectangle (r43);
\draw[fill=red!30,draw=none] (x85) rectangle (r54);
\draw[fill=red!30,draw=none] (x85) rectangle (r65);
\draw[fill=red!30,draw=none] (x85) rectangle (r76);
\draw[fill=red!30,draw=none] (x85) -- (r87) -- (x53) -- cycle;
\node[draw, fill=red!20, circle,scale=0.5] at (r43) {};
\node[draw, fill=red!20, circle,scale=0.5] at (r54) {};
\node[draw, fill=red!20, circle,scale=0.5] at (r65) {};
\node[draw, fill=red!20, circle,scale=0.5] at (r76) {};
\node[draw, fill=red!20, circle,scale=0.5] at (r87) {};
\node[draw, fill=red!20, circle,scale=0.5] at (r98) {};
\node[draw, fill=red!20, circle,scale=0.5] at (r109) {};
}

\newcommand\directedresultfigure[1]{
\coordinate (new2) at (5/3,3/2);
\coordinate (new3) at (7/4,4/3);
\coordinate (new4) at (9/5,5/4);
\coordinate (new5) at (11/6,6/5);
\coordinate (new6) at (13/7,7/6);
\coordinate (new7) at (15/8,8/7);
\coordinate (new8) at (17/9,9/8);
\coordinate (new9) at (19/10,10/9);
\coordinate (new10) at (21/11,11/10);
\coordinate (new11) at (23/12,12/11);
\coordinate (new12) at (25/13,13/12);
\coordinate (new13) at (27/14,14/13);
\coordinate (new14) at (29/15,15/14);
\coordinate (new20) at (41/21,21/20);
\coordinate (new30) at (61/31,31/30);

\draw[fill=red!70,draw=none] (x85) rectangle (new2);
\draw[fill=red!70,draw=none] (x85) rectangle (new3);
\draw[fill=red!70,draw=none] (x85) rectangle (new4);
\draw[fill=red!70,draw=none] (x85) rectangle (new5);
\draw[fill=red!70,draw=none] (x85) rectangle (new6);
\draw[fill=red!70,draw=none] (x85) rectangle (new7);
\draw[fill=red!70,draw=none] (x85) rectangle (new8);
\draw[fill=red!70,draw=none] (x85) rectangle (new9);
\draw[fill=red!70,draw=none] (x85) rectangle (new10);
\draw[fill=red!70,draw=none] (x85) rectangle (new11);
\draw[fill=red!70,draw=none] (x85) rectangle (new12);
\draw[fill=red!70,draw=none] (x85) rectangle (new13);
\draw[fill=red!70,draw=none] (x85) rectangle (new14);
\draw[fill=red!70,draw=none] (29/15,1) -- (new14) -- (2,1) -- cycle;
\node[draw, fill=red!70, circle,scale=0.5] at (new2) {};
\node[draw, fill=red!70, circle,scale=0.5] at (new3) {};
\node[draw, fill=red!70, circle,scale=0.5] at (new4) {};
\node[draw, fill=red!70, circle,scale=0.5] at (new5) {};
\node[draw, fill=red!70, circle,scale=0.5] at (new6) {};
\node[draw, fill=red!70, circle,scale=0.5] at (new7) {};
\node[draw, fill=red!70, circle,scale=0.5] at (new8) {};
\node[draw, fill=red!70, circle,scale=0.5] at (new9) {};
\node[draw, fill=red!70, circle,scale=0.5] at (new10) {};
#1

}
\newcommand\nsethfigure[1]{
  \coordinate (nseth1) at (5/3,4/3);
  \coordinate (nseth2) at (7/4,5/4);
  \coordinate (nseth3) at (9/5,6/5);
  \coordinate (nseth4) at (11/6,7/6);
  \coordinate (nseth5) at (13/7,8/7);
  \coordinate (nseth6) at (15/8,9/8);
  \coordinate (nseth7) at (17/9,10/9);
  \draw[pattern=crosshatch dots, pattern color=cyan!100,draw=none] (nseth1) rectangle (2,3/2);
  \draw[pattern=crosshatch dots, pattern color=cyan!100,draw=none] (nseth2) rectangle (2,4.001/3);
  \draw[pattern=crosshatch dots, pattern color=cyan!100,draw=none] (nseth3) rectangle (2,5.001/4);
  \draw[pattern=crosshatch dots, pattern color=cyan!100,draw=none] (nseth4) rectangle (2,6.001/5);
  \draw[pattern=crosshatch dots, pattern color=cyan!100,draw=none] (nseth5) rectangle (2,7.001/6);
  \draw[pattern=crosshatch dots, pattern color=cyan!100,draw=none] (nseth6) rectangle (2,8.001/7);
  \draw[pattern=crosshatch dots, pattern color=cyan!100,draw=none] (nseth7) rectangle (2,9.001/8);
  \draw[pattern=crosshatch dots, pattern color=cyan!100,draw=none] (nseth7) -- (2,10/9) -- (2,1) -- cycle;
  \begin{pgfonlayer}{above}
    \node[draw, fill=white, pattern=none, pattern color=cyan!100, circle,scale=0.5] at (nseth1) {};
    \node[draw, fill=white, pattern=none, pattern color=cyan!100, circle,scale=0.5] at (nseth2) {};
    \node[draw, fill=white, pattern=none, pattern color=cyan!100, circle,scale=0.5] at (nseth3) {};
    \node[draw, fill=white, pattern=none, pattern color=cyan!100, circle,scale=0.5] at (nseth4) {};
    \node[draw, fill=white, pattern=none, pattern color=cyan!100, circle,scale=0.5] at (nseth5) {};
    \node[draw, fill=white, pattern=none, pattern color=cyan!100, circle,scale=0.5] at (nseth6) {};
    \node[draw, fill=white, pattern=none, pattern color=cyan!100, circle,scale=0.5] at (nseth7) {};
  \end{pgfonlayer}
  #1
}

\newcommand\undirectedunweightedall{
  \begin{tikzpicture}[xscale=5.5, yscale=5.5, font=\scriptsize]
  \pgfdeclarelayer{background}
  \pgfdeclarelayer{above}
  \pgfsetlayers{background,main,above}
  \nsethfigure{
    \node[align=left,rotate=-45] (nsethlabel) at (1.8, 1.14) {Thm~\ref{thm:nseth-intro}};
  }
  \undirectedfigure{8/5}{BFS}
  \CGRfigure
  \unweightedresultfigure
  \figureaxes
  \end{tikzpicture}
}
\newcommand\undirectedweightedall{
  \begin{tikzpicture}[xscale=5.5, yscale=5.5, font=\scriptsize]
  \pgfdeclarelayer{background}
  \pgfdeclarelayer{above}
  \pgfsetlayers{background,main,above}
  \nsethfigure{
    \node[align=left,rotate=-45] (nsethlabel) at (1.8, 1.14) {Thm~\ref{thm:nseth-undir}};
  }
  \undirectedfigure{5/3}{SP}
  \CGRfigure
  \figureaxes
  \end{tikzpicture}
}
\newcommand\directedunweightedall{
  \begin{tikzpicture}[xscale=5.5, yscale=5.5, font=\scriptsize]
  \pgfdeclarelayer{background}
  \pgfdeclarelayer{above}
  \pgfsetlayers{background,main,above}
  \nsethfigure{
    \begin{pgfonlayer}{above}
      \node[align=left] (nsethlabel) at (1.85, 1.4) {Thm~\ref{thm:nseth-intro}};
    \end{pgfonlayer}
  }
  \undirectedfigure{8/5}{BFS}
  \directedresultfigure{
    \node[align=left] (2epslabel) at (1.8, 1.05) {Thm~\ref{thm:2eps}};
  }
  \directedunweightfigure
  \figureaxes
  \end{tikzpicture}
}
\newcommand\directedweightedall{
  \begin{tikzpicture}[xscale=5.5, yscale=5.5, font=\scriptsize]
  \pgfdeclarelayer{background}
  \pgfdeclarelayer{above}
  \pgfsetlayers{background,main,above}
  \directedresultfigure{
    \node[align=left] (2epslabel) at (1.85, 1.05) {Thm~\ref{thm:2eps}};
    \draw[fill=red!40,draw=none] (x53) rectangle (7/4,4/3);
    \node[draw, fill=red!40, circle,scale=0.5] at (7/4,4/3) {};
    \node[rotate=270,font=\tiny] (bon20) at (1.7, 1.17) {\cite{Bonnet20}};
    \coordinate (nsethw) at (5/3,19/13);
    \draw[pattern=crosshatch dots, pattern color=cyan!100,draw=cyan!100] (nsethw) rectangle (2,1.5);
    \node[] (nsethlabel) at (1.85, 1.45) {Thm~\ref{thm:nseth-weight-intro}};
    \begin{pgfonlayer}{above}
      \node[draw, fill=white, pattern=crosshatch dots, pattern color=cyan!100, circle,scale=0.5] at (nsethw) {};
    \end{pgfonlayer}
  }
  \undirectedfigure{5/3}{SP}
  \figureaxes
  \end{tikzpicture}
}

\newcommand\undirectedunweightedtable{
  \def\arraystretch{1.2}
  \begin{tabular}{|l|l|l|}
    \hline
    \textbf{Citation} & \textbf{Runtime} & \textbf{Approx.}  \\ 
    \hline
    \multicolumn{3}{|l|}{\textbf{Upper bounds}} \\
    \hline
    APSP & $\tilde O(m^2)$ & 1  \\
    \hline
    BFS & $O(m)$ & 2  \\
    \hline
    \cite{RodittyW13,ChechikLRSTW14}  & $\tilde O(m^{3/2})$ & 3/2  \\
    \hline
    \cite{CairoGR16} & $\tilde O(m^{1+1/(k+1)})$ & almost $2-2^{-k}$ \\
    \hline
    \multicolumn{3}{|l|}{\textbf{Lower bounds (under SETH)}} \\
    \hline
    \cite{RodittyW13} & $n^{2-o(1)}$ & $3/2-\varepsilon$  \\
    \hline
    \cite{BackursRSWW18} & $n^{3/2-o(1)}$ & $8/5-\varepsilon$  \\
    \hline
    Theorem~\ref{thm:53} & $n^{3/2-o(1)}$ & $5/3-\varepsilon$  \\
    \hline
    \multicolumn{3}{|l|}{\textbf{Non-SETH-hardness (under (NU)NSETH)}} \\
    \hline
    Theorem~\ref{thm:nseth-intro} & $m^{1+1/k+\delta}$ & $2-\frac{1}{k}+\varepsilon$ \\
    \hline
  \end{tabular}
}

\newcommand\undirectedweightedtable{
  \def\arraystretch{1.2}
  \begin{tabular}{|l|l|l|}
    \hline
    \textbf{Citation} & \textbf{Runtime} & \textbf{Approx.}  \\ 
    \hline
    \multicolumn{3}{|l|}{\textbf{Upper bounds}} \\
    \hline
    APSP & $\tilde O(m^2)$ & 1  \\
    \hline
    Single SP & $O(m)$ & 2  \\
    \hline
    \cite{RodittyW13,ChechikLRSTW14}  & $\tilde O(m^{3/2})$ & 3/2  \\
    \hline
    \cite{CairoGR16} & $\tilde O(m^{1+1/(k+1)})$ & almost $2-2^{-k}$ \\
    \hline
    \multicolumn{3}{|l|}{\textbf{Lower bounds (under SETH)}} \\
    \hline
    \cite{RodittyW13} & $n^{2-o(1)}$ & $3/2-\varepsilon$  \\
    \hline
    \cite{BackursRSWW18} & $n^{3/2-o(1)}$ & $5/3-\varepsilon$  \\
    \hline
    \multicolumn{3}{|l|}{\textbf{Non-SETH-hardness (under (NU)NSETH)}} \\
    \hline
    Theorem~\ref{thm:nseth-undir} & $m^{1+1/k+\delta}$ & $2-\frac{1}{k}+\varepsilon$ \\
    \hline
  \end{tabular}
}

\newcommand\directedunweightedtable{
  \def\arraystretch{1.2}
  \begin{tabular}{|l|l|l|}
    \hline
    \textbf{Citation} & \textbf{Runtime} & \textbf{Approx.}  \\ 
    \hline
    \multicolumn{3}{|l|}{\textbf{Upper bounds}} \\
    \hline
    APSP & $\tilde O(m^2)$ & 1  \\
    \hline
    BFS & $O(m)$ & 2  \\
    \hline
    \cite{RodittyW13,ChechikLRSTW14}  & $\tilde O(m^{3/2})$ & 3/2  \\
    \hline
    \multicolumn{3}{|l|}{\textbf{Lower bounds (under SETH)}} \\
    \hline
    \cite{RodittyW13} & $n^{2-o(1)}$ & $3/2-\varepsilon$  \\
    \hline
    \cite{BackursRSWW18} & $n^{1+1/(k-1)-o(1)}$ & $\frac{5k-7}{3k-4}-\varepsilon$  \\
    \hline
    Theorem~\ref{thm:2eps} & $n^{1+1/(k-1)-o(1)}$ & $2-\frac{1}{k}-\varepsilon$  \\
    \hline
    \multicolumn{3}{|l|}{\textbf{Non-SETH-hardness (under (NU)NSETH)}} \\
    \hline
    Theorem~\ref{thm:nseth-intro} & $m^{1+1/k+\delta}$ & $2-\frac{1}{k}+\varepsilon$ \\
    \hline
  \end{tabular}
}

\newcommand\directedweightedtable{
  \def\arraystretch{1.2}
  \begin{tabular}{|l|l|l|}
    \hline
    \textbf{Citation} & \textbf{Runtime} & \textbf{Approx.}  \\ 
    \hline
    \multicolumn{3}{|l|}{\textbf{Upper bounds}} \\
    \hline
    APSP & $\tilde O(m^2)$ & 1  \\
    \hline
    Single SP & $O(m)$ & 2  \\
    \hline
    \cite{RodittyW13,ChechikLRSTW14}  & $\tilde O(m^{3/2})$ & 3/2  \\
    \hline
    \multicolumn{3}{|l|}{\textbf{Lower bounds (under SETH)}} \\
    \hline
    \cite{RodittyW13} & $n^{2-o(1)}$ & $3/2-\varepsilon$  \\
    \hline
    \cite{BackursRSWW18} & $n^{3/2-o(1)}$ & $5/3-\varepsilon$  \\
    \hline
    \cite{Bonnet20} & $n^{4/3-o(1)}$ & $7/4-\varepsilon$  \\
    \hline
    Theorem~\ref{thm:2eps} & $n^{1+1/(k-1)-o(1)}$ & $2-\frac{1}{k}-\varepsilon$  \\
    \hline
    \multicolumn{3}{|l|}{\textbf{Non-SETH-hardness (under (NU)NSETH)}} \\
    \hline
    Theorem~\ref{thm:nseth-weight-intro} & $m^{(1+\delta)19/13}$ & $5/3+\varepsilon$ \\
    \hline
  \end{tabular}
}

\begin{document}
\title{Settling SETH vs. Approximate Sparse Directed Unweighted Diameter (up to (NU)NSETH)}
\author{Ray Li\thanks{Department of Computer Science, Stanford University. Email: \texttt{rayyli@cs.stanford.edu}. Research supported by the National Science Foundation (NSF) under Grant No. DGE - 1656518 and by Jacob Fox's Packard Fellowship.}}
\date{\today}
\maketitle

\begin{abstract}
We prove several tight results on the fine-grained complexity of approximating the diameter of a graph.  First, we prove that, for any $\varepsilon>0$, assuming the Strong Exponential Time Hypothesis (SETH), there are no near-linear time $2-\varepsilon$-approximation algorithms for the Diameter of a sparse directed graph, even in unweighted graphs.  This result shows that a simple near-linear time 2-approximation algorithm for Diameter is optimal under SETH, answering a question from a survey of Rubinstein and Vassilevska-Williams (SIGACT '19) for the case of directed graphs.

In the same survey, Rubinstein and Vassilevska-Williams also asked if it is possible to show that there are no $2-\varepsilon$ approximation algorithms for Diameter in a directed graph in $O(n^{1.499})$ time.  We show that, assuming a hypothesis called NSETH, one \emph{cannot} use a deterministic SETH-based reduction to rule out the existence of such algorithms. 

Extending the techniques in these two results, we \emph{characterize} whether a $2-\varepsilon$ approximation algorithm running in time $O(n^{1+\delta})$ for the Diameter of a sparse directed unweighted graph can be ruled out by a deterministic SETH-based reduction for every $\delta\in(0,1)$ and essentially every $\varepsilon\in(0,1)$, assuming NSETH. This settles the SETH-hardness of approximating the diameter of sparse directed unweighted graphs for deterministic reductions, up to NSETH.
We make the same characterization for randomized SETH-based reductions, assuming another hypothesis called NUNSETH.

We prove additional hardness and non-reducibility results for undirected graphs.
\end{abstract}

\thispagestyle{empty}

\newpage
\setcounter{page}{1}

\section{Introduction}

The diameter $D$ of a graph $G=(V,E)$ is the maximum shortest path distance between two vertices.
A basic algorithmic question, Diameter, is computing or approximating the diameter of a graph.
The diameter is a useful measure of the complexity of a graph or network, so efficient algorithms for Diameter are desirable in practice \cite{WattsS98, BE05, PelegRT12, BorassiCHKMT15, LinWCW16}.
In this work, we prove several tight results on the fine-grained complexity of approximating the diameter of a graph.

Diameter algorithms have been studied in dense graphs, sparse graphs, and special classes of graphs with additional structure \cite{FarleyP80, AingworthCIM99, Eppstein00, CDHP01, CDV02, BBST07, RodittyW13, ChechikLRSTW14, CyganGS15, CairoGR16, Damaschke16, BringmanHM18, Ducoffe18, GKMSW18, BentertN19, CDP19, DHV20}.
In this work, we focus our attention on sparse graphs, i.e., graphs satisfying $m=n^{1+o(1)}$ where $n$ is the number of vertices and $m$ is the number of edges. 
We state our results for generic graphs, but our results are the strongest for sparse graphs, so it is helpful to imagine $m=n^{1+o(1)}$ throughout the paper.

In sparse graphs, the fastest exact Diameter algorithms in fact compute the distances between every pair of vertices, solving the All-Pairs-Shortest-Paths (APSP) problem, which takes $\tilde O(mn)$ time.
Because finding exact algorithms is challenging, it is natural to try to find fast approximation algorithms.
By running shortest path from a single vertex $v$ and returning $\tilde D$, the distance of $v$ to or from the furthest vertex, we obtain a simple 2-approximation in $\tilde O(m)$ time by the triangle inequality. 
Improving on the approximation ratio, a line of work \cite{AingworthCIM99, ChechikLRSTW14, RodittyW13} gave a 3/2-approximation algorithm for Diameter in time $\tilde O(m^{3/2})$.
Cairo, Grossi and Rizzi \cite{CairoGR16} generalized this algorithm to give an almost\footnote{``Almost'' means the algorithm loses an additive factor that depends on the edge weights.}-$(2-\frac{1}{2^k})$-approximation in time $\tilde O(m^{1+\frac{1}{k+1}})$ (for integers $k\ge 2$), but only in undirected graphs.

A natural question is whether these known algorithms are optimal.
A line of work \cite{RodittyW13, ChechikLRSTW14, BackursRSWW18, Bonnet20} based on \emph{fine-grained complexity} has given some partial answers to these questions (see Figure~\ref{fig:res} for results in directed graphs and Figure~\ref{fig:res-undir} for results in undirected graphs).
These works prove conditional hardness results assuming the Strong Exponential Time Hypothesis (SETH) \cite{ImpagliazzoPZ01}.\footnote{SETH states that, for every $\varepsilon>0$, there exists a positive integer $k$ such that $k$-SAT needs $\Omega(2^{(1-\varepsilon)n})$ time.}
In a survey on fine-grained complexity of approximation problems, Rubinstein and Vassilevska-Williams \cite{RubinsteinW19} asked three questions on the fine-grained complexity of approximating Diameter, and we address two of them in this work. First,
\begin{question}[Open Question 2.2 of \cite{RubinsteinW19}]
  \label{q:1}
  Is the simple near-linear time 2-approximation of Diameter optimal, or do there exist near-linear time algorithms giving better than a 2-approximation?
\end{question}
In directed graphs, as stated above, \cite{ChechikLRSTW14} gave a 3/2-approximation in $\tilde O(m^{3/2})$, but for running time $O(m^{3/2-\delta})$, the basic 2-approximation algorithm is the best known algorithm, as the algorithm of Cairo, Grossi, and Rizzi \cite{CairoGR16} crucially uses that the graph is undirected. 
Hence, Rubinstein and Vassilevska-Williams also ask,
\begin{question}[Open Question 2.3 of \cite{RubinsteinW19}]
  \label{q:2}
  Is there an $O(m^{3/2-\delta})$ time $(2-\varepsilon)$-approximation algorithm (for $\delta,\varepsilon>0$) for Diameter in directed graphs, or can we show that a $(2-\varepsilon)$-approximation for the diameter in sparse directed graphs needs $n^{3/2-o(1)}$ time?
\end{question}
Our contributions include answering Question~\ref{q:1} for the case of directed graphs (Corollary~\ref{cor:2eps}), conditionally answering Question~\ref{q:2} for SETH-hardness results (Corollary~\ref{cor:nseth-weight-intro}), and conditionally resolving the SETH-hardness of Diameter in directed unweighted graphs (Corollary~\ref{cor:nseth-intro}).
We highlight our results in the context of existing work in Figure~\ref{fig:res} (directed graphs) and Figure~\ref{fig:res-undir} (undirected graphs).

\paragraph{SETH-hardness results.}
Previously, the best fined grained lower bounds for near-linear time algorithms for Diameter in directed graphs ruled out $5/3-\varepsilon$-approximation algorithms in unweighted graphs \cite{BackursRSWW18} and $7/4-\varepsilon$-approximation algorithms in weighted graphs \cite{Bonnet20} under SETH. We improve both of these results to $2-\varepsilon$, which is optimal.
\begin{theorem}
  \label{thm:2eps-intro}
  \label{thm:2eps}
  Let $k\ge 3$ be an integer and $\varepsilon>0$.
  Assuming SETH, a $2-\frac{1}{k}-\varepsilon$ approximation of the Diameter of a directed graph, weighted or unweighted, needs time $n^{1 + 1/(k-1)-o(1)}$.
\end{theorem}
In particular, assuming SETH, the simple 2-approximation algorithm is optimal among near-linear time Diameter algorithms for directed graphs, answering Question~\ref{q:1} for the case of directed graphs.
\begin{corollary}
  \label{cor:2eps}
  Assuming SETH, for all $\varepsilon>0$, there exists a $\delta>0$ such that a $2-\varepsilon$ approximation of Diameter of a directed graph, weighted or unweighted, needs $n^{1+\delta-o(1)}$ time.
\end{corollary}
Our result improves over the result in \cite{BackursRSWW18} by improving the approximation ratio, and improves over \cite{Bonnet20} not only in the approximation ratio, but also in that our result rules out algorithms even in unweighted graphs. 

Not only does Corollary~\ref{cor:2eps} achieve the best possible approximation ratio of $2-\varepsilon$ for near-linear time algorithms, it also obtains a good dependence of the runtime $n^{1+\delta-o(1)}$ on the approximation ratio $2-\varepsilon$ for approximation ratios less than 2.
Later (Theorem~\ref{thm:nseth-intro}), we in fact show that this tradeoff is conditionally optimal for directed unweighted graphs.

Question~\ref{q:1} remains open for undirected graphs.
That is, it is open whether the simple 2-approximation algorithm of the Diameter in undirected graphs is optimal for near-linear time algorithms.
For undirected graphs, the best fine-grained lower bounds for near-linear time algorithms rule out a $5/3-\varepsilon$-approximation algorithm \cite{BackursRSWW18}.
However, this lower bound only holds for weighted graphs, and the best lower bound for unweighted graphs only rules out $8/5-\varepsilon$-approximation algorithms \cite{BackursRSWW18}.
For undirected unweighted graphs we improve this lower bound from $8/5-\varepsilon$ to $5/3-\varepsilon$, and we make the improvement with (we believe) a simpler reduction (see Theorem~\ref{thm:53}).

\paragraph{Concurrent work by Dalirrooyfard and Wein.}
In concurrent and independent work, Dalirrooyfard and Wein \cite{DalirrooyfardW2020} obtained Theorem~\ref{thm:2eps} with a different proof.

\begin{figure}
  \begin{center}
    \textbf{{Diameter in sparse directed unweighted graphs}}

    \vspace{0.1in}
    
    \begin{minipage}{0.46\linewidth}
      \hypersetup{citecolor=Black,linkcolor=Black}        
      \directedunweightedall
    \end{minipage}
    \,\,
    \begin{minipage}{0.5\linewidth}
      \directedunweightedtable
    \end{minipage}

    \vspace{0.1in}

    \textbf{{Diameter in sparse directed weighted graphs}}

    \vspace{0.1in}

    \begin{minipage}{0.46\linewidth}
      \hypersetup{citecolor=Black,linkcolor=Black}        
      \directedweightedall
    \end{minipage}
    \,\,
    \begin{minipage}{0.5\linewidth}
      \directedweightedtable
    \end{minipage}
  \end{center}
  \caption{Prior and new upper bounds, lower bounds, and non-reducibility results for Diameter in sparse directed graphs. 
Blue regions are feasible, red are infeasible assuming SETH.
Different shades of blue or red denote different results.
In cyan dotted regions, no deterministic SETH reductions can rule out the existence of algorithms, assuming NSETH, and no randomized SETH reductions can rule out the existence of algorithms assuming NUNSETH.
 }
  \label{fig:res}
\end{figure}

\begin{figure}
  \begin{center}
    \textbf{{Diameter in sparse undirected unweighted graphs}}

    \vspace{0.1in}
    
    \begin{minipage}{0.46\linewidth}
      \hypersetup{citecolor=Black,linkcolor=Black}        
      \undirectedunweightedall
    \end{minipage}
    \,\,
    \begin{minipage}{0.5\linewidth}
      \undirectedunweightedtable
    \end{minipage}

    \vspace{0.1in}

    \textbf{{Diameter in sparse undirected weighted graphs}}

    \vspace{0.1in}

    \begin{minipage}{0.46\linewidth}
      \hypersetup{citecolor=Black,linkcolor=Black}        
      \undirectedweightedall
    \end{minipage}
    \,\,
    \begin{minipage}{0.5\linewidth}
      \undirectedweightedtable
    \end{minipage}
  \end{center}
  \caption{Prior and new upper bounds, lower bounds, and non-reducibility results for Diameter in sparse undirected graphs. 
Blue regions are feasible, red are infeasible assuming SETH.
Different shades of blue or red denote different results.
In cyan dotted regions, no deterministic SETH reductions can rule out the existence of algorithms, assuming NSETH, and no randomized SETH reductions can rule out the existence of algorithms assuming NUNSETH.
 }
  \label{fig:res-undir}
\end{figure}

\paragraph{Non-reducibility results.}
For intermediate runtimes $O(n^{1+\delta})$ for $\delta\in(0,1/2)$, for both directed and undirected graphs, and for both weighted and unweighted graphs, there are gaps between the best approximation ratios achieved by algorithms and the best approximation ratios ruled out by hardness results, even in light of Theorem~\ref{thm:2eps}.
Therefore, it is natural to wonder whether stronger hardness of approximation results exist.
In this section, we conditionally rule out better SETH-based hardness results with \emph{non-reducibility results} in some settings.
To prove such non-reducibility results, we utilize a framework of Carmosi, Gao, Impagliazzo, Mikhailin, Paturi, Schneider \cite{CarmosinoGIMPS16}, which originally was used to conditionally prove that there is no deterministic reduction from SETH to 3-SUM and All-Pairs-Shortest-Path.
Assuming the Nondeterministic Strong Exponential Time Hypothesis (NSETH)\footnote{NSETH states: for every $\varepsilon>0$, there exists a $k$ such that \textsc{$k$-TAUT} is not in $\ntime[2^{n(1-\varepsilon)}]$, where \textsc{$k$-TAUT} is the language of all $k$-DNF which are tautologies.} \cite{CarmosinoGIMPS16}, (1) we answer Question~\ref{q:2} for SETH-based hardness results and (2) we characterize which combinations of approximation ratios and runtimes can be ruled out by deterministic SETH-based hardness results for the diameter of directed unweighted graphs.

While there is some evidence towards the veracity of NSETH, namely that a disproof would yield nontrivial circuit lower bounds \cite{CarmosinoGIMPS16}, there is also some evidence against NSETH \cite{Williams16}, namely that a Merlin-Arthur SETH and Arthur-Merlin SETH are false in a strong way.
However, NSETH itself seems challenging to disprove \cite{CarmosinoGIMPS16, Williams18}, so we view our NSETH-based non-reducibility results as showing a natural barrier to proving better SETH hardness results.

These non-reducibility results are proved \emph{without} constructing algorithms.
In Question~\ref{q:2}, if the requested algorithm existed, then the question would be settled and there would be no hardness result.
However, despite not finding such an algorithm, we are still able to rule out SETH-hardness results.
Instead of finding an algorithm, we prove our non-reducibility results by constructing \emph{non-deterministic} algorithms.
This technique of using non-deterministic algorithms to obtain NSETH-based non-reducibility results was also done for the Gomory-Hu tree problem \cite{AKT20}.
The non-deterministic choices of our algorithm can be deterministically found in polynomial time, so our non-deterministic algorithms also provide potentially interesting \emph{certifying algorithms} for Diameter. \cite{ABMR11,MMNS11,Kun18}.

Given a problem $\Pi$ and a time complexity $T$, the pair $(\Pi,T)$ is \emph{SETH-hard} \cite{CarmosinoGIMPS16} if there is a deterministic Turing reduction from \texttt{CNFSAT} to $\Pi$ such that, if SETH holds, then $\Pi$ needs $T^{1-o(1)}$ time (see Section~\ref{sec:prelims} for a formal definition).
In this language, Theorem~\ref{thm:2eps} states that a $2-\frac{1}{k}-\varepsilon$ approximation of Diameter with time complexity $n^{1+1/(k-1)}$ is SETH-hard.
Our results prove that certain approximations of Diameter are not SETH hard under NSETH.
Under this definition of SETH-hard, proving that $(\Pi,T)$ is not SETH-hard does not rule out \emph{randomized reductions}, when the Turing reduction in the SETH-hard definition is a probabilistic oracle machine (again see Section~\ref{sec:prelims} for a formal definition).
However, if we assume another hypothesis called \emph{Non-Uniform NSETH (NUNSETH)} \cite{CarmosinoGIMPS16}\footnote{NUNSETH states that for all $\varepsilon>0$, there exists $k$ such that there are no nondeterministic circuit families of size $O(2^{n(1-\varepsilon)})$ recognizing the language $k$-\textsc{TAUT}.}, we obtain non-reducibility results even for randomized reductions.

In Question~\ref{q:2}, Rubinstein and Vassilevska-Williams ask if we can show that a $(2-\varepsilon)$-approximation for the Diameter of sparse directed graph needs $n^{1.5-o(1)}$ time.
We shows that the answer is \emph{no}, if NSETH is true and we restrict ourselves to deterministic SETH reductions, or if NUNSETH is true and we restricted ourselves to randomized SETH reductions.
\begin{theorem}
  Under NSETH (NUNSETH), all $\delta>0$ and $\varepsilon>0$, a $5/3+\varepsilon$ approximation of the Diameter on directed graphs, weighted or unweighted, with time complexity $m^{(1+\delta)19/13}$ is \emph{not} SETH-hard for deterministic (randomized) reductions.
\label{thm:nseth-weight-intro}
\end{theorem}
As $5/3 < 2$ and $19/13 < 3/2$, taking $\delta$ and $\varepsilon$ to be sufficiently small constants in Theorem~\ref{thm:nseth-weight-intro} immediately gives following corollary, which states our ``no'' answer to Question~\ref{q:2} for SETH-hardness results.
\begin{corollary}
\label{cor:nseth-weight-intro}
  Under NSETH (NUNSETH), there exist $\delta>0$ and $\varepsilon>0$ such that a $2-\varepsilon$ approximation of the Diameter on directed graphs, weighted or unweighted, with time complexity $m^{3/2-\delta}$ is \emph{not} SETH-hard for deterministic (randomized) reductions.
\end{corollary}
Towards resolving Question~\ref{q:2} fully, Corollary~\ref{cor:nseth-weight-intro} can be viewed in two possible ways: 
On one hand, it can be seen as evidence for an algorithm, that there exists a $\tilde O(m^{3/2-\delta})$-time $(2-\varepsilon)$-approximation of Diameter in directed graphs, answering Question~\ref{q:2} in full (see \cite{AKT20, AKT20b} for an example of a progression from non-reducibility/nondeterministic algorithms to an algorithm).
On the other hand, if one believes such an algorithm does not exist, Corollary~\ref{cor:nseth-weight-intro} can be seen as evidence that, unless NSETH is disproved, hypotheses other than SETH are necessary to conditionally rule out such algorithms.
To our knowledge, all existing fine-grained hardness results for Diameter are based on deterministic reductions from SETH, so Corollary~\ref{cor:nseth-weight-intro} indeed shows that, under NSETH, new ideas would be needed to prove such hardness results.

Using a variant of the technique for Theorem~\ref{thm:nseth-weight-intro}, we prove that the tradeoff between the runtime and approximation ratio in Theorem~\ref{thm:2eps-intro} is optimal among deterministic (randomized) SETH-reductions for directed unweighted graphs unless NSETH (NUNSETH) is disproved.
\begin{theorem}
\label{thm:nseth-intro}
  Let $k\ge 2$ be a positive integer and let $\delta$ and $\varepsilon$ be positive reals.
  Assuming NSETH (NUNSETH), for any $\delta>0$, a $2-\frac{1}{k}+\varepsilon$ approximation of Diameter on directed unweighted graphs with time complexity $m^{1+1/k+\delta}$ is \emph{not} SETH-hard for deterministic (randomized) reductions under NSETH.
\end{theorem}
Theorem~\ref{thm:nseth-intro} is quantitatively stronger than Theorem~\ref{thm:nseth-weight-intro} as Theorem~\ref{thm:nseth-intro} applies to a larger parameter setting of runtimes and approximation ratios, but Theorem~\ref{thm:nseth-intro} is weaker in that only rules out SETH-hardness results for unweighted graphs.
Theorem~\ref{thm:nseth-intro}  in fact also holds for a large class of weighted graphs (see Theorem~\ref{thm:nseth-intro-1} for the more general statement), but not \emph{arbitrary} weighted graphs as in Theorem~\ref{thm:nseth-weight-intro}.
Combining Theorem~\ref{thm:2eps-intro} with Theorem~\ref{thm:nseth-intro} yields the following \emph{complete characterization} of the SETH-hardness of Diameter in sparse directed unweighted graphs.
\begin{corollary}
\label{cor:nseth-intro}
  Assuming NSETH (NUNSETH), for any $\delta\in(0,1)$ and for essentially all $\varepsilon\in(0,1)$,\footnote{For every $\delta$, there is only one value of $\varepsilon$ for which the characterization does not apply, namely $1/\varepsilon = \floor{1/\delta}+1$.} a $2-\varepsilon$ approximation of Diameter of a sparse directed unweighted graph with time complexity $n^{1+\delta}$ is SETH-hard for deterministic (randomized) reductions if and only if $\floor{1/\varepsilon}\le \floor{1/\delta}$.
\end{corollary}
Corollary~\ref{cor:nseth-intro} settles, for (essentially) \emph{every} $\delta,\varepsilon\in(0,1)$, whether a $2-\varepsilon$ approximation of Diameter with time complexity $n^{1+\delta}$ is SETH-hard for deterministic reductions, under NSETH, and SETH-hard for randomized reductions, under NUNSETH.
We can view Theorem~\ref{thm:nseth-intro} as establishing a similar dichotomy to  Theorem~\ref{thm:nseth-weight-intro}: either there exist certain algorithms for directed unweighted diameter matching our SETH-hardness results (Theorem~\ref{thm:2eps}) or, unless NSETH is disproved, hypotheses other than SETH are needed to improve our hardness results for directed unweighted graphs.

We note that Theorem~\ref{thm:nseth-weight-intro} (but not Theorem~\ref{thm:nseth-intro}) depends on a construction of \emph{hopsets} in \cite{CaoFR20}, for which currently only an extended abstract is available (see Section~\ref{sec:prelims} for the construction that we use and Appendix~\ref{app:hopset} for an explanation of how the construction is implicit in \cite{CaoFR20}).
Faster and better constructions of hopsets in directed graphs would extend Theorem~\ref{thm:nseth-weight-intro} to a larger setting of $\delta$ and $\varepsilon$, possibly matching the parameter setting of Theorem~\ref{thm:nseth-intro} (see Theorem~\ref{thm:nseth-weight-conj} for the quantitative improvement implied by better hopset constructions).

The non-reducibility results of Theorem~\ref{thm:nseth-weight-intro} and Theorem~\ref{thm:nseth-intro} automatically apply to undirected graphs.
However, the guarantees implied for undirected unweighted graphs are better than the guarantees implied for undirected weighted graphs.
Using a construction of hopsets for undirected graphs, we prove non-reducibility results for in undirected weighted graphs matching those implied by Theorem~\ref{thm:nseth-intro} for undirected unweighted graphs (see Theorem~\ref{thm:nseth-undir} for the result).
Among other things, we show that, assuming NSETH, the lower bounds of \cite{BackursRSWW18} and Theorem~\ref{thm:53}, which show a $5/3-\varepsilon$ approximation needs $n^{3/2-o(1)}$ time in undirected graphs, are tight in the approximation ratio.

\subsection{Other related work}
\label{sec:related}

See \cite{Williams18} for a survey of fine-grained complexity and \cite{RubinsteinW19} for a survey of fine-grained complexity for approximation problems.

Diameter algorithms in dense graphs have also been studied.
In weighted dense graphs, like in sparse graphs, the fastest Diameter algorithms in fact compute the distances between every pair of vertices, solving the All-Pairs-Shortest-Paths (APSP) problem.
The fastest APSP algorithms take time $\tilde O(\min(mn, n^3/\exp(\Omega(\sqrt{n}))))$ \cite{Williams14, Pettie04, PettieR05}.
For unweighted dense graphs, Diameter can be solved \cite{CyganGS15} in time $\tilde O(\min(mn,n^\omega))$, where $\omega<2.373$ is the constant for matrix multiplication \cite{Williams12, Le14, Stothers10}.
In undirected, unweighted graphs, \cite{BackursRSWW18} gave a nearly-3/2 approximation in $\tilde O(n^2)$ expected time.

A number of generalizations and extensions of Diameter have been considered and studied, including eccentricities \cite{ChechikLRSTW14, BackursRSWW18}, ST-diameter~\cite{BackursRSWW18, DalirrooyfardW019a}, bichromatic diameter~\cite{DalirrooyfardW019a}, roundtrip diameter~\cite{RubinsteinW19}, and min-diameter \cite{AbboudWW16, DalirrooyfardW019}.

\subsection{Outline of the paper}
In Section~\ref{sec:techniques}, we outline the main ideas behind the proofs of our main results.
In Section~\ref{sec:prelims}, we state some technical preliminaries for the proofs of our main results.
In Section~\ref{sec:2eps}, we present the full proof of our SETH-hardness result, Theorem~\ref{thm:2eps}.
We then move on to the proofs of our non-reducibility results.
As the proofs of the non-reducibility results that use hopsets, Theorems~\ref{thm:nseth-weight-intro} and \ref{thm:nseth-undir}, are simpler, we prove Theorems~\ref{thm:nseth-weight-intro} and \ref{thm:nseth-undir} first in Section~\ref{sec:nseth-weight}.
In Section~\ref{sec:nseth}, we prove Theorem~\ref{thm:nseth-intro}.
In Section~\ref{sec:53}, we prove Theorem~\ref{thm:53}.

We leave the proofs of lemmas stated in the preliminaries, Section~\ref{sec:prelims}, to the appendix, as they are technical and not the focus of this paper.
In Appendix~\ref{app:hopset}, we justify our use of the hopset constructions, first for undirected graphs, and then for directed graphs, where we use the construction in \cite{CaoFR20}.
In Appendix~\ref{app:fgc}, we prove some results on fine-grained reductions.

\section{Techniques}
\label{sec:techniques}

In this section, we sketch the ideas behind our SETH-hardness result (Theorem~\ref{thm:2eps}) and our non-reducibility results (Theorem~\ref{thm:nseth-intro} and Theorem~\ref{thm:nseth-weight-intro}).

\subsection{SETH-hardness for $k=5$}
First, we highlight the ideas behind our improved lower bound, Theorem~\ref{thm:2eps}, which shows that a $2-\frac{1}{k}-\varepsilon$ approximation of diameter in directed graphs needs $n^{1+1/(k-1)-o(1)}$ time.
To do so, we sketch the ideas of the proof when $k=5$.

Our construction (for $k=5$) is based on a clever lower bound construction for a problem called ST-Diameter, given by Backurs, Roditty, Segal, Vassilevska-Williams, and Wein \cite{BackursRSWW18}.
This lower bound construction reduces from the 5-Orthogonal-Vectors (5-OV) problem, and use this construction for ST-Diameter to prove that there are no near-linear time (in fact no $O(n^{3/2-\delta})$ time) $5/3-\varepsilon$ approximations for diameter in directed graphs.
An important idea in our construction is to instead reduce from a variant of 5-OV called \emph{Single-Set 5-Orthogonal-Vectors}.

5-OV takes as input $5$ sets each of $\tilde{n}$ vectors $A,B,C,D,E\subset\{0,1\}^{r}$ of dimension $r=\Theta(\log \tilde{n})$, and outputs whether there exists $k$ vectors $a\in A,b\in B,c\in C,d\in D,e\in E$ that are orthogonal, i.e., $a[x]\cdots e[x]=0$ for all coordinates $x\in[r]$. 
Single Set 5-OV is the same problem, but additionally assumes $A=B=C=D=E$.
Assuming SETH, both $5$-OV and Single-Set 5-OV need $\tilde{n}^{5-o(1)}$ time \cite{Williams05}.

From a Single-Set 5-OV instance $\Phi$ of size $\tilde{n}$, we construct a graph $G$ of size $n\defeq \tilde O(\tilde{n}^4)$ time such that, if $\Phi$ has a solution, the diameter of $G$ is at least 9, and otherwise the diameter of $G$ is at most 5.
This shows that distinguishing between graphs of diameter 9 and diameter 5 needs $\tilde{n}^{5-o(1)} = n^{5/4-o(1)}$ time assuming SETH, so a $9/5-\varepsilon$-approximation of diameter also needs $n^{5/4-o(1)}$ time assuming SETH.

To construct the graph $G$, we start with the following construction of a graph $G_{ST}$, which is exactly the construction in \cite{BackursRSWW18} for ST-Diameter.
Let the Single-Set 5-OV instance $\Phi$ have vector sets $A\subset\{0,1\}^r$.
In \cite{BackursRSWW18}, the vertex set of $G_{ST}$ has six vertex subsets $L_0,\dots,L_5$, with carefully chosen edges between $L_i$ and $L_{i+1}$ for $i=0,\dots,4$. 
Sets $L_0$ and $L_5$ have one vertex for each element of $A^4$, respectively.
Sets $L_1,\dots,L_4$ have (at most) one vertex for each element of $A^3\times [r]^4$.
The graph $G_{ST}$ has the property that, if the 5-OV instance $\Phi$ has a solution, there exists $u\in L_0$ and $v\in L_5$ such that $d(u,v)\ge 13$, and otherwise $d(u,v)\le 5$ for all $u\in L_0$ and $v\in L_5$.
The basic setup of this construction is illustrated in Figure~\ref{fig:2eps-sketch-1}.
\begin{figure}
\begin{center}
\begin{tikzpicture}[scale=0.5, decoration={markings,mark=at position 0.5 with{\arrow{>}}}]
  \node[inner sep=20, outer sep=0,draw,label=below:\tiny{$\tilde{n}^{4}$ vtxs}] (0) at (0,0)  {$L_0$};
  \node[inner sep=20, outer sep=0,draw,label=below:\tiny{$\tilde{n}^{4}$ vtxs}] (5) at (27,0) {$L_5$};
  \node[inner sep=10, outer sep=0,draw,label=below:\tiny{$\tilde{n}^{3+o(1)}$ vtxs}] (1) at (6,0)  {$L_1$};
  \node[inner sep=10, outer sep=0,draw,label=below:\tiny{$\tilde{n}^{3+o(1)}$ vtxs}] (2) at (11,0) {$L_2$};
  \node[inner sep=10, outer sep=0,draw,label=below:\tiny{$\tilde{n}^{3+o(1)}$ vtxs}] (3) at (16,0) {$L_3$};
  \node[inner sep=10, outer sep=0,draw,label=below:\tiny{$\tilde{n}^{3+o(1)}$ vtxs}] (4) at (21,0) {$L_4$};
  \foreach \s/\t in {0/1,1/2,2/3,3/4,4/5} {
    \draw[] (\s) -- (\t);
    \draw[transform canvas={yshift=1mm}] (\s) -- (\t);
    \draw[transform canvas={yshift=-1mm}] (\s) -- (\t);
    \draw[transform canvas={yshift=2mm}] (\s) --  (\t) node[midway,above]{\tiny{$\stackrel{\tilde{n}^{4+o(1)}}{\text{edges}}$}};
    \draw[transform canvas={yshift=-2mm}] (\s) -- (\t);
  } 
\end{tikzpicture}
\end{center}
\caption{$G_{ST}$: the ST-Diameter construction from \cite{BackursRSWW18} when $k=5$}
\label{fig:2eps-sketch-1}
\end{figure}
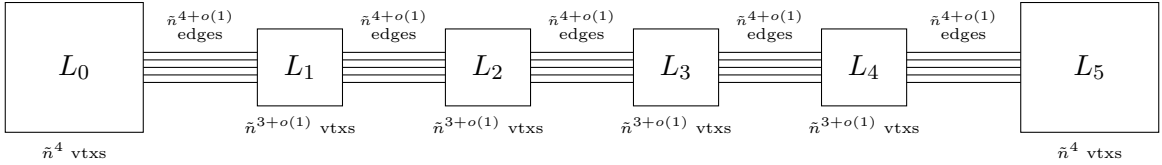

This construction does not immediately yield a hardness construction for Diameter.
Indeed, when the 5-OV instance $\Phi$ has no solution, even though the so-called ST-Diameter between $L_0$ and $L_5$ is 5, the Diameter is larger. The distance from a vertex in $L_0$ to a vertex in $L_1,\dots,L_4$ can be larger than 5, and the distance between any two vertices in $L_1,\dots,L_4$ can also be larger than 5.
In \cite{BackursRSWW18}, they turn this construction into a hardness construction for Diameter by adding more edges and vertices, but only show a $5/3-\varepsilon$ hardness of approximation for near-linear time directed Diameter algorithms.

Our directed Diameter lower bound builds on this construction in a different way.
The key step, illustrated in Figure~\ref{fig:2eps-sketch-2} is to (1) add edges from each of $L_1,L_2,L_3,L_4$ to $L_1$ and (2) add edges from $L_4$ to each of $L_1,L_2,L_3,L_4$, giving a graph $G'$.
These edges crucially use the fact that we are reducing from \emph{Single-Set} orthogonal vectors.
These edges have the property that, for any vertex $(a,b,c,d)\in L_1$ and any coordinate-tuple $\vec i\in[r]^4$, the out-neighborhood of $(a,b,c,d)$ is a \emph{subset} of the out neighborhood of $(a,b,c,\vec i)\in L_1\cup\cdots\cup L_4$.
Thus, when $\Phi$ has no solution, any vertex that vertex $(a,b,c,d)$ can reach in 5 steps can also be reached in 5 steps by any vertex $(a,b,c,\vec i)\in L_1\cup\cdots\cup L_4$. 
Since every vertex of $L_0$ is distance 5 from every vertex of $L_5$ in graph $G_{ST}$, every vertex in $L_0\cup\cdots\cup L_4$ is distance 5 from every vertex in $L_5$ in the new graph $G'$.
By a similar argument, any vertex that can reach $(b,c,d,e)\in L_5$ in 5 steps can also reach any $(c,d,e,\vec i)\in L_1\cup\cdots\cup L_4$ in 5 steps, so every vertex in $L_0\cup\cdots\cup L_4$ is distance 5 from every vertex in $L_1\cup\cdots\cup L_5$.

It remains to ensure that, when $\Phi$ has no solution, any two vertices in $L_0$ are at distance 5, and any two vertices in $L_5$ are at distance 5.
Because vertices in $L_0$ and $L_5$ are each identified by elements of $A^4$, there is a natural bijection between vertices in $L_0$ and vertices in $L_5$.
In our final construction, we (1) direct the edges of $G'$ between $L_i$ and $L_{i+1}$ to point towards the vertex in $L_{i+1}$, and (2) contract the pairs of vertices in $L_0$ and $L_5$ that correspond to the same 4-tuple $A^4$, giving our final Diameter instance $G$.
The contraction step also uses that we are reducing from Single-Set orthogonal vectors and not ordinary orthogonal vectors.
We already showed that every vertex in $L_0\cup\cdots\cup L_4$ is distance 5 from every vertex in $L_1\cup\cdots\cup L_5$, but now that $L_0=L_5$ (and directing the edges did not disrupt any paths), we have that the Diameter of this final construction $G$ is at most 5 when $\Phi$ has no solution.

With some care, we can also show that the Diameter of $G$ is at least 9 when $\Phi$ has a solution, giving our desired reduction. In particular, we show that, if $a,b,c,d,e$ are orthogonal vectors in $\Phi$, the distance from vertex $(a,b,c,d)\in L_0$ to vertex $(b,c,d,e)\in L_0$ is at least 9.

\begin{figure}
\begin{center}
\begin{tikzpicture}[scale=0.5, decoration={markings,mark=at position 0.5 with{\arrow{>}}}]
  \node[inner sep=20, outer sep=0,draw] (0) at (0,0) {$L_0$};
  \node[inner sep=20, outer sep=0,draw] (5) at (27,0) {$L_5$};
  \node[inner sep=10, outer sep=0,draw] (1) at (6,0) {$L_1$};
  \node[inner sep=10, outer sep=0,draw] (2) at (11,0) {$L_2$};
  \node[inner sep=10, outer sep=0,draw] (3) at (16,0) {$L_3$};
  \node[inner sep=10, outer sep=0,draw] (4) at (21,0) {$L_4$};
  \foreach \s/\t in {0/1,1/2,2/3,3/4,4/5} {
    \draw[] (\s) -- (\t);
    \draw[transform canvas={yshift=1mm}] (\s) -- (\t);
    \draw[transform canvas={yshift=-1mm}] (\s) -- (\t);
    \draw[transform canvas={yshift=2mm}] (\s) -- (\t);
    \draw[transform canvas={yshift=-2mm}] (\s) -- (\t);
  }
  \foreach \s/\t in {2/1,3/1,4/1} {
    \draw[postaction={decorate}, red] (\s) to [bend right=40](\t);
    \draw[postaction={decorate}, red] (\s) to [bend right=35](\t);
    \draw[postaction={decorate}, red] (\s) to [bend right=45](\t);
  }
  \foreach \s/\t in {4/1,4/2,4/3} {
    \draw[postaction={decorate}, red] (\s) to [bend left=40](\t);
    \draw[postaction={decorate}, red] (\s) to [bend left=35](\t);
    \draw[postaction={decorate}, red] (\s) to [bend left=45](\t);
  }
  \path (1) edge [postaction={decorate},red,loop above, in=100,out=80, looseness=5] node {} (1);
  \path (1) edge [postaction={decorate},red,loop above, in=110,out=70, looseness=5] node {} (1);
  \path (1) edge [postaction={decorate},red,loop above, in=120,out=60, looseness=5] node {} (1);
  \path (4) edge [postaction={decorate},red,loop below, in= -100,out= -80, looseness=5] node {} (4);
  \path (4) edge [postaction={decorate},red,loop below, in= -110,out= -70, looseness=5] node {} (4);
  \path (4) edge [postaction={decorate},red,loop below, in= -120,out= -60, looseness=5] node {} (4);
\end{tikzpicture}
\end{center}
\caption{$G'$: almost our directed Diameter construction for $k=5$}
\label{fig:2eps-sketch-2}
\end{figure}
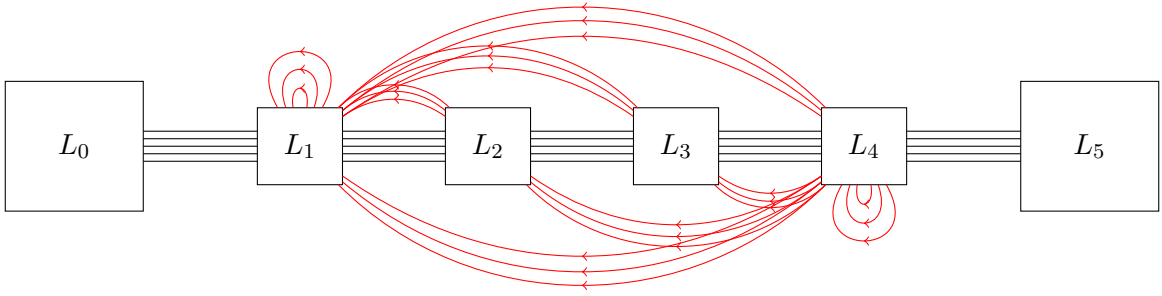

\subsection{Non-reducibility for undirected unweighted 7/4-Diameter}
In this section, we highlight the key ideas of Theorem~\ref{thm:nseth-weight-intro} and Theorem \ref{thm:nseth-intro}, which show that under (NU)NSETH, a $\alpha$ approximation of the Diameter with time complexity $n^{1+\delta}$ is not SETH-hard, for certain values of $\alpha$ and $\delta$.
Define the $D'/D$-Diameter problem as the (promise) problem whose input is a graph, and such that an algorithm must always accept when the graph has diameter at least $D'$, always reject when the graph has diameter at most $D$, and can accept or reject otherwise.
We show (Lemma~\ref{lem:gap-to-approx}) that, to prove for some $\alpha\ge 1$ that an $\alpha$-approximation of Diameter with some time complexity $T$ is not SETH-hard, it suffices to prove that $\alpha D/D$-Diameter with some time complexity $T$ is not SETH-hard for all $D$. 
To demonstrate how we might prove that some $\alpha D/D$-Diameter is not SETH-hard, we show how to rule out SETH-hardness for distinguishing between undirected unweighted graphs  of diameter 7 and 4.
\begin{proposition}
  \label{prop:74-1}
  Under NSETH (NUNSETH), for any $\delta>0$, 7/4-Diameter on undirected unweighted graphs with time complexity $m^{4/3+\delta}$ is \emph{not} SETH-hard for deterministic (randomized) reductions.
\end{proposition}
Following the framework of \cite{CarmosinoGIMPS16}, to prove Proposition~\ref{prop:74-1}, it (almost, modulo a small technicality described after Lemma~\ref{lem:nseth-old}) suffices to prove the following proposition, that 7/4-Diameter on directed unweighted graphs is in $\ncon[m^{4/3+o(1)}]$.
\begin{proposition}
  \label{prop:74-2}
  7/4-Diameter on undirected unweighted graphs is in $\ncon[m^{4/3+o(1)}]$.
\end{proposition}
\begin{proof}
By the definition of $\ncon$, it suffices to (1) verify that a graph of diameter at least 7 indeed has diameter greater than 4 and (2) verify that a graph of diameter at most 4 indeed has diameter less than 7, each in time $m^{4/3+o(1)}$.
In a graph of diameter 7, we can easily verify the graph has diameter at least 7, and thus greater than 4, by nondeterministically selecting a vertex and checking with BFS in $O(m)$ time if the eccentricity is at least $7$.

The hard step is verifying that the diameter is less than 7 when the diameter is at most 4. 
The key idea is to observe that in a graph of diameter at most 4, one of the following exists.
\begin{enumerate}
\item A set $X$ of $\tilde O(m^{2/3})$ vertices such that every vertex is distance 1 from some vertex in $X$.
\item A set $Y$ of $\tilde O(m^{1/3})$ vertices such that every vertex is distance 2 from some vertex in $Y$.
\end{enumerate}
To prove this, if any vertex has $|N_2(v)|\le m^{1/3}$ (here, $N_r(v)$ is the set of all vertices at distance at most $r$ from $v$), we can take $Y=N_2(v)$.
Since each vertex is distance at most 4 from $v$, each vertex is distance at most 2 from some vertex in $Y$, as desired.
Thus, if $Y$ does not exist, we must have $|N_2(v)|\ge m^{1/3}$ for all vertices $v$.
In particular, at least $m^{1/3}$ edges are incident to $N_1(v)$.
If we take $X$ to be the vertices incident to $\tilde O(m^{2/3})$ uniformly random edges, we end up with a vertex in each $N_1(v)$ with high probability, as desired.

From these two possibilities, we can certify that the diameter is at most 6 by first nondeterministically picking a set $X$ of $\tilde O(m^{2/3})$ vertices and also nondeterministically picking a length 4 path between any two vertices in $X$, and then nondeterministically picking a set $Y$ of $m^{1/3}$ vertices.
We then run multi-source shortest path from $X$ and single-source shortest path from each element of $Y$.
We accept if each vertex has distance 1 to some vertex of $X$ and if all length-4 paths between pairs of elements in $X$ are valid.
We also accept if $Y$ has distance 2 to the rest of the graph and every element of $Y$ has eccentricity 4.
We reject otherwise.
This all takes time $\tilde O(m^{4/3})$.
Crucially, checking the distances in $X$ only takes time $O(|X|^2)$ rather than $O(|X|\cdot m)$, because we guess the paths between vertices in $X$.
If the graph is diameter 4, by above, there either exists $X$ such that we accept or there exists $Y$ such that we accept.
On the other hand, one can check that, if we accept, the diameter of the graph is at most 6, as desired.
\end{proof}

In the full proof of Theorems~\ref{thm:nseth-weight-intro}, \ref{thm:nseth-intro}, and \ref{thm:nseth-undir}, we need to generalize Proposition~\ref{prop:74-2} to (1) arbitrary diameters, (2) directed graphs, (3) weighted graphs, and/or (4) other runtimes and approximation ratios.
The first step is the most challenging, and we describe the challenge and our solution here.
If, for example, the graph is unweighted and $D$ and $D'$ are small then generalizing Proposition~\ref{prop:74-2} is straightforward.
The problem comes when $D$ and $D'$ are large, such as $m^{\Theta(1)}$, in which case the runtime of ``guessing the path between pairs in $X$'' could take too long if there are too many vertices on the paths.
One solution is to use a \emph{hopset}, a useful object originally defined by Cohen~\cite{Cohen00} in the context of parallel algorithms.
Formally, a $(\beta,\varepsilon)$-hopset for a graph is a set of edges such that, when added to the graph, any two vertices have a $\beta$-edge path between them whose length is within a multiplicative $(1+\varepsilon)$ factor of the true shortest path.
With a hopset, we only need to guess paths of at most $\beta$ vertices, saving on the runtime if $\beta$ is sufficiently small and the hopset is constructed sufficiently quickly, and losing up to an $\varepsilon D$ additive factor in the estimated distances, which we can afford.
This works for Theorem~\ref{thm:nseth-weight-intro} and Theorem~\ref{thm:nseth-undir}.

However, in directed graphs, the state-of-the-art constructions of hopsets are either too slow and/or have too large of a hop-bound $\beta$ to be useful for Theorem~\ref{thm:nseth-intro}.
Instead, we leverage the following idea.
Suppose we have a directed unweighted graph with a large diameter $D=m^{\Theta(1)}$. 
In the generalization of Proposition~\ref{prop:74-2}, we end up taking some set $Y$ formed by the $\tilde D$-neighborhood of some vertex $v$ for some $\tilde D$ so that $Y$ is distance at most $D-\tilde D$ from every vertex.
Then, crucially, for every other vertex $u$, by considering the path from $u$ to $v$, there are at least $\varepsilon D$ vertices in $Y$ of distance at most $D-\tilde D+\varepsilon D$ from $u$ (this step uses that the graph is unweighted).
Thus, we can replace $Y$ by a random subset $Y'$ of size $\tilde O(|Y|/(\varepsilon D))$, saving a factor of $D$ in the runtime, which we need, and losing only an additive $\varepsilon D$ in the estimated distances, which we can afford.

\section{Preliminaries}
\label{sec:prelims}

All logs are base $e$ unless otherwise specified.
When $a_1,a_2,\dots,$ is a sequence, the subsequence $a_i,\dots,a_j$ with $i > j$ is taken to be an empty sequence.
For an integer $a$, let $[a]\defeq \{1,\dots,a\}$.
In graphs, $n$ always denotes the number of vertices, $m$ always denotes the number of edges, and $W_{max}$ ($W_{min}$) always denotes the maximum (minimum) edge-weight.
In a graph $G$, let $d_G(u,v)$ be the length of the shortest path from $u$ to $v$. We omit the subscript $G$ when it is clear from the context.
We assume all input graphs to Diameter problems are connected/strongly connected, so that $d(u,v)$ always exist. This can be verified in near-linear time and the diameter is not defined otherwise. 

\paragraph{Graph notation.}
In an undirected graph, the eccentricity of a vertex $v$ is defined as $\epsilon(v)=\max_{u\in V}d(u,v)$.
In a directed graph, the (out-)in-eccentricity of a vertex $v$ is defined as $\epsilon^{in}(v)=\max_{u\in V}d(u,v)$ ($\epsilon^{out}(v)=\max_{u\in V}d(v,u)$).
In an undirected graph, let $N_{\le r}(v)=\{u:d(v,u)\le r\}$ denote the set of vertices of distance at most $r$ from/to $v$.
In a directed graph, let $N_{\le r}^{in}(v)=\{u:d(u,v)\le r\}$ denote the set of vertices of distance at most $r$ to $v$.
In a directed graph, let $N_{\le r}^{out}(v)=\{u,d(v,u)\le r\}$ denote the set of vertices of distance at most $r$ from $v$.
Let $N^{in}(v)$ denote the in-neighbors of a vertex $v$ and $N^{out}(v)$ denote out-neighbors of a vertex $v$.
Let $E_{\le r}(v)$ denote the set of edges incident to some vertex of $N_{\le r}(v)$.
Let $E_{\le r}^{in}(v)$ denote the set of edges $(a,b)$ (from $a$ to $b$) with $b\in N_{\le r}^{in}(v)$.
Let $E_{\le r}^{out}(v)$ denote the set of edges $(a,b)$ with $a\in N_{\le r}^{out}(v)$.
For a vertex $v$ in a graph and a subset $X$ of vertices, let $d(v,X)\defeq \min_{x\in X} d(v,x)$. Similarly, let $d(X,v)\defeq \min_{x\in X}d(x,v)$ (which may be different if the graph is directed.)

\paragraph{SETH.} The Strong Exponential Time Hypothesis (SETH) \cite{ImpagliazzoPZ01} states that, for every $\varepsilon>0$, there exists a $k$ such that $k$-SAT on $n$ variables cannot be solved in $O(2^{(1-\varepsilon)n})$ time (say, on a word-RAM with $O(\log n)$-bit words).
For $k\ge 2$, the \emph{$k$-Orthogonal Vectors} ($k$-OV) problem asks, given sets $A_1,\dots,A_k\subset\{0,1\}^d$ of $n$ vectors each, determine whether there exists vectors $a_i\in A_i$ such that $\prod_{i=1}^{k} a_i[x]=0$ for all $x\in[d]$. 
The \emph{Single-Set $k$-Orthogonal Vectors} is the $k$-OV problem when $A_1=A_2=\cdots=A_k$.
Williams~\cite{Williams05} showed that, assuming SETH, for a large enough constant $c$, $k$-OV and Single-Set $k$-OV need $n^{k-o(1)}$ time when $d=c\log n$.

\paragraph{Promise problem.}
  A \emph{promise problem} $\Pi$ is a pair of non-intersecting sets, denoted $(\Pi_{YES},\Pi_{NO})$, with $\Pi_{YES},\Pi_{NO}\subset\{0,1\}^*$.
  The set $\Pi_{YES}\cup \Pi_{NO}$ is called the \emph{promise}.
  A Turing reduction from a problem $\Pi_1$ to a promise problem $\Pi_2$ is an oracle Turing machine $\mathcal{M}^{\Pi_2}$ such that the output of $\mathcal{M}^{\Pi_2}$ is always correct whenever the outputs of $\Pi_2$ are ``YES'' when the input is in $\Pi_{YES}$ and are ``NO'' when the input is in $\Pi_{NO}$ and are arbitrarily ``YES'' or ``NO'' when the input is not in the promise.

\paragraph{Fine grained reductions.}
Let $\Pi_1,\Pi_2$ be a decision problem and $T_1,T_2$ be time bounds. We say that $(\Pi_1,T_1)$ \emph{fine-grained reduces} to $(\Pi_2,T_2)$, denoted $(\Pi_1,T_1)\le_{FGR}(\Pi_2,T_2)$ if
\begin{enumerate}
\item  For all $\varepsilon>0$, there exists a $\delta>0$ and a deterministic Turing reduction $M^{\Pi_2}$ from $\Pi_1$ to $\Pi_2$, such that $\TIME[\mathcal{M}]\le T_1^{1-\delta}$ and
\item Let $\tilde Q(\mathcal{M},x)$ denote the set of queries made by $\mathcal{M}$ to the oracle on a input $x$ of length $n$. Then the query lengths obey the following time bound:
\begin{align}
  \sum_{q\in\tilde Q(\mathcal{M},x)}^{} (T_2(|q|))^{1-\varepsilon}\le (T_1(n))^{1-\delta}.
\end{align}
\end{enumerate}
A problem $\Pi$ with time complexity $T$ is \emph{SETH-hard} if there is a fine-grained reduction from \textsc{CNFSAT} with time complexity $2^n$ to $(\Pi,T)$.
We note that one can make the same definition for promise problems and function problems, with one exception for function problems.
If $(\Pi_2,T_2)$ is a function problem, then we also bound the sizes of the answers given by the $\Pi_2$-oracle.
\begin{align}
  \sum_{q\in\tilde Q(\mathcal{M},x)}^{} (|\Pi_2(q)|)^{1-\varepsilon}\le (T_1(n))^{1-\delta}.
\end{align}
This definition holds even if a function has multiple acceptable outputs, as is the case for approximation problems where the output is acceptable if it is within an $\alpha$-factor of the optimal.

Randomized fine-grained reductions, denoted $(\Pi_1,T_1)\le_{FGR,r}(\Pi_2,T_2)$ are defined exactly as deterministic fine-grained reductions, except that the Turing reduction from $(\Pi_1,T_1)$ to $(\Pi_2,T_2)$ is a probabilistic machine with some two-sided error bound
\begin{align}
  \Pr[\mathcal{M}^{\Pi_2}(x)\in\Pi_1(x)]\ge \frac{2}{3},
\end{align}
where $\Pi_1(x)$ is the set of acceptable outputs for $\Pi_1$ on input $x$.
Now, a problem $\Pi$ with time complexity $T$ is \emph{SETH-hard with a randomized reduction} if there is a randomized fine-grained reduction from \textsc{CNFSAT} with time complexity $2^n$ to $(\Pi,T)$.

\paragraph{NSETH and non-reducibility.}
The Nondeterministic Strong Exponential Time Hypothesis (NSETH) \cite{CarmosinoGIMPS16} states that for every $\varepsilon>0$, there exists a $k$ such that \textsc{$k$-TAUT} is not in $\ntime[2^{n(1-\varepsilon)}]$, where \textsc{$k$-TAUT} is the language of all $k$-DNF which are tautologies. 
The following key result allowed \cite{CarmosinoGIMPS16} to prove non-reducibility results.
\begin{lemma}[{\cite[Theorem 5.1]{CarmosinoGIMPS16}}]
  \label{lem:cgimps-old}
  \label{lem:nseth-old}
  If NSETH holds and $\Pi\in \ncon[T]$ for some decision or function problem $\Pi$, then $(\Pi,T^{1+\delta})$ is not SETH-hard (with a deterministic reduction) for any $\delta>0$.
\end{lemma}
However, in \cite{CarmosinoGIMPS16}, promise problems are not considered, and we were not able to reprove Lemma~\ref{lem:cgimps-old} for promise problems. However, we prove a slightly weaker result, which is still sufficient for our non-reducibility results.
To state this result, we need the following definition.
\begin{definition}
\label{def:nconplus}
Let $\nconplus[T]$ be the complexity class containing the promise problems $\Pi=(\Pi_{YES},\Pi_{NO})$ such that there exists nondeterministic Turing machines $\mathcal{A}_N$ and $\mathcal{A}_{coN}$ each running in time $T$ such that:
\begin{enumerate}
\item For all inputs $x\in\Pi_{YES}$, there exist nondeterministic choices such that $\mathcal{A}_N$ outputs ``YES'', and for all inputs $x\in \Pi_{NO}$, all nondeterministic choices cause $\mathcal{A}_N$ to output ``NO''.
\item For all inputs $x\in\Pi_{NO}$, there exist nondeterministic choices such that $\mathcal{A}_{coN}$ outputs ``NO'', and for all inputs $x\in \Pi_{YES}$, all nondeterministic choices cause $\mathcal{A}_{coN}$ to output ``YES''.
\item For all $x\notin \Pi_{YES}\cup\Pi_{NO}$, either some nondeterministic choices allow $\mathcal{A}_{N}$ to output ``YES'', or some nondeterministic choices allow $\mathcal{A}_{coN}$ to output ``NO''.
\end{enumerate}
\end{definition}
For context, we note that satisfying the first two conditions is equivalent to being in $\ncon[T]$.
We now state the slightly weaker result that we use, which is proved in Appendix~\ref{app:fgc}.
\begin{lemma}[Lemma~\ref{lem:nseth-old} for promise problems]
  \label{lem:cgimps}
  \label{lem:nseth}
  If NSETH holds and $\Pi\in \nconplus[T]$ for some promise problem $\Pi$, then $(\Pi,T^{1+\delta})$ is not SETH-hard for any $\delta>0$.
\end{lemma}

\paragraph{NUNSETH and non-reducibility for randomized reductions.}
The \emph{Non-Uniform Nondeterministic Strong Exponential Time Hypothesis (NUNSETH)} states that, for every $\varepsilon > 0$, there exists $k$ such that there are no nondeterministic circuit families of size $O(2^{n(1-\varepsilon)})$ recognizing the language \textsc{$k$-TAUT}.
The following is a result of \cite{CarmosinoGIMPS16} that NUNSETH rules out randomized SETH-hardness results, but adapted to promise problems.
\begin{lemma}[Lemma 3.9 of \cite{CarmosinoGIMPS16} for promise problems]
  \label{lem:nunseth}
  If NUNSETH holds and $\Pi\in \nconplus[T]$ for some promise problem $\Pi$, then $(\Pi,T^{1+\delta})$ is not SETH-hard for randomized reductions for any $\delta>0$.
\end{lemma}

\paragraph{Diameter approximations.}
Given positive real numbers $D'>D$, the \emph{$D'/D$-Diameter} problem is the promise problem $(\Pi_{YES},\Pi_{NO})$ whose input is a graph, and where $\Pi_{YES}$ is the set of graphs with Diameter at least $D'$ and $\Pi_{NO}$ is the set of graphs with Diameter at most $D$.
For a parameter $\alpha\ge 1$, the \emph{$\alpha$-approximate-Diameter} problem is the function problem whose input is a graph, and whose output must be a value between $D$ and $\alpha D$, where $D$ is the diameter of the graph.

Our non-reducibility results (e.g. Theorem~\ref{thm:nseth-weight-intro} and Theorem~\ref{thm:nseth-intro})  describe when an $\alpha$-approximation of Diameter is not SETH-hard for some $\alpha$.
The following lemma shows that it suffices to prove that the corresponding promise problem $\alpha/1$-Diameter is not SETH-hard.
We note that, even without this lemma, ruling out SETH-hardness for the promise problem is already interesting, as (to our knowledge) all known lower bounds for Diameter \cite{RodittyW13, ChechikLRSTW14, BackursRSWW18, Bonnet20} prove SETH-hardness of a corresponding promise problem.
We also note that, because of Lemma~\ref{lem:gap-to-approx}, we can sidestep the inconvenience of having to define and work with nondeterministic algorithms for function problems.
Lemma~\ref{lem:gap-to-approx} is proved in Appendix~\ref{app:fgc}.
\begin{lemma}
\label{lem:gap-to-approx}
  Let $\alpha\ge 1$ and $\beta>0$ be constants, $\rho\in[1,\infty]$, and $T$ be some time complexity. If $\alpha/1$-Diameter on graphs\footnote{If $\rho=\infty$, this is all graphs} satisfying $W_{max}/W_{min}\le \rho$ with time complexity $T$ is not SETH hard for deterministic (randomized) reductions, then $\alpha+\beta$-approximate-Diameter on graphs satisfying $W_{max}/W_{min}\le \rho$ with time complexity $T$ is not SETH-hard for deterministic (randomized) reductions.
\end{lemma}

\paragraph{Hopsets.}
In a graph $G=(V,E)$, given a set of edges $E'$ on vertices $V$, let $G+E'$ denote the graph $G$ with edges $E'$ added, i.e. $G+E'=(V,E\cup E')$.
In a graph $G$, a \emph{$(\beta,\varepsilon)$-hopset} \cite{Cohen00} is a set of weighted edges $E'$ (sometimes called \emph{shortcuts}) such that, in the graph $G+E'$, all shortest-path distances are the same as in $G$, and for any two vertices $v$ and $v'$, there exists in $G+E'$ a $\beta$-edge path between $v$ and $v'$ of length at most $(1+\varepsilon)\cdot d(v,v')$.
In a graph $G$ with diameter $D$, a \emph{$(\beta,\varepsilon)$-additive-hopset} is a set of weighted edges $E'$ such that, in the graph $G+E'$, all shortest-path distances are the same as in $G$, and for any two vertices $v$ and $v'$, there exists in $G+E'$ a $\beta$-edge path between $v$ and $v'$ of length at most $d(v,v') + \varepsilon D$.
Note that any $(\beta,\varepsilon)$-hopset is by definition a $(\beta,\varepsilon)$-additive-hopset because all shortest path distances are at most $D$.

We use the following hopset constructions.
The first one is a hopset construction for directed graphs.
The construction is implicit in \cite{CaoFR19,CaoFR20,CaoPrivate}, and we explain how in Appendix~\ref{app:hopset}.
We thank Arun Jambulapati for the reference \cite{CaoFR20}.
\begin{lemma}[Implicit in \cite{CaoFR19,CaoFR20,CaoPrivate}]
  Let $\alpha\in (0,1/2)$ and $\varepsilon\in(0,1)$.
  These exists a randomized algorithm with running time $\tilde O_\varepsilon(m^{1+4\alpha})$ such that, given a directed weighted graph $G$ on $n$ vertices and $m$ edges, computes a set of weighted edges $E'$ such that, (1) with probability 1, when $E'$ is added to $G$, all shortest path distances stay the same, and (2) with positive probability, $E'$ forms a $(n^{1/2-\alpha+o(1)},\varepsilon)$-additive-hopset of $G$.
\label{lem:hopset-dir}
\end{lemma}
The second hopset construction is for undirected graphs. 
\begin{lemma}
  Let $\delta\in(0,1)$ and $\varepsilon\in(0,1)$.
  These exists a randomized algorithm with running time $\tilde O_{\delta,\varepsilon}(m^{1+\delta})$ such that, given an undirected weighted graph $G$ on $n$ vertices and $m$ edges, computes a set of weighted edges $E'$ such that, (1) with probability 1, when $E'$ is added to $G$, all shortest path distances stay the same, and (2) with positive probability, $E'$ forms a $(O_{\delta,\varepsilon}(1),\varepsilon)$-additive-hopset of $G$.
\label{lem:hopset-undir}
\end{lemma}
Lemma~\ref{lem:hopset-undir} asks for a hopset construction but with the additional guarantee that, when the randomness is bad, the algorithm fails ``safe'' according to condition (1).
We expect that Lemma~\ref{lem:hopset-undir} can be inferred from existing hopset constructions, but we include a proof in Appendix~\ref{app:hopset} for completeness.

\section{Directed Diameter $2-\varepsilon$ hardness}
\label{sec:2eps}
In this section, we prove Theorem~\ref{thm:2eps-intro}.
\begin{proof}[Proof of Theorem~\ref{thm:2eps}]
  We show that, given a single-set $k$-OV instance $\Phi$ of size $\tilde{n}$ and dimension $d=\Theta(\log n)$, it is possible to construct a graph $G$ in $\tilde O(\tilde{n}^{k-1})$ time such that, if $\Phi$ has a solution, then the diameter is at least $2k-1$, and if $\Phi$ has no solution, then the diameter is at most $k$.

  Let $\Phi$ be given by a set $A\subset\{0,1\}^d$ of $\tilde{n}$ vectors, where $d=\Theta(\log n)$.

  When $a\in\{0,1\}^d$ is a vector and $x\in[d]$, we let $a[x]$ denote the $x$th coordinate $a_x$ of $a$.
  When dealing with elements of $A$, we index coordinates using $a[x]$ rather than $a_x$ for clarity and to stay consistent with the notation in \cite{BackursRSWW18}.
  We use $\bar a$ and $\bar x$ to refer to tuples of vectors and indices, respectively.
  We refer to $\bar a$ in $A^{k-1}$ or $A^{k-2}$ as a \emph{vector-tuple}, and refer to $\bar x\in[d]^{k-1}$ as an \emph{index-tuple}.

  We make the following useful definitions.
  \begin{enumerate}
  \item (Property $(i,L_1)$) For $i=1,\dots,k-1$, an index-tuple $\bar x\in [d]^{k-1}$ of $k-1$ indices, and a vector $a\in\{0,1\}^d$, we say the pair $(a, \bar x)$ \emph{has property $(i,L_1)$} if $a[x_1]=\cdots=a[x_{k-i}]=1$.
  \item (Property $(i,L_k)$) For $i=2,\dots,k$, an index-tuple $\bar x\in [d]^{k-1}$ of $k-1$ indices, and a vector $a\in\{0,1\}^d$, we say the pair $(a,\bar x)$ \emph{has property $(i,L_{k-1})$} if $a[x_{k+1-i}]=\cdots=a[x_{k-1}]=1$.
  \end{enumerate}

  Construct a graph $G$ with vertex set $V=L_0\cup L_1\cup\cdots\cup L_{k-1}$.
  \begin{enumerate}
  \item Create one vertex in $L_0$ for each element of $A^{k-1}$.
  \item For $2\le i \le k-2$, create one vertex in $L_i$ for each element of $(\bar a, \bar x)\in A^{k-2}\times [d]^{k-1}$.
  \item Create one vertex in $L_1$ for each element of $(\bar a, \bar x)\in A^{k-2}\times [d]^{k-1}$ such that $(a_j,\bar x)$ has property $(j,L_1)$ for $1\le j \le k-2$.
  \item Create one vertex in $L_{k-1}$ for each element of $(\bar a, \bar x)\in A^{k-2}\times [d]^{k-1}$ such that $(a_{j-2},\bar x)$ has property $(j,L_{k-1})$ for $3\le j \le k$.
  \end{enumerate}
  We use $(\bar a)_{L_0}$ or $(\bar a, \bar x)_{L_i}$ to denote the corresponding vertex of $L_i$, using the subscript $L_i$ for disambiguation.
  We construct the edges of $G$ as follows.
  \begin{enumerate}
  \item ($L_0\to L_1$) For all vectors $a_1,\dots,a_{k-1}\in A$ and index-tuples $\bar x\in [d]^{k-1}$ such that $\alpha_1=(a_1,\dots,a_{k-2},\bar x)_{L_1}$ exists as a vertex in $L_1$ and such that $a_{k-1}[x_1]=1$, add an edge from vertex $(a_1,\dots,a_{k-1})_{L_0}\in L_0$ to vertex $\alpha_1$.

  \item ($L_{k-1}\to L_0$) For all vectors $a_2,\dots,a_k\in A$ and index-tuples $\bar x\in[d]^{k-1}$ such that $\alpha_{k-1}=(a_3,\dots,a_k,\bar x)_{L_{k-1}}$ exists as a vertex in $L_{k-1}$ and such that $a_2[x_{k-1}]=1$, add an edge from vertex $\alpha_{k-1}\in L_{k-1}$ to vertex $(a_2,\dots,a_k)_{L_0}\in L_0$.

  \item ($L_i\to L_{i+1}$) For any $i=1,\dots,k-2$, any index-tuple $\bar x\in[d]^{k-1}$, and any vector tuples $\bar a$ and $\bar b$ such that $a_j=b_j$ for all $j\neq k-1-i$, add an edge from vertex $\alpha_i=(\bar a,\bar x)_{L_i}$ to vertex $\alpha_{i+1}=(\bar b,\bar x)_{L_{i+1}}$ if both vertices exist.

  \item ($L_i\to L_1$) For every vector-tuple $\bar a\in A^{k-2}$ and every two index-tuples $\bar x,\bar x'\in[d]^{k-1}$, and every $1\le i\le k-1$, add a directed edge from $(\bar a,\bar x)_{L_i}$ to $(\bar a,\bar x')_{L_1}$ if both vertices exist.

  \item ($L_{k-1}\to L_i$) For every vector-tuple $\bar a\in A^{k-2}$ and every two index-tuples $\bar x,\bar x'\in[d]^{k-1}$, and every $1\le i\le k-1$, add a directed edge from $(\bar a,\bar x)_{L_{k-1}}$ to $(\bar a,\bar x')_{L_i}$ if both vertices exist.
  \end{enumerate}
  We refer to the first three types of edges as \emph{vector-changing} edges, as the vector-tuple changes, and the index tuple, if present does not.
  We refer to the latter two types of edges as \emph{index-changing} edges, as the vector-tuple stays constant, while the index tuple changes.

\paragraph{Runtime.}
The number of vertices in each $L_i$ is at most $O(\tilde{n}^{k-1}+\tilde{n}^{k-2}d^{k-1}) = \tilde O(\tilde{n}^{k-1})$.
Each edge has at least one endpoint in $L_1\cup\cdots\cup L_{k-1}$, and each such vertex has at most $O(\tilde{n}+d^k)$ neighbors.
$L_1\cup\cdots\cup L_{k-1}$ have at most $O(\tilde{n}^{k-2}d^{k-1})$ vertices, so the total number of edges is thus at most $O(\tilde{n}^{k-1}d^{2k-1}) \le \tilde O(\tilde{n}^{k-1})$.
Thus, the runtime to produce $G$ from $\Phi$ is at most $\tilde O(\tilde{n}^{k-1})$ as desired.

\tikzset{ 
table/.style={
  matrix of nodes,
  nodes={rectangle,draw=black,text width=1,align=center},
  text height=1,
  nodes in empty cells
  },
texto/.style={font=\footnotesize},
title/.style={font=\small\sffamily}
}
\newcommand\CellText[2]{%
  \node[texto,anchor=east]
  at (mat#1.west)
  {#2};
}
\newcommand\RightCellText[2]{%
  \node[texto,anchor=west]
  at (mat#1.east)
  {#2};
}
\newcommand\AboveCellText[2]{%
  \node[texto]
  at ($(mat#1.north)+(0,0.4)$)
  {#2};
}
\newcommand\BelowCellText[2]{%
  \node[texto]
  at ($(mat#1.south)+(0,-0.4)$)
  {#2};
}
\newcommand\RowTitle[2]{%
\node[title,left=6.3cm of mat#1,anchor=west]
  at (#1.north west)
  {#2};
}
\newcommand\fil{|[fill=blue]|}

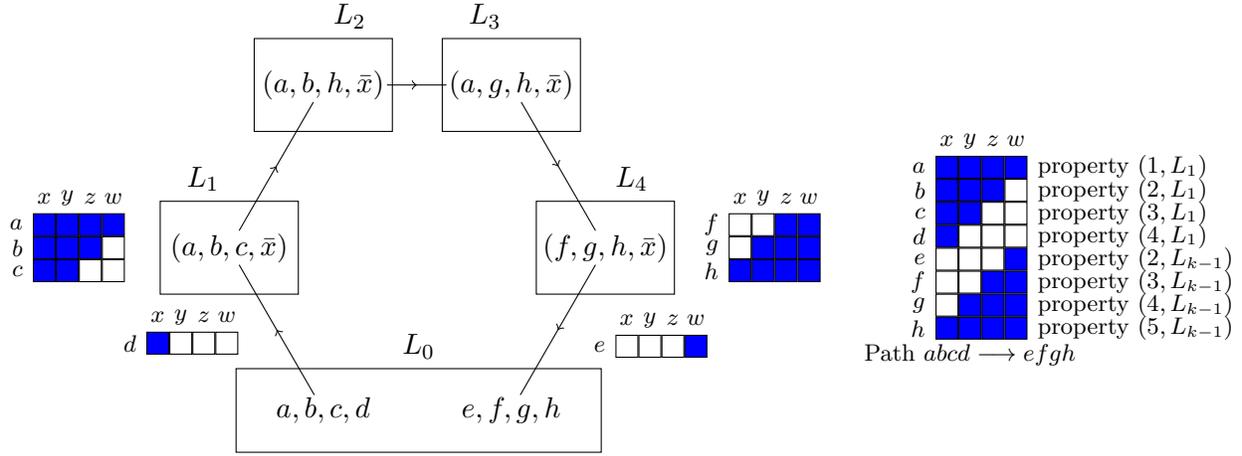
\begin{figure}
\begin{center}
\begin{tikzpicture}[scale=0.5, decoration={markings,mark=at position 0.5 with{\arrow{>}}}]
  \node[inner sep=1, outer sep=0] (0) at (240:5) {$a,b,c,d$};
  \node[inner sep=1, outer sep=0] (5) at (300:5) {$e,f,g,h$};
  \node[inner sep=1, outer sep=0] (1) at (180:5) {$(a,b,c,\bar x)$};
  \node[inner sep=1, outer sep=0] (2) at (120:5) {$(a,b,h,\bar x)$};
  \node[inner sep=1, outer sep=0] (3) at (60:5) {$(a,g,h,\bar x)$};
  \node[inner sep=1, outer sep=0] (4) at (0:5) {$(f,g,h,\bar x)$};
  \begin{pgfonlayer}{background}
  \node[inner sep=17, outer sep=0, draw,label={[anchor=south east]$L_1$}] (1a) at (1) {\white{\ldots\,\,}};
  \node[inner sep=17, outer sep=0, draw,label={[anchor=south west]$L_2$}] (2a) at (2) {\white{\ldots\,\,}};
  \node[inner sep=17, outer sep=0, draw,label={[anchor=south east]$L_3$}] (3a) at (3) {\white{\ldots\,\,}};
  \node[inner sep=17, outer sep=0, draw,label={[anchor=south west]$L_4$}] (4a) at (4) {\white{\ldots\,\,}};
  \end{pgfonlayer}
  \draw[] ($(0.north west)+(-1,0.7)$) rectangle ($(5.south east)+(1,-0.7)$) node[pos=0.5,label={[yshift=12]$L_0$}]{};
  \draw[postaction={decorate}] (0)--(1);
  \draw[postaction={decorate}] (1)--(2);
  \draw[postaction={decorate}] (2)--(3);
  \draw[postaction={decorate}] (3)--(4);
  \draw[postaction={decorate}] (4)--(5);

  \begin{scope}[shift=(203:6.5)]
    \matrix[table] (mat0)
    {
     \fil &  &  &  \\
    };
    \CellText{0-1-1}{$d$};
    \AboveCellText{0-1-1}{$x$};
    \AboveCellText{0-1-2}{$y$};
    \AboveCellText{0-1-3}{$z$};
    \AboveCellText{0-1-4}{$w$};
  \end{scope}
  \begin{scope}[shift=(180:9)]
    \matrix[table] (mat1)
    {
    \fil & \fil & \fil & \fil \\
    \fil & \fil & \fil & \\
    \fil & \fil & & \\
    };
    \CellText{1-1-1}{$a$};
    \CellText{1-2-1}{$b$};
    \CellText{1-3-1}{$c$};
    \AboveCellText{1-1-1}{$x$};
    \AboveCellText{1-1-2}{$y$};
    \AboveCellText{1-1-3}{$z$};
    \AboveCellText{1-1-4}{$w$};
  \end{scope}
  \begin{scope}[shift=(0:9.5)]
    \matrix[table] (mat4)
    {
     &  & \fil & \fil \\
     & \fil & \fil & \fil \\
    \fil & \fil & \fil & \fil \\
    };
    \CellText{4-1-1}{$f$};
    \CellText{4-2-1}{$g$};
    \CellText{4-3-1}{$h$};
    \AboveCellText{4-1-1}{$x$};
    \AboveCellText{4-1-2}{$y$};
    \AboveCellText{4-1-3}{$z$};
    \AboveCellText{4-1-4}{$w$};
  \end{scope}
  \begin{scope}[shift=(-22:7)]
    \matrix[table] (mat5)
    {
     &  &  & \fil \\
    };
    \CellText{5-1-1}{$e$};
    \AboveCellText{5-1-1}{$x$};
    \AboveCellText{5-1-2}{$y$};
    \AboveCellText{5-1-3}{$z$};
    \AboveCellText{5-1-4}{$w$};
  \end{scope}

  \begin{scope}[shift=(0:15)]
    \matrix[table] (matA)
    {
     \fil & \fil & \fil & \fil \\
     \fil & \fil & \fil &  \\
     \fil & \fil &  &  \\
     \fil &  &  &  \\
     &  &  & \fil \\
     &  & \fil & \fil \\
     & \fil & \fil & \fil \\
    \fil & \fil & \fil & \fil \\
    };
    \CellText{A-1-1}{$a$};
    \CellText{A-2-1}{$b$};
    \CellText{A-3-1}{$c$};
    \CellText{A-4-1}{$d$};
    \CellText{A-5-1}{$e$};
    \CellText{A-6-1}{$f$};
    \CellText{A-7-1}{$g$};
    \CellText{A-8-1}{$h$};
    \AboveCellText{A-1-1}{$x$};
    \AboveCellText{A-1-2}{$y$};
    \AboveCellText{A-1-3}{$z$};
    \AboveCellText{A-1-4}{$w$};
    \BelowCellText{A-8-2}{Path $abcd\longrightarrow efgh$};
    \RightCellText{A-1-4}{property $(1,L_1)$};
    \RightCellText{A-2-4}{property $(2,L_1)$};
    \RightCellText{A-3-4}{property $(3,L_1)$};
    \RightCellText{A-4-4}{property $(4,L_1)$};
    \RightCellText{A-5-4}{property $(2,L_{k-1})$};
    \RightCellText{A-6-4}{property $(3,L_{k-1})$};
    \RightCellText{A-7-4}{property $(4,L_{k-1})$};
    \RightCellText{A-8-4}{property $(5,L_{k-1})$};
  \end{scope}
\end{tikzpicture}
\end{center}
\caption{Here, we illustrate $G$ when $k=5$.
For clarity, we use $a,b,c,\dots$ instead of $a_1,a_2,a_3,\dots$, and we use $x,y,z,w$ instead of $x_1,x_2,x_3,x_4$, and we use $\bar x\defeq (x,y,z,w)$.
The blue cells indicate coordinates that need to equal 1 for the edge to exist.
The table on the right indicates all coordinates that need to equal 1 for the path from $(a,b,c,d)$ to $(e,f,g,h)$ to exist.
}
\end{figure}

\paragraph{$k$-OV no solution.}
We first show that if there is no $k$-OV solution, then the diameter is at most $k$.
For any tuple $\bar a=(a_1,\dots,a_l)$ of at most $k$ vectors, by assumption, there exists some index $x$ such that $a_1[x]=a_2[x]=\cdots=a_l[x]=1$.
Let $\ind{\bar a}$ denote one such index $x$.

The crucial claim of this section is the following.
\begin{claim}
  \label{cl:nosoln}
  For any two vertices in $L_0$, there is a length $k$ path from one to the other.
\end{claim}
\begin{proof}
  Let the vertices be $\alpha_0\defeq (a_1,\dots,a_{k-1})_{L_0}$ and $\alpha_k\defeq (b_2,\dots,b_k)_{L_0}$.
  For $i=1,\dots,k-1$, let $x_i = \ind{a_1,\dots,a_{k-i},b_{k-i+1},\dots,b_k}$ and $\bar x = (x_1,\dots,x_{k-1})$.
  By definition of $x_1,\dots,x_{k-1}$, for $i=1,\dots,k-1$, we have $a_i[x_1]=\cdots=a_i[x_{k-i}]=1$, so the pair $(a_i,\bar x)$ satisfies property $(i,L_1)$.
  By definition of $x_1,\dots,x_{k-1}$, for $i=2,\dots,k$, we have $b_i[x_{k+1-i}]=\cdots=b_i[x_{k-1}]=1$, so the pair $(b_i,\bar x)$ satisfies property $(i,L_{k-1})$.
  Thus, $\alpha_1\defeq (a_1,\dots,a_{k-2},\bar x)_{L_1}$ is a valid vertex in $L_1$ and $\alpha_{k-1}\defeq (b_3,\dots,b_k)_{L_{k-1}}$ is a valid vertex in $L_{k-1}$, and furthermore there are edges from $\alpha_0$ to $\alpha_1$ and from $\alpha_{k-1}$ to $\alpha_k$.

  For $i=2,\dots,k-2$, vertex $\alpha_i\defeq (a_1,\dots,a_{k-1-i},b_{k-i+2},\dots,b_k, \bar x)_{L_i}$ exists as a vertex in $L_i$.
  For $i=1,\dots,k-2$, $\alpha_i$ and $\alpha_{i+1}$ have the same index-tuple and have vector tuples that differ only in the $k-1-i$'th vector of the vector-tuple.
  Thus, there is an edge from $\alpha_i$ to $\alpha_{i+1}$ for all $i=1,\dots,k-2$.
  Hence, $\alpha_0,\dots,\alpha_k$ is a length $k$ path from $\alpha_0$ to $\alpha_k$, as desired.
\end{proof}
\begin{claim}
  \label{cl:neigh}
  For any vertex $v\in V$ in our graph, there exist vertices $u,w\in L_0$ such that $N^{out}(u)\subset N^{out}(v)$ and $N^{in}(w)\subset N^{in}(v)$. 
\end{claim}
\begin{proof}
  If $v\in L_0$, simply set $u=w=v$.
  Otherwise, $v\in L_1\cup\cdots\cup L_{k-1}$ has some vector tuple $(a_1,\dots,a_{k-2})\in A^{k-2}$ and index tuple $\bar x\in [d]^{k-1}$.
  Let $b\in A$ be an arbitrary vector and let $u = (a_1,\dots,a_{k-2},b)_{L_0}$ and $w=(b,a_1,\dots,a_{k-2})_{L_0}$.
  The only edges out of $L_0$ are to $L_1$, so all out-neighbors of $u$ are in $L_1$ and of the form $(a_1,\dots,a_{k-2},\bar x')_{L_1}$.
  All such vertices are also out-neighbors of $v$ as desired.
  Similarly, the only edges into $L_0$ are from $L_{k-1}$, so all in-neighbors of $w$ are in $L_{k-1}$ and of the form $(a_1,\dots,a_{k-2},\bar x')_{L_{k-1}}$.
  All such vertices are also in-neighbors of $v$ as desired.
\end{proof}

Let $v,v'\in V$.
By Claim~\ref{cl:neigh}, there exists $u,w'\in L_0$ such that $N^{out}(u)\subset N^{out}(v)$ and $N^{in}(w')\subset N^{in}(v')$.
By Claim~\ref{cl:nosoln}, there exists a length $k$ path $u,v_1,v_2,\dots,v_{k-1},w'$ from $u$ to $w'$.
Since $v_1\in N^{out}(u)$, we must have $v_1\in N^{out}(v)$, and since $v_{k-1}\in N^{in}(w')$, we must have $v_{k-1}\in N^{in}(v')$.
Thus, $v,v_1,\dots,v_{k-1},v'$ is a length $k$ path from $v$ to $v'$.
This holds for any vertices $v$ and $v'$, so the diameter is at most $k$, as desired.
\paragraph{$k$-OV solution.}

We now show that, if there is a $k$-OV solution, then the diameter is at least $2k-1$.
Let $a_1,\dots,a_k$ be the $k$ orthogonal vectors, i.e. $a_1[x]\cdots a_k[x]=0$ for all $x\in[d]$.
We claim that the distance from $\alpha_0 = (a_1,\dots,a_{k-1})_{L_0}\in L_0$ to $\alpha_k = (a_2,\dots,a_k)_{L_0}\in L_0$ is at least $2k-1$.

Suppose for contradiction there exists a length $\ell\le 2k-2$ path $\mathcal{P}$ from $\alpha_0$ to $\alpha_k$.
Let $\beta_0=\alpha_0,\beta_1,\dots,\beta_{\ell}=\alpha_k$ denote the vertices of path $\mathcal{P}$.
We note the following two observations about $G$ that follow from the edge definitions.
\begin{fact}
  \label{fact:g-1}
  The only edges in graph $G$ that go from vertex subset $L_i$ to vertex subset $L_j$ for $0\le i < j\le k-1$ are edges from $L_i$ to $L_{i+1}$ from $i=0,\dots,k-1$.
\end{fact}
\begin{fact}
  \label{fact:g-2}
  In graph $G$, any edge from $L_i$ to $L_j$ for $k-1\ge i \ge j\ge 0$ satisfy either $i=k-1$ or $j=1$.
\end{fact}
The only edges to $L_0$ are those from $L_{k-1}$, and any path from $L_0$ to $L_{k-1}$ has length at least $k-1$ by Fact~\ref{fact:g-1}.
Thus, two vertices in path $\mathcal{P}$ in $L_0$ must be separated by distance at least $k$.
As the path $\beta_0,\dots,\beta_\ell$ starts and ends in $L_0$ and has length at most $2k-2$, it follows that the path visits $L_0$ only at the first vertex $\beta_0$ and the last vertex $\beta_\ell$.
Thus, $\beta_1,\dots,\beta_{\ell-1}$ lie entirely in $L_1,\dots,L_{k-1}$.

By Fact~\ref{fact:g-1}, we must have $\beta_1\in L_1$.
Let $r$ be the largest integer such that vertex $\beta_r$ is in $L_1$.
By Fact~\ref{fact:g-1}, there are at least $k-1$ more vertices on the path $\mathcal{P}$, so $r\le \ell-(k-1)\le k-1$.
By Fact~\ref{fact:g-2}, as none of $\beta_{r+1},\dots,\beta_{r+k-2}$ are in $L_1$, we must have $\beta_{r+i}\in L_{i+1}$ for $i=1,\dots,k-2$.

Let $s$ denote the largest integer such that the subpath $\beta_0,\dots,\beta_r$ has a vertex in subset $L_s$, and let $t$ denote the smallest integer such that the subpath $\beta_{r+k-2},\dots,\beta_\ell$ has a vertex in subset $L_t$.
By Fact~\ref{fact:g-1}, the subpath $\beta_1,\dots,\beta_{r-1}$ must stay in $L_1\cup\cdots\cup L_{r-1}$. 
Hence, if $r\ge 2$, then $s\le r-1$, and $s=r=1$ otherwise.
Similarly, the path $\beta_{r+k-1},\dots,\beta_{\ell-1}$ has length $\ell-1-(r+k-1)\le k-r-2$ and thus must stay in $L_{r+1}\cup\cdots\cup L_{k-1}$.
Hence, if $r+k-1 \le \ell - 1$, then $t\ge r+1$, and $t=k-1$ otherwise.
Thus, if $r\ge 2$ and $r+k-1\le \ell-1$, we must have $t-s\ge 2$.

We know $\beta_1=(a_1,\dots,a_{k-2},\bar x)$ for some $\bar x$.
All vertices in $L_1,\dots,L_s$ are of the form $(a_1',\dots,a_{k-2}',\bar x')_{L_i}$, and any of the edges between them does not change the vectors $a_1',\dots,a_{k-1-s}'$: $L_i\to L_1$ edges do not change any of the vectors, and edges from $L_{i}\to L_{i+1}$  only changes $a_{k-1-i}'$.
Hence, we have that $\beta_r=(a_1,\dots,a_{k-1-s},a_{k-s}',\dots,a_{k-2}',\bar y)_{L_1}\in L_1$ for some vectors $a_{k-s}',\dots,a_{k-2}'$ and index tuple $\bar y$.
Because $\beta_r$ exist as a vertex in $L_1$, the pair $(a_i,\bar y)$ has property $(i,L_1)$ for $i\le k-1-s$.

As the index tuple of $\beta_r$ is $\bar y$, the path $\beta_r,\dots,\beta_{r+k-2}$ only uses $L_i\to L_{i+1}$ edges, and the index tuple does not change in $L_i\to L_{i+1}$ edges, the index tuple of $\beta_{r+k-2}$ is also $\bar y$.

We also have $\beta_{\ell-1}=(a_3,\dots,a_{k}, \bar x')$ for some $\bar x'$.
All vertices in $L_t,\dots,L_{k-1}$ are elements $(a_3',\dots,a_k',\bar x')$ of $A^{k-2}\times [d]^{k-1}$, and any of the edges between them does not change the vectors $a_{k+2-t}',\dots,a_k'$: $L_{k-1}\to L_i$ edges do not change any of the vectors, and $L_{i}\to L_{i+1}$ edges only change $a_{k+1-i}'$ (here we indexed vectors starting from $a_3'$). Hence, for some $a_{3}'',\dots,a_{k+1-t}''$, we have that $\beta_{r+k-2}$ is the vertex $(a_3'',\dots,a_{k+1-t}'',a_{k+2-t},\dots,a_k,\bar y)_{L_{k-1}}\in L_{k-1}$.
Because $\beta_{k-2+r}$ exists as a vertex in $L_{k-1}$, the pair $(a_{i},\bar y)$ has property $(i,L_{k-1})$ for $i\ge k+2-t$.

When $1\le i\le k-1-s$, we saw that $(a_i,\bar y)$ has property $(i,L_1)$, so $a_i[y_j]=1$ for all $j\le k-i$ and in particular for all $j\le k-(k-1-s) = s+1$.
Similarly, when $k+2-t\le i \le k$, we saw that $(a_i,\bar y)$ has property $(i,L_{k-1})$, so $a_i[y_j] = 1$ for $j\ge k+1-i$, and in particular for all $j\ge k+1-(k+2-t) = t-1$.
If $t-s\ge 2$, then we have that $a_1[j]=a_2[j]=\cdots=a_k[j] = 1$ for all $t-1\le j \le s+1$, a contradiction of orthogonality of $a_1,\dots,a_k$.

Thus, we must have $t-s\le 1$.
By an earlier argument, if $r>1$ and $r+k-2<\ell-1$, then $t-s\ge 2$.
Hence, we must either have $r=1$ or $r+k-2= \ell-1$.

If $r=1$, then $s=1$ and $t=2$, and by above, $(a_i,\bar y)$ has property $(i,L_1)$ for $i\le k-1-s$, and in particular for all $i\le k-2$.
Additionally, $(a_i,\bar y)$ has property $(i,L_{k-1})$ for $i=k$.
Thus, $a_1[y_1]=a_2[y_1]=\cdots=a_{k-2}[y_1] = 1$ and $a_k[y_1]=1$, respectively.
Furthermore, since $r=1$, we have $\beta_0 = (a_1,\dots,a_{k-1})_{L_0}$ and $\beta_1=(a_1,\dots,a_{k-2},\bar y)_{L_1}$ are the first two vertices of the path $\mathcal{P}$, so the first edge of the path implies additionally that $a_{k-1}[y_1] = 1$.
We conclude $a_1[y_1]=\cdots=a_k[y_1]=1$, contradicting orthogonality.

If $r+k-2=\ell-1$, then $s\le r-1 \le \ell-k \le k-2$ and $t=s+1\le k-1$.
By above, $(a_i,\bar y)$ has property $(i,L_1)$ for $i\le k-1-s$, and in particular $i=1$.
Additionally, $(a_i,\bar y)$ has property $(i,L_{k-1})$ for $i\ge k+2-t$, and in particular for all $3\le i\le k$.
We conclude that $a_1[y_{k-1}] = 1$ and $a_3[y_{k-1}] = \cdots = a_{k}[y_{k-1}] = 1$.
Additionally, since $r = \ell-1$, we have that $\beta_{\ell-1}=(a_3,\dots,a_k,\bar y)_{L_{k-1}}$ and $\beta_\ell=(a_2,\dots,a_k)_{L_0}$ are the last two vertices of the path $\mathcal{P}$, so the edge between them implies additionally that $a_2[y_{k-1}]=1$.
We thus have $a_1[y_{k-1}]=\cdots=a_k[y_{k-1}]=1$, contradicting orthogonality.

We have thus shown that in the three cases that the path $\beta_0,\dots,\beta_\ell$ could satisfy, (1) $r>1$ and $r+k-2<\ell-1$, (2) $r=1$, and (3) $r+k-2=\ell-1$, there is a contradiction.
This covers all cases, we have found a contradiction, so there cannot exist a length $\ell\le 2k-2$ path from $\alpha_0$ to $\alpha_k$, as desired.
\end{proof}

\section{Non-reducibility via hopsets}
\label{sec:nseth-weight}

\subsection{Directed weighted graphs}
By Lemma~\ref{lem:gap-to-approx}, to prove Theorem~\ref{thm:nseth-weight-intro} regarding the non-reducibility of approximate Diameter, it suffices to prove the following theorem on the non-reducibility of the corresponding promise problem.
\begin{theorem}
  Let $D$ be a positive real number, $\varepsilon>0$ and $D'=(5/3+\varepsilon)D$.
  For all $\delta>0$, under NSETH (NUNSETH), $D'/D$-Diameter on directed weighted graphs with time complexity $m^{19/13(1+\delta)}$ is \emph{not} SETH-hard for deterministic (randomized) reductions.
\label{thm:nseth-weight}
\end{theorem}

Theorem~\ref{thm:nseth-weight} in fact follows from the following more general theorem, which states that non-reducibility results follow from hopset constructions.
Using Theorem~\ref{thm:nseth-weight-conj} with the state-of-the-art hopset constructions \cite{CaoFR20}, we obtain Theorem~\ref{thm:nseth-weight}.
Furthermore, faster and better hopset constructions would give even stronger non-reducibility results.
\begin{theorem}
  Let $\varepsilon>0$ be fixed and $\beta$ be a real number, possibly depending on $n$ and $m$.
  Suppose there exists an algorithm with running time $T=T(n,m)$ that, given a directed weighted graph $G$ on $n$ vertices and $m$ edges, computes a set of weighted edges $E'$ such that, (1) when $E'$ is added to $G$, all shortest path distances stay the same, and (2) $E'$ forms a $(\beta,\varepsilon/2)$-additive hopset of $G$.
  
  Let $D$ be a real number and $k\ge 2$ be a positive integer, and $D' = (2-\frac{1}{k}+\varepsilon)D$.
  Under NSETH (NUNSETH), for any $\delta>0$, $D'/D$-Diameter on directed weighted graphs with time complexity $(m^{1+1/k}\beta^{1-2/k}+T)^{1+\delta}$ is \emph{not} SETH-hard for deterministic (randomized) reductions.
\label{thm:nseth-weight-conj}
\end{theorem}
By Lemma~\ref{lem:cgimps} and Lemma~\ref{lem:nunseth}, Theorem~\ref{thm:nseth-weight-conj} follows from the following theorem.
\begin{theorem}
  Suppose $\varepsilon>0$ and the algorithm of Theorem~\ref{thm:nseth-weight-conj} exists.
  Let $D$ be a real number and $k\ge 2$ be a positive integer, and $D' = (2-\frac{1}{k}+\varepsilon)D$.
  Then $D'/D$-Diameter on directed weighted graphs is in $\nconplus[m^{1+1/k}\beta^{1-2/k}+T]$.
\label{thm:nseth-weight-conj-ncon}
\end{theorem}
To prove Theorem~\ref{thm:nseth-weight-conj-ncon}, we need a few technical lemmas.
The first is a standard lemma about hitting sets.
\begin{lemma}
  \label{lem:select}
  Let $n$, $M$ and $K$ be positive integers. Let $A_1,\dots,A_n$ denote $\ell$ sets of size at least $K$ over a universe of size $M$. Then there exists a set $X$ of $2M/K\log n$ elements of $M$ such that $X\cap A_i\neq\emptyset$ for all $i=1,\dots,n$.
\end{lemma}
\begin{proof}
  Let $X$ be $2M/K\log n$ elements chosen uniformly at random with replacement from $U$. The probability $X\cap A_i=\emptyset$ is at most $(1-K/M)^{|X|} \le e^{-|X|\cdot K/M} = n^{-2}$.
  By the union bound, the probability there exists an $i$ such that $X\cap A_i=\emptyset$ is at most $n\cdot n^2 < 1$, so some choice of $X$ yields $X\cap A_i\neq\emptyset$ for all $i$.
\end{proof}
We apply Lemma~\ref{lem:select} in the following structural result about graphs with diameter $D$.
\begin{lemma}
  Let $G$ be a directed weighted graph of diameter $D$, and let $M_1,M_2$ be positive real numbers with $M_1M_2\ge 8m\log m$, and let $D_1,D_2$ be positive real numbers with $D_1+D_2\ge D$. Then $G$ either has a vertex subset $Y$ satisfying $|Y|\le M_1$ and $d(v,Y)\le D_1$ for all vertices $v$, or $G$ has a vertex subset $Z$ satisfying $|Z|\le M_2$ and $d(Z,v)\le D_2$.
\label{lem:struct}
\end{lemma}
\begin{proof}
  Suppose there exists a vertex $u$ with $|E_{\le D_2}^{in}(u)|\le M_1/2$. Then set $X_1$ to be the vertices incident to $E_{\le D_2}^{in}(u)$.
  For each vertex $v$, there is a length at most $D$ path $v=v_0,v_1,\dots,v_\ell=u$ from $v$ to $u$. If $i$ is the largest index such that $d(v_i,u) > D_2$, we have $v_i\in X_1$, so $d(v,v_i)<D-D_2 < D_1$.
  Thus, $d(v,X_1)< D_1$, as desired.

  Now suppose that for all vertices $u$, we have $|E_{\le D_2}^{in}(u)|\ge M_1/2$.
  By Lemma~\ref{lem:select} on the universe of $m$ edges with $n$ sets $E_{\le D_2}^{in}(u)$, there exists a set $F$ of $2m/(M_1/2)\cdot \log n\le M_2/2$ edges intersecting each $E_{\le D_2}^{in}(u)$. Setting $X_2$ to be the vertices incident to $F$ gives $|X_2|\le M_2$ and $d(X_2,u)\le D_2$ for all vertices $u$.
\end{proof}
We now prove Theorem~\ref{thm:nseth-weight-conj-ncon}.
\begin{proof}[Proof of Theorem~\ref{thm:nseth-weight-conj-ncon}]
  Let $G$ be the input graph.
  Recall that $D'/D$-Diameter is a promise problem $(\Pi_{YES},\Pi_{NO})$ where the instances in $\Pi_{YES}$ are graphs of diameter at least $D'$, and the instances in $\Pi_{NO}$ are graphs of diameter at most $D$.
  Let $\mathcal{A}_N$ be the nondeterministic algorithm that guesses a vertex, outputs ``YES'' if the eccentricity is greater than $D$, and ``NO'' otherwise.
  For any graph \emph{not} in $\Pi_{NO}$, i.e., any graph with diameter greater than $D$, there exist nondeterministic choices such that $\mathcal{A}_N$ outputs ``YES''. Furthermore, for graphs in $\Pi_{NO}$, $\mathcal{A}_N$ always outputs ``NO''.
  As $\mathcal{A}_N$ runs in time $m^{1+o(1)}$, we have that properties 1 and 3 of Definition~\ref{def:nconplus} are satisfied.

  It now remains to construct a nondeterministic algorithm $\mathcal{A}_{coN}$ satisfying property 2 of Definition~\ref{def:nconplus}.
  Let $\mathcal{A}_{hopset}$ be the algorithm that is assumed with parameters $\beta$ and $\varepsilon$.
  Define $\mathcal{A}_{coN}$ to be the following algorithm.
  \begin{enumerate}
  \item For $\ell=1,\dots,k-1$, nondeterministically choose vertex subsets $X_\ell^{in}$ and $X_\ell^{out}$ of size $8m^{1-(\ell/k)}\beta^{-1 + 2\ell/k}\log m$. 
  \item Run single source shortest path from every vertex in $X_{k-1}^{in}$. If each vertex in $X_{k-1}^{in}$ has out-eccentricity at most $D$, and if $d(v,X_{k-1}^{in})\le (\frac{k-1}{k})D$ for all vertices $v$, output ``NO''.
  \item Run single source shortest path from every vertex in $X_{k-1}^{out}$. If each vertex in $X_{k-1}^{out}$ has in-eccentricity at most $D$, and if $d(X_{k-1}^{out},v)\le (\frac{k-1}{k})D$ for all vertices $v$, output ``NO''.
  \item Nondeterministically choose a random seed for the hopset algorithm $\mathcal{A}_{hopset}$.
  \item Run $\mathcal{A}_{hopset}$ to obtain a set $E'$ and a new graph $G'=G+E'$.
  \item For $\ell=1,\dots,k-2$, run multi-source shortest path from each of $X_{\ell}^{in}$ and $X_{\ell}^{out}$.
  \item Let $\ell\in\{1,\dots,k-2\}$ be such that $d(v,X_\ell^{in})\le (\frac{\ell}{k})D$ and $d(X_{k-1-\ell}^{out},v) \le (\frac{k-\ell}{k})D$ for all vertices $v$. If no such vertices exist, output ``YES''.
  \item Nondeterministically choose a sequence $s_{x,x'}$ of at most $\beta$ vertices for every $x\in X_{\ell}^{in}$ and $x'\in X_{k-1-\ell}^{out}$. If each $s_{x,x'}$ forms a path from $x$ to $x'$ of length at most $D(1+\varepsilon/2)$ in $G'$, and $d(v,X_{\ell}^{in})\le D/3$ and $d(X_{k-1-\ell}^{out},v)\le D/3$ for all vertices $v$ output ``NO''. Otherwise output ``YES''.
  \end{enumerate}

  \textbf{Runtime.}
  Choose the sets $X_{\ell}^{in}, X_{\ell}^{out}$ takes time $\tilde O(m)$.
  Steps 2 and 3 each take $\tilde O(m^{1+1/k}\beta^{1-2/k})$, as each shortest path takes time $\tilde O(m)$.
  Steps 4 and 5 together take time $T$ by assumption on algorithm $\mathcal{A}_{hopset}$.
  Step 6 takes time $\tilde O(m)$.
  Step 7 takes time $\tilde O(|X_{\ell}^{in}|\cdot |X_{\ell}^{out}|\cdot \beta) = \tilde O(m^{1+1/k}\beta^{1-2/k})$.

  Thus, the total runtime is $\tilde O(m^{1+1/k}\beta^{1-2/k}+T)$.
  
  \textbf{Correctness.}
  We first show that if the algorithm outputs ``NO'', the diameter is less than $D'$.
  If we output ``NO'' at Step 2, then for any vertices $v$ and $v'$, there exists a vertex $x\in X_{k-1}^{in}$ such that $d(v,x)\le (\frac{k-1}{k})D$. As $x\in X_{k-1}^{in}$ has out-eccentricity at most $D$, we have $d(v,v')\le d(v,x)+d(x,v')\le (\frac{2k-1}{k})D < D'$ by the triangle inequality.
  If we output ``NO'' at Step 3, then for any vertices $v$ and $v'$, there exists a vertex $x\in X_{k-1}^{out}$ such that $d(x,v')\le (\frac{k-1}{k})D$. As $x\in X_{k-1}^{out}$ has in-eccentricity at most $D$, we have $d(v,v')\le d(v,x)+d(x,v')\le (\frac{2k-1}{k})D  < D'$ by the triangle inequality.
  If we output ``NO'' at Step 7, then for any $v,v'\in V$, there exists $x$ and $x'$ such that $d_{G'}(v,x)\le (\frac{\ell}{k})D$, and $d_{G'}(x',v')\le (\frac{k-1-\ell}{k})D$, and we also must have $d_{G'}(x,x')\le (1+\varepsilon/2)D$ as $x\in X_{\ell}^{in}$ and $x'\in X_{k-1-\ell}^{out}$.
  Then $d_{G'}(v,v')\le d_{G'}(v,x)+d_{G'}(x,x')+d_{G'}(x',v') \le (\frac{2k-1}{k}+\varepsilon/2)D < D'$ by the triangle inequality.
  As, shortest path distances are the same in $G'$ as in $G$, we have $d_G(v,v')< D'$ for all $v,v'\in V$, as desired.

  Now we show that the if the diameter is at most $D$, there exists a sequence of nondeterministic choices such that our algorithm outputs ``NO''.
  For $\ell=1,\dots,k-1$, if there exists $X_\ell^{in}$ exists such that $d(v,X_\ell^{in})\le (\frac{\ell}{k})D$, let $X_\ell^{in}$ be that set.
  Similarly if there exists $X_\ell^{out}$ exists such that $d(v,X_\ell^{out})\le (\frac{\ell}{k})D$, let $X_\ell^{out}$ be that set.
  Call such an $X_\ell^{in}$ or $X_\ell^{out}$ good.
  Let $X_\ell^{in}$ or $X_\ell^{out}$ be arbitrary if it is not good, and call such an $X_\ell^{in}$ or $X_\ell^{out}$ \emph{bad}.
  Additionally, let the randomness of $\mathcal{A}_{hopset}$ be such that $E'$ is indeed a $(\beta,\varepsilon/2)$-hopset.

  By Lemma~\ref{lem:struct} with $D_1 = (\ell/k)D$, $D_2=((k-\ell)/k)D$, $M_1=|X_\ell^{in}| = 8m^{1-\ell/k}\beta^{-1+2\ell/k}\log m$, and $M_2 = |X_{k-\ell}^{out}| = 8m^{\ell/k}\beta^{1-2\ell/k}\log m$, we have that either $X_\ell^{in}$ is good or $X_{k-\ell}^{out}$ is good for all $\ell=1,\dots,k-1$.
  Thus, among, $X_1^{in},\dots,X_{k-1}^{in},X_1^{out},\dots,X_{k-1}^{out}$, there are at most $k-1$ \emph{bad} vertex subsets.
  Hence, among the $k$ sets $\{X_{k-1}^{in}\}, \{X_{k-1}^{out}\}$, and $\{X_\ell^{in}, X_{k-1-\ell}^{out}\}$ for $\ell=1,\dots,k-2$, one of these sets of vertex subsets has only \emph{good} vertex subsets.
  First, if $X_{k-1}^{in}$ is good, then choosing this $X_{k-1}^{in}$ as $X_{k-1}^{in}$ in the algorithm causes step 2 to output ``NO'': we indeed will see that $d(v,X_{k-1}^{in})\le (\frac{k-1}{k})D$, and all eccentricities are at most $D$ because the diameter is at most $D$.
  Similarly, if $X_{k-1}^{out}$ is good, then choosing this $X_{k-1}^{out}$ as $X_{k-1}^{out}$ in the algorithm causes step 3 to output ``NO''.
  Finally, if there exists $X_\ell^{in}$ and $X_{k-1-\ell}^{out}$ such that both are good, our algorithm will output ``NO'' in step 7: since $E'$ is a $(\beta,\varepsilon/2)$-hopset, there there exists an at-most-$\beta$ edge path from $x$ to $x'$ in $G'$ for any $x,x'\in X$ of length at most $d_G(x,x')+(\varepsilon/2)D \le (1+\varepsilon/2)D$, so  nondeterministically choosing these paths causes us to output ``NO''.
  This covers all possible cases, so we these nondeterministic choices cause us to output ``NO'', as desired.
\end{proof}

\begin{proof}[Proof of Theorem~\ref{thm:nseth-weight-conj}]
  With Theorem~\ref{thm:nseth-weight-conj-ncon}, apply Lemma~\ref{lem:cgimps} (for deterministic reductions) and apply Lemma~\ref{lem:nunseth} (for randomized reductions).
\end{proof}

\begin{proof}[Proof of Theorem~\ref{thm:nseth-weight}]
  In Theorem~\ref{thm:nseth-weight-conj}, let $k=3$, and let the proposed hopset algorithm be given by Lemma~\ref{lem:hopset-dir} with $\alpha=3/26$, which indeed computes a hopset for all directed weighted graphs.
  It follows that $(5/3+\varepsilon)D/D$-Diameter on directed weighted graphs with time complexity $(m^{4/3+\frac{1/2 - 3/26}{3}} + m^{1 + 4\cdot 3/26})^{1+\delta} = m^{19/13(1+\delta)}$ is not SETH-hard.
  This yields Theorem~\ref{thm:nseth-weight}.
\end{proof}
\begin{proof}[Proof of Theorem~\ref{thm:nseth-weight-intro}]
  Apply Lemma~\ref{lem:gap-to-approx} with $\rho=\infty$, $\alpha=\frac{5}{3}+\frac{\varepsilon}{2}$, $\beta=\frac{\varepsilon}{2}$, and $T=m^{(1+\delta)19/13}$ to Theorem~\ref{thm:nseth-weight}.
\end{proof}

\subsection{Undirected weighted graphs}
Using the same technique, we can prove non-reducibility results for undirected graphs in a larger parameter setting.
\begin{theorem}
  \label{thm:nseth-undir}
  Let $k\ge 2$ be a positive integer and $\varepsilon>0$. 
  Under NSETH (NUNSETH), for any $\delta>0$, a $2-\frac{1}{k}+\varepsilon$-approximation of Diameter on undirected weighted graphs with time complexity $m^{1+1/k+\delta}$ is \emph{not} SETH-hard for deterministic (randomized) reductions.
\end{theorem}
Again, by Lemma~\ref{lem:gap-to-approx}, it suffices to prove the following Theorem.
\begin{theorem}
  \label{thm:nseth-undir-gap}
  Let $D$ be a positive real number, $k\ge 2$ be a positive integer, $\varepsilon>0$ and $D' = (2-\frac{1}{k}+\varepsilon)D$.
  Under NSETH (NUNSETH), for any $\delta>0$, $D'/D$-Diameter on undirected weighted graphs with time complexity $m^{1+1/k+\delta}$ is \emph{not} SETH-hard for deterministic (randomized) reductions.
\end{theorem}
\begin{proof}
  Let $D$ be a real number, $k\ge 2$ be a positive integer, and $D'=(2-\frac{1}{k}+\varepsilon)D$ be as in Theorem~\ref{thm:nseth-undir}.
  We note that Theorem~\ref{thm:nseth-weight-conj} also holds for undirected graphs, when one assumes a hopset construction for unweighted graphs.
  Since one can in fact use the same proof, we omit the proof for brevity.
  By Lemma~\ref{lem:hopset-undir}, there exists an algorithm running in time $T=\tilde O(m^{1+1/k})$, given an \emph{undirected} weighted graph $G$ on $n$ vertices and $m$ edges, computes a set $E'$ of weighted edges such that (1) when $E'$ is added to $G$, all shortest path distances stay the same, and (2) $E'$ forms a $(\beta,\varepsilon/2)$-additive hopset of $G$, where $\beta=O_{k,\varepsilon}(1)$.
  Then, following the proof of Theorem~\ref{thm:nseth-weight-conj}, we have that for all $\delta>0$, $D'/D$-Diameter on undirected weighted graphs with time complexity $(m^{1+1/k}\beta^{1-2/k} + T)^{1+\delta}$ is \emph{not} SETH-hard.
  As $\beta = O_{k,\varepsilon}(1)$ and $T=\tilde O(m^{1+1/k})$, for all $\delta>0$, $D'/D$-Diameter on undirected weighted graphs with time complexity $m^{1+1/k+\delta}$ is \emph{not} SETH-hard, as desired.
\end{proof}
\begin{proof}[Proof of Theorem~\ref{thm:nseth-undir}]
Apply Lemma~\ref{lem:gap-to-approx} to Theorem~\ref{thm:nseth-undir-gap} with $\rho=\infty$, $\alpha=2-\frac{1}{k}+\frac{\varepsilon}{2}$, $\beta=\frac{\varepsilon}{2}$, and $T=m^{1+1/k}$.
\end{proof}

\section{Optimal non-reducibility for directed unweighted graphs}
\label{sec:nseth}

We prove the following theorem, which is a generalization of Theorem~\ref{thm:nseth-intro}.
\begin{theorem}
\label{thm:nseth-intro-1}
  Let $k\ge 2$ be a positive integer.
  Assuming NSETH (NUNSETH), for any $\delta>0$ and $\varepsilon>0$, a $2-\frac{1}{k}+\varepsilon$ approximation of Diameter on directed weighted graphs satisfying $W_{max}/W_{min}\le m^{o(1)}$ with time complexity $m^{1+1/k+\delta}$ is \emph{not} SETH-hard for deterministic (randomized) reductions.
\end{theorem}
As in Section~\ref{sec:nseth-weight}, by Lemma~\ref{lem:gap-to-approx}, to prove non-reducibility of the approximation problem, Theorem~\ref{thm:nseth-intro-1}, it suffices to prove the non-reducibility of the promise problem.
\begin{theorem}
  \label{thm:nseth}
  Let $D$ be a real number, $k\ge 2$ be a positive integer, $\varepsilon>0$ and $D' = (2-\frac{1}{k}+\varepsilon)D$.
  Under NSETH (NUNSETH), for any $\delta>0$, $D'/D$-Diameter on directed weighted graphs with time complexity $m^{1+1/k+\delta}\cdot (W_{max}/W_{min})$ is \emph{not} SETH-hard for deterministic (randomized) reductions.
\end{theorem}

By Lemma~\ref{lem:cgimps} and Lemma~\ref{lem:nunseth}, to prove Theorem~\ref{thm:nseth}, it suffices to prove the following result.
\begin{theorem}
  Let $D$ be a real number, $k$ be a positive integer, $\varepsilon>0$ and $D' = (2-\frac{1}{k}+\varepsilon)D$.
  $D'/D$-Diameter on directed weighted graphs with $W_{min}>0$ is in $\nconplus[\tilde O(m^{1+1/k}\cdot (W_{max}/W_{min}))]$. 
\label{thm:ncon}
\end{theorem}

We prove Theorem~\ref{thm:ncon} even when $D$ can depend on $m$.
Since $D$ can depend on $m$, we may without loss of generality assume $W_{min}=1$, by scaling the edge weights by $1/W_{min}$ so that the new parameters $D$ and $D'$ are $D/W_{min}$ and $D'/W_{min}$, respectively, and the quantity $W_{max}/W_{min}$ remains unchanged.
We now show that $D'/D$-Diameter on directed weighted graphs with $W_{min}=1$ is in $\nconplus[m^{1+1/k+o(1)}W_{max}]$.

Throughout this section, fix a positive integer $k$ and a positive real number $D$, and a real number $\varepsilon>0$.
Let $r=\varepsilon D/2$. 
Throughout, $G=(V,E)$ is a graph with minimum edge-weight 1 and maximum edge-weight $W_{max}$.
Suppose we are looking to distinguish between graphs of diameter at most $D$ and graphs of diameter at least $D'=(2-\frac{1}{k}+\varepsilon)D$.

For $\ell=0,\dots,k$, call a subset $X$ of vertices \emph{$(\ell,r,k,D)$-out} (or simply $\ell$-out when $r,k,D$ are understood) if $|X|\le 8m^{1-\ell/k}W_{max}\log m$ and call it \emph{good $(\ell,r,k,D)$-out} (good $\ell$-out) if additionally $d(X,v) \le \frac{\ell}{k}D+r$ for all vertices $v$.
For $\ell=0,\dots,k$, call a subset $X$ of vertices \emph{$(\ell,r,k,D)$-in} (or simply $\ell$-in when $r,k,D$ are understood) if $|X|\le 8m^{1-\ell/k}r^{-1}\log m$ and call it \emph{good $(\ell,r,k,D)$-in} (good $\ell$-in) if additionally $d(v,X) \le \frac{\ell}{k}D$ for all vertices $v$.

\begin{lemma}
  For all $\ell=1,\dots,k-1$, there either exists a set $X$ that is good $\ell$-out or a set $X'$ that is good $k-\ell$-in.
\label{lem:struct-1}
\end{lemma}
\begin{proof}
  Suppose there exists a vertex $v$ with $|E_{\le \frac{k-\ell}{k}D}^{out}(v)| \le m^{1-\ell/k}r$.
  Let $U$ denote the vertices incident to $E^*\defeq E_{\le \frac{k-\ell}{k}D}^{out}(v)$.
  For every vertex $u$, let $Z_u$ denote the set of vertices $w$ such that $d(w,u)\le \frac{\ell}{k}D + r$ and such that $w$ is incident to an edge in $E^*$.
  
  We claim that each $Z_u$ has size at least $r$. 
  consider a shortest path $v=z_0,z_1,\dots,z_s=u$ from $v$ to $u$.
  Let $i$ denote the largest index such that $d(z_i,u)> \frac{\ell}{k}D + r$ and $j$ denote the smallest index such that $d(v,z_j) > \frac{k-\ell}{k}D$.
  We must have $j> i$ or else $d(v,u) \ge d(v,z_j) + d(z_i,u) > D+r$, contradicting diameter $D$.
  By minimality of $j$, we have $d(v,z_{j'}) < \frac{k-\ell}{k}D$ for $j'\le j-1$.
  In particular $z_{i+1},z_{i+2},\dots,z_{j-1}$ is in $N_{\le \frac{k-\ell}{k}D}^{out}(v)$, so $z_{j'-1}z_{j'}$ is in $E^*$ for $j'\le j$ and thus $z_{j'}\in U$ for $j'\le j$.
  By maximality of $i$, we have $d(z_{j'},u)\le \frac{\ell}{k}D+r$ for $j'>i$, so $z_{i+1},\dots,z_j$ are in $Z_u$ and $|Z_u|\ge j-i$. 
  We note that
  \begin{align}
    D \ge d(v,u) &= d(v,z_j) + d(z_i,u) - d(z_i,z_j)  \nonumber\\
    &> \frac{k-\ell}{k}D + \frac{\ell}{k}D+r - d(z_i,z_j) = D + r - d(z_i,d_j).
  \end{align}
  We conclude that $d(z_i,z_j) > r$, so in particular $|Z_u|\ge j-i > r/W_{max}$. 

  By Lemma~\ref{lem:select} with universe as vertex set $U$, which has size at most $M\defeq m^{1-\ell/k}r$, with $n$ sets $Z_u\subset U$ of size at least $K\defeq r/W_{max}$, we have that there exists a vertex subset $X\subset U$ of size at most $2M/K \log n = 2m^{1-\ell/k}W_{max}\log n$ that intersects each $Z_u$.
  This implies that $d(X,u)\le \frac{\ell}{k}D+r$ for all vertices $u\in V$.
  We conclude that $X$ is good $\ell$-out.

  Now suppose to the contrary that $|E_{\le \frac{k-\ell}{k}D}^{out}(v)| \ge m^{1-\ell/k}r$ for all vertices $v$.
  By Lemma~\ref{lem:select} with universe as the set $E$ of all edges, which has size $m$, with $n$ sets $E_{\le\frac{k-\ell}{k}D}^{out}(v)$ of size at least $m^{1-\ell/k}r$, we have that there exists a edge subset $F$ of size $2M/K\log n = 2m^{\ell/k}r^{-1}\log n$.
  Let $X'$ be the vertices incident to $F$, so that $|X'|\le 4m^{\ell/k}r^{-1}\log m$.
  Then, $X'$ contains a vertex inside each $N_{\le\frac{k-\ell}{k}D}^{out}(v)$, so $d(v,X')\le \frac{k-\ell}{k}D$ for all vertices $v$.
  We conclude $X'$ is good $(k-\ell)$-in.
\end{proof}
\begin{lemma}
  \label{lem:struct-2}
  Let $G$ be a graph with a good $(k-\ell)$-in set $X$ and a good $(\ell-1)$-out set $X'$. 
  Suppose further that $d(x,x')\le D$ for all $x\in X$ and $x'\in X'$.
  Then $G$ has diameter less than $D'$.
\end{lemma}
\begin{proof}
  Let $v$ and $v'$ be vertices in $G$.
  As $X$ is good $(k-\ell)$-in, we have $d(v,X)\le \frac{k-\ell}{k}D$.
  As $X'$ is good $(\ell-1)$-out, we have $d(X',v')\le \frac{\ell-1}{k}D + r$.
  Thus, there exists $x\in X$ and $x'\in X'$ such that $d(v,x)\le \frac{k-\ell}{k}D$ and $d(x',v')\le \frac{\ell-1}{k}D + r$.
  By the triangle inequality, we have
  \begin{align}
    d(v,v')
    \le d(v,x) + d(x,x')+d(x',v') 
    &\le \frac{k-\ell}{k}D + D + \frac{\ell-1}{k}D + r  \nonumber\\
    &= \left(2 - \frac{1}{k}\right)D + r 
    < D'
  \label{}
  \end{align}
  This holds for all $v,v'$, so the diameter is at most $(2-\frac{1}{k}+\varepsilon)D$.
\end{proof}
\begin{lemma}
  Let $G$ be a graph with a good $(k-1)$-out set $X$.
  Suppose further that $d(v,x)\le D$ for all vertices $x\in X$ and all vertices $v\in V$.
  Then $G$ has diameter less than $D'$.
\label{lem:struct-3}
\end{lemma}
\begin{proof}
  For any vertices $v,v'\in V$, there exists a vertex $x\in X$ such that $d(x,v') \le \frac{k-1}{k}D+r$.
  Then by the triangle inequality, $d(v,v') \le d(v,x)+d(x,v') \le D + \frac{k-1}{k}D + r < D'$.
  This holds for all $v,v'$, so the diameter is less than $D'$.
\end{proof}
\begin{lemma}
  For any $\ell=1,\dots,k-1$, given a $\ell$-out ($\ell$-in) set $X$ it is possible to check in time $\tilde O(m)$ whether $X$ is good $\ell$-out ($\ell$-in).
\label{lem:good-1}
\end{lemma}
\begin{proof}
  Run multi-source shortest path from $X$ to compute all $d(X,v)$ or all $d(v,X)$.
\end{proof}

We now present our algorithm.
\begin{proof}[Proof of Theorem~\ref{thm:ncon}]
  Let $G$ be the input graph.
  Again recall that $D'/D$-Diameter is a promise problem $(\Pi_{YES},\Pi_{NO})$ where the instances in $\Pi_{YES}$ are graphs of diameter at least $D'$, and the instances in $\Pi_{NO}$ are graphs of diameter at most $D$.
  Let $\mathcal{A}_N$ be the nondeterministic algorithm that guesses a vertex, outputs ``YES'' if the eccentricity is greater than $D$, and ``NO'' otherwise.
  For any graph \emph{not} in $\Pi_{NO}$, i.e., any graph with diameter greater than $D$, there exist nondeterministic choices such that $\mathcal{A}_N$ outputs ``YES''. Furthermore, for graphs in $\Pi_{NO}$, $\mathcal{A}_N$ always outputs ``NO''.
  As $\mathcal{A}_N$ runs in time $m^{1+o(1)}$, we have that properties 1 and 3 of Definition~\ref{def:nconplus} for $\nconplus[m^{1+1/k+o(1)}W_{max}]$ are satisfied.

  It now remains to construct a nondeterministic algorithm $\mathcal{A}_{coN}$ satisfying property 2 of Definition~\ref{def:nconplus} for $\nconplus[m^{1+1/k+o(1)}W_{max}]$.
  Define $\mathcal{A}_{coN}$ to be the following algorithm.
\begin{enumerate}
\item For each $\ell=0,\dots,k$, nondeterministically choose a set $X_\ell^{out}$ and $X_{\ell}^{in}$, such that $X_\ell^{out}$ is $\ell$-out and $X_{\ell}^{in}$ is $\ell$-in.
\item For each $\ell=0,\dots,k$, check if $X_{\ell}^{out}$ is good $\ell$-out, and check if $X_{\ell}^{in}$ is good $\ell$-in.
\item If $X_{k-1}^{out}$ is good $(k-1)$-out, then run shortest path from each element of $X_{k-1}^{out}$. If each vertex in $X_{k-1}^{out}$ has in-eccentricity at most $D$, output ``NO''. 
\item Let $\ell\in\{0,\dots,k-1\}$ be such that $X_{\ell-1}^{out}$ is good $(\ell-1)$-out and $X_{k-\ell}^{in}$ is good $(k-\ell)$-in. If no such $\ell$ exists, then output ``YES''.
\item For each $x\in X_{\ell-1}^{out}$ and each $x'\in X_{k-\ell}^{in}$, nondeterministically choose up to $D$ vertices that form a path from $x$ to $x'$.
\item If each path has length at most $D$, output ``NO''. Otherwise, output ``YES''.
\end{enumerate}

\textbf{Runtime.}
The first step takes time $\tilde O(m)$ to choose the $X_\ell^{out}$.
The second step takes time $\tilde O(m)$, as it takes $\tilde O(m)$ time to check if $X_\ell^{out}$ is good $\ell$-out or $X_\ell^{in}$ is good $\ell$-in by Lemma~\ref{lem:good-1}.
The third step takes time $\tilde O(m\cdot |X_{k-1}^{out}|)\le \tilde O(m^{1+1/k}W_{max})$.
The fourth step takes constant time (we assume $k$ is constant).
The  fifth step takes time $O(|X_{\ell-1}^{out}|\cdot |X_{k-\ell}^{in}|\cdot D) = \tilde O(m^{1-(\ell-1)/k}\cdot m^{\ell/k}/(\varepsilon D)\cdot D) = \tilde O_\varepsilon(m^{1+1/k})$.
The total running time is thus $\tilde O_\varepsilon(m^{1+1/k}W_{max})$

\textbf{Correctness.}
We first show that, if we output ``NO'', the diameter must be less than $D'$.
Suppose we output ``NO'' at step 3.
Then each $x\in X_{k-1}^{out}$ satisfies $d(v,x)\le D$ for all $v\in V$.
Then, by Lemma~\ref{lem:struct-3}, the diameter is less than $D'$.

Now suppose we output ``NO'' at step 4, and let $\ell$ be the parameter chosen in Step 4.
Then we know $X_{\ell-1}^{out}$ is good $(\ell-1)$-out and $X_{k-\ell}^{in}$ is good $(k-\ell)$-in.
Since we output ``NO'', we also have that $d(x,x')\le D$ for each $x\in X_{\ell-1}^{out}$ and each $x'\in X_{k-\ell}^{in}$.
Setting $X = X_{k-\ell}^{in}$ and $X' = X_{\ell-1}^{out}$ in Lemma~\ref{lem:struct-2}, we have that the diameter is less than $(2-\frac{1}{k}+\varepsilon)D = D'$, as desired.

We now show that, if the diameter is at most $D$, there exists some nondeterministic choices such that we output ``NO''.
For each $\ell=0,\dots,k$, if there exists a good $\ell$-in set, let $X_{\ell}^{in}$ be that set, and if there exists a good $\ell$-out set, let $X_{\ell}^{out}$ be that set.
By Lemma~\ref{lem:struct-1}, for each pair of sets $(X_1^{out}, X_{k-1}^{in}), \dots,(X_{k-1}^{out}, X_1^{in})$, at least one must satisfy the corresponding goodness property.
As the set of all vertices is good 0-out by definition, we have $X_0^{out}$ is good 0-out.
Hence, among, $X_0^{out},X_1^{out},\dots,X_{k-1}^{out},X_1^{in},\dots,X_{k-1}^{in}$, there are at least $k$ sets that satisfy their corresponding goodness properties and at most $k-1$ sets that do not.

If $X_{k-1}^{out}$ is good $(k-1)$-out, then we will output ``NO'' at step 3, as the diameter is at most $D$.
Furthermore, if, for any $\ell$, $X_{k-\ell}^{out}$ is good $(k-\ell)$-out and $X_{\ell-1}^{in}$ is good $(\ell-1)$-in, then we output ``NO'' because we choose the corresponding parameter $\ell$ in Step 4, and the path of length $D$ between every pair of vertices $x\in X_{k-\ell}^{out}$ and $x'\in X_{\ell-1}^{in}$ exists.
Thus, to not output ``NO'', one set from each of $\{X_{k-1}^{out}\}, \{X_{k-2}^{out},X_1^{in}\},\dots,\{X_0^{out}, X_1^{in}\}$ must not satisfying the corresponding goodness properties, so $k$ sets cannot satisfy their goodness properties, which is a contradiction of the previous paragraph.
Thus, there exist nondeterministic choices such that we output ``NO'', as desired.
\end{proof}
\begin{proof}[Proof of Theorem~\ref{thm:nseth}]
  With Theorem~\ref{thm:ncon}, apply Lemma~\ref{lem:cgimps} (for deterministic reductions) and apply Lemma~\ref{lem:nunseth} (for randomized reductions).
\end{proof}
\begin{proof}[Proof of Theorem~\ref{thm:nseth-intro} and Theorem~\ref{thm:nseth-intro-1}]
  Apply Lemma~\ref{lem:gap-to-approx} with $\rho=m^{o(1)}$, $\alpha=2-\frac{1}{k}+\frac{\varepsilon}{2}$, $\beta=\frac{\varepsilon}{2}$, and $T=m^{1+1/k+\delta}$ to Theorem~\ref{thm:nseth}.
\end{proof}
\begin{proof}[Proof of Corollary~\ref{cor:nseth-intro}]
  First suppose $\floor{1/\varepsilon}\le \floor{1/\delta}$.
  Let $k = \floor{1/\varepsilon}+1$.
  Then $2-\varepsilon < 2 - \frac{1}{k}$ so there exists a positive $\varepsilon'>0$ such that $2-\varepsilon = 2 - \frac{1}{k}-\varepsilon'$.
  Furthermore, $k-1 \le \floor{1/\delta}\le 1/\delta$, so $1+\delta\le 1+1/(k-1)$.
  In Theorem~\ref{thm:2eps}, it is shown that a $2-\frac{1}{k}-\varepsilon'=2-\varepsilon$ approximation of the Diameter with time complexity $n^{1+1/(k-1)}$ is SETH-hard.
  Hence, since decreasing time complexity preserves SETH-hardness, a $2-\varepsilon$-approximation of the Diameter with time complexity $n^{1+\delta}$ is SETH-hard, as desired.

  Now suppose $\floor{1/\varepsilon} > \floor{1/\delta}$ and $1/\varepsilon\neq \floor{1/\delta}+1$.
  Let $k=\floor{1/\delta}+1$.
  Since $1/\varepsilon\ge \floor{1/\varepsilon}\ge \floor{1/\delta}+1=k$, and $1/\varepsilon\neq \floor{1/\delta}+1$, we have $1/\varepsilon > k$.
  Thus, there exists an $\varepsilon'>0$ such that $2-\frac{1}{k} + \varepsilon' = 2-\varepsilon$.
  Additionally, $1/k < \delta$, so there exists $\delta'$ such that $1+1/k+\delta'=1+\delta$.
  Theorem~\ref{thm:nseth-intro} proves that, under NSETH (NUNSETH), a $2-\frac{1}{k}+\varepsilon'$ approximation of Diameter in directed unweighted graphs with time complexity $m^{1+1/k+\delta'}$ is not SETH-hard for deterministic (randomized) reductions, so a $2-\varepsilon$ approximation of Diameter in directed unweighted graphs with time complexity $m^{1+\delta}$ is not SETH-hard for deterministic (randomized) reductions, as desired.
\end{proof}

\section{Lower bound for undirected unweighted graphs}
\label{sec:53}
In this section, we prove the following result, implying there is no $5/3-\varepsilon$ approximation of the diameter of an undirected unweighted graph in near-linear time.
\begin{theorem}
   Assuming SETH, for all $\varepsilon>0$ a $(\frac{5}{3}-\varepsilon)$-approximation of Diameter in \emph{unweighted, undirected} graphs on $n$ vertices needs $n^{3/2-o(1)}$ time.
   \label{thm:53}
\end{theorem}
\begin{proof}
Start with a Single-Set 3-OV instance $\Phi$ given by a set $A\subset \{0,1\}^d$ with $|A|=\tilde{n}$ and $d = c\log n_{NOV}$.
We may add the all-1s vector to $A$ without loss of generality, as this does not change whether there is an OV solution or not. 
We construct a graph with $\tilde O(\tilde{n}^2)$ vertices and edges from the 3-OV instance such that (1) if $\Phi$ has no solution, any two vertices are at distance 3, and (2) if $\Phi$ has a solution, then there exists two vertices at distance 5.
Any $(5/3-\varepsilon)$-approximation for Diameter distinguishes between graphs of diameter 3 and 5. 
Since solving $\Phi$ needs $\tilde{n}^{3-o(1)}$ time under SETH, a $5/3-\varepsilon$ approximation of diameter needs $n^{3/2-o(1)}$ time under SETH.
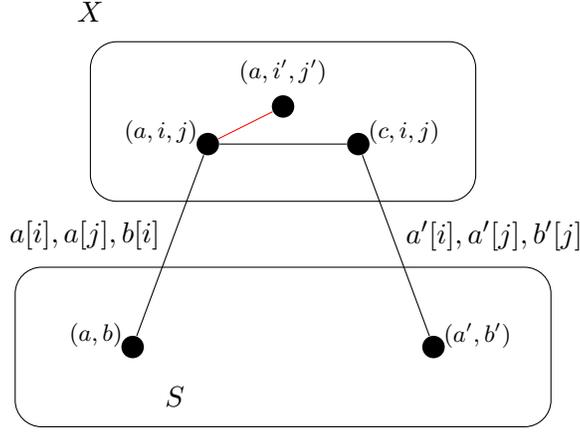
\begin{figure}
\begin{center}
\begin{tikzpicture}
  \node[inner sep=30] (a) at (-2.5,-1) {};
  \node[inner sep=30] (b) at (2.5,-1) {};
  \node[label={[anchor=south]$S$}] (f) at (a.south east) {};
  \draw[rounded corners=10] (a.north west) rectangle (b.south east);
  \node[inner sep=30] (d) at (-1.5,2) {};
  \node[inner sep=30] (e) at (1.5,2) {};
  \node[label={[anchor=south]$X$}] (f) at (d.north west) {};
  \draw[rounded corners=10] (d.north west) rectangle (e.south east);
  \node[fill, circle, inner sep=3,label={\footnotesize$(a,i',j')$}] (a2ij) at (0,2.2) {};
  \node[fill, circle, inner sep=3,label={[anchor=east]\footnotesize$(a,i,j)$}] (aij) at (-1,1.7) {};
  \node[fill, circle, inner sep=3,label={[anchor=west]\footnotesize$(c,i,j)$}] (cij) at (1,1.7)  {};
  \node[fill, circle, inner sep=3,label={[anchor=east]\footnotesize$(a,b)$}] (ab) at (-2,-1)  {};
  \node[fill, circle, inner sep=3,label={[anchor=west]\footnotesize$(a',b')$}] (cb) at (2,-1)   {};
  \draw (ab) --  node[pos=0.5, label={[anchor=east]$a[i],a[j],b[i]$}]{}(aij);
  \draw (aij) -- (cij);
  \draw (cb) --  node[pos=0.5, label={[anchor=west]$a'[i],a'[j],b'[j]$}]{}(cij);
  \draw[red] (aij) -- (a2ij);
\end{tikzpicture}
\end{center}
\caption{5 vs. 3 diameter instance. The coordinates $a[i]$ along the edges must be 1 for the edge to exist.}
\label{fig:1}
\end{figure}
\paragraph{Construction of the graph}
The graph $G$ is illustrated in Figure~\ref{fig:1} and constructed as follows. 
The vertex set $S\cup X$ is defined on
\begin{align}
  S &= A^2,  \nonumber\\
  X &= \{(a,i,j)\in A\times [d]^2:a[i]=a[j]=1\}
\end{align}
Throughout, we identify tuples $(a,b)\in A^2$ and $(a,i,j)\in A\times [d]^2$ with vertices of $G$, 
Throughout we denote vertices in $S$ and $X$ by $(a,b)_S$ and $(a,i,j)_X$, respectively.
The (undirected unweighted) edges are all of the following.
\begin{itemize}
\item Edge between $(a,b)_S$ and $(a,i,j)_X$ if $(a,i,j)_X$ exists and $b[i]=1$.
\item Edge between $(a,b)_S$ and $(a,i,j)_X$ if $(a,i,j)_X$ exists and $b[j]=1$.
\item Edge between $(a,i,j)_X$ and $(b,i,j)_X$ always if vertices exist.
\item Edge between $(a,i,j)_X$ and $(a,i',j')_X$ always if vertices exist.
\end{itemize}
Note that each vertex of $S$ has $O(d^2)$ neighbors, each vertex of $X$ has $O(\tilde{n})$ neighbors, and each vertex of $Y$ has $O(\tilde{n})$ neighbors.
The total number of edges and vertices is thus $O(\tilde{n}^2d^2)=\tilde O(\tilde{n}^2)$.
We now show that this construction has diameter 3 when $\Phi$ has no solution and diameter at least 5 when $\Phi$ has a solution.

\paragraph{3-OV no solution}
Assume that the 3-OV instance has no solution, so that no three (or two) vectors are orthogonal.
We show that any pair of vertices have distance at most 3, by casework on which of $S,X$ the two vertices are in.
\begin{itemize}
\item \textbf{Both vertices are in $S$:} Let the vertices be $(a,b)_S$ and $(c,d)_S$. As there is no 3-OV solution, there exists indices $i$ and $j$ in $[d]$ such that $a[i]=b[i]=c[i]$ and $a[j]=c[j]=d[j]$. Then $(a,b)_S-(a,i,j)_X-(c,i,j)_X-(c,d)_S$ is a valid path.

\item \textbf{One vertex is in $S$ and the other vertex is in $X$:} Let the vertices be $(a,b)_S$ and $(c,i,j)_X$: As there is no 3-OV solution, there exists an index $i'$ such that $a[i']=b[i']=c[i']=1$. Then $(a,b)_S-(a,i',i')_X-(c,i',i')_X-(c,i,j)_X$ is a valid path.

\item \textbf{Both vertices are in $X$:} Let the vertices be $(a,i,j)_X$ and $(c,i',j')_X$. As there is no 3-OV solution, there exists an index $i''$ such that $a[i''] = c[i'']$. Then $(a,i,j)_X-(a,i'',i'')_X-(c,i'',i'')_Y-(c,i',j')_X$ is a valid path.
\end{itemize}

\paragraph{3-OV has solution}
Now assume that the 3-OV instance has a solution.
That is, assume there exists $a,b,c\in A$ such that $a[i]\cdot b[i]\cdot c[i] = 0$ for all $i$.
We show there are no paths of length at most 4 from $(a,b)_S$ to $(c,b)_S$.
Any such path must use an $X-X$ edge or else the first entry $a$ of the vertex's tuple does not change.
Because of this, the path cannot revisit the set $S$, as it would otherwise need at least 5 edges.

First, suppose that the path does not use an \emph{index-changing} edge, namely an edge of the form $(a,i,j)_X-(a,i',j')_X$.
The path must be $(a,b)_S-(a,i,j)_X-(c,i,j)_X-(c,b)_S$ or $(a,b)_S-(a,i,j)_X-(v,i,j)_X-(c,i,j)_X-(c,b)_S$ for some indices $i,j\in[d]$ and some vector $v$.
In either case, the existence of vertex $(a,i,j)_X$ requires that $a[i]=a[j]=1$, and the existence of vertex $(c,i,j)_X$ requires that $c[i]=c[j]=1$.
The first edge requires that at least one of the coordinates $b[i]$ or $b[j]$ is 1. 
Thus, at least one of $a[i]=b[i]=c[i]=1$ or $a[j]=b[j]=c[j]=1$ holds, contradicting orthogonality of $a,b,c$.

Now suppose the path uses an index-changing edge.
The index-changing edge must either be the second edge or the third edge.
These cases are symmetric to each other so it suffices to consider only one.
If the index-changing edge is the second edge, the path must be $(a,b)_S-(a,i,j)_X-(a,i',j')_X-(c,i',j')_X-(c,b)$ for some indices $i,j,i',j'\in[d]$.
The existence of the vertex $(a,i',j')_X$ implies that $a[i']=a[j']=1$, and the existence of the vertex $(c,i',j')_X$ implies that $c[i']=c[j']=1$.
The last edge implies that $b[i']=1$ or $b[j']=1$.
Thus, at least one of $a[i']=b[i']=c[i']=1$ or $a[j']=b[j']=c[j']=1$ holds, contradicting orthogonality of $a,b,c$.

This shows that $(a,b)_S$ and $(c,b)_S$ are at distance at least 5, completing the proof.
\end{proof}

\section{Acknowledgements}

The author would like to thank Aviad Rubinstein for many helpful discussions, guidance, encouragement, and feedback on this writeup. 
The author would like to thank Mary Wootters for helpful discussions and feedback on this writeup.
The author would like to thank Joshua Brakensiek for helpful discussions.
The author would like to thank Thuy Duong Vuong and Arun Jambulapati for helpful discussions on hopsets, and Arun Jambulapati for the reference \cite{CaoFR20}.
The author would like to thank Nairen Cao, Jeremy T. Fineman, and Katina Russell for helpful discussions on their work \cite{CaoFR19,CaoFR20}.
The author would like to thank anonymous reviewers for helpful feedback.

\newcommand{\etalchar}[1]{$^{#1}$}

\appendix

\section{Hopset results}
\label{app:hopset}

\subsection{Undirected hopset}
Here, we prove Lemma~\ref{lem:hopset-undir}.
Let $M$ be a positive integer.
In a graph $G$, let $\tilde N_M(v)$ denote the set of vertices visited when Dijkstra is run from a vertex $v$ until $2M$ edges are visited, and let $\tilde d_M(v)$ denote the distance from $v$ to the furthest vertex in $\tilde N_M(v)$.
Let $\tilde E_M(v)$ denote the set of vertices incident to $\tilde N_M(v)$. 
By definition, we have $|\tilde E_M(v)|\ge M$ (each edge can be visited at most twice).
\begin{lemma}
  Let $G=(V,E)$ be an undirected weighted graph, and $M$ be an integer.
  For each $v$, in time $O(M\log M)$ we can compute $\tilde N_M(v)$ and $d(v,u)$ for all $u\in \tilde N_M(v)$.
\label{lem:app-hopset-1}
\end{lemma}
\begin{proof}
  Run Dijkstra from $v$ until $2M$ edges are visited.
  Then every vertex in $\tilde N_M(v)$ is visited by definition of $\tilde N_M(v)$, so we know $d(v,u)$ for all $u\in \tilde N_M(v)$.
\end{proof}
\begin{proof}[Proof of Lemma~\ref{lem:hopset-undir}]
  Any $(\beta, \varepsilon)$-additive-hopset is also a $(\beta,\varepsilon')$-additive-hopset for $\varepsilon'>\varepsilon$, so we may assume without loss of generality that $\varepsilon\le 1/100$.
  Let $k = \ceil{1/\delta}$ and assume without loss of generality that $n$ and $m$ are sufficiently large in terms of $k$ (all guarantees suppress dependencies on $\delta$).
  For $i=1,\dots,k$, let $M_i = m^{(k+1-i)/k}$.
  Let $D$ denote the diameter of the graph.
  We use the following algorithm.
  \begin{enumerate}
  \item For $i=1,\dots,k$, let $S_{i}$ be the vertices incident to $4m^{i/k}\log m$ uniformly random edges.
  \item For $i=1,\dots,k$, for each $v\in S_i$, run the algorithm in Lemma~\ref{lem:app-hopset-1} for $M=M_i$. For each $u\in \tilde N_{M_i}(v)$, add an edge of weight $d_G(v,u)$ from $v$ to $u$ to $E'$.
  \end{enumerate}
  Note that, for $i=1$, we simply end up running Dijkstra from each vertex.

  \textbf{Runtime.}
  For $i=1,\dots,k$, each Dijkstra takes $O(M_i\log M_i)$ time, so the running time for that step is $O(m^{i/k}\log m\cdot m^{(k+1-i)/k}\log m) = \tilde O(m^{1+1/k})$.
  Thus the total running time is $\tilde O(m^{1+1/k})$.

  \textbf{Correctness.}
  First, by Lemma~\ref{lem:app-hopset-1}, we know that for each $i=1,\dots,k$, each $v\in S_i$, and each $u\in \tilde N_{M_i}(v)$, the distance $d(v,u)$ is accurately computed, so the edge added from $v$ to $u$ does not decrease the shortest path from $v$ to $u$ and hence does not decrease any shortest path.
  Thus, we have the first guarantee that adding $E'$ preserves all shortest paths.

  Now we show the second guarantee.
  By same reasoning as in Lemma~\ref{lem:select}, for all $i=1,\dots,k$, with probability $1-1/n$, because $4m^{i/k}\log m\cdot M_{i+1}> 2m\log n$, the edges used to generate $S_i$ intersect each of the $n$ sets $\tilde E_{M_{i+1}}(v)$ for all vertices $v$.
  Hence, $S_i$ intersects each $\tilde N_{M_{i+1}}(v)$.
  By the union bound, with probability $1-k/n > 0$, for all $i=1,\dots,k$ and all vertices $v$, the set $S_i$ intersects each set $\tilde N_{M_{i+1}}(v)$.
  Fix the choice of randomness such that this holds.
  We now show that $E'$ is a $(O_{\delta,\varepsilon}(1), 6\varepsilon)$ hopset, which, after reparameterizing $\varepsilon$, shows that we can obtain a $(O_{\delta,\varepsilon}(1), \varepsilon)$-additive-hopset.

  Fix two vertices $v$ and $v'$ and a shortest path $v=v_0,\dots,v_L=v'$ between them.
  We show how to construct a path in $G+E'$ from $v$ to $v'$ of length at most $d(v,v') + \varepsilon D$ using at most $O_{\delta,\varepsilon}(1)$ vertices.
  We prove the following claim.
  \begin{claim}
    \label{cl:app-hopset}
    For every $\ell\in\{0,\dots,L\}$, there exists $\ell'>\ell$ and an at-most $k+1$ edge path from $v_\ell$ to $v_{\ell'}$ such that either (1) $d(v_{\ell},v_{\ell'}) > \varepsilon^k D$ and the path has length at most $(1+4\varepsilon)d(v_{\ell},v_{\ell'})$ or (2) $\ell'=L$ and the path has length at most $d(v_\ell,v_{\ell'}) + 2\varepsilon D$.
  \end{claim}
  \begin{proof}
    Fix $\ell$.
    For each vertex $u$ and $i=2,\dots,k$, let $f_i(u)$ be an arbitrary vertex in $S_{i-1}\cap \tilde N_{M_{i}}(v)$.
    Such a vertex exists by construction of $S_{i-1}$.
    Furthermore, if $u\in S_i$, by the definition of the algorithm, there always exists an edge from $u$ to $f_i(u)$.
    Say a vertex $u\in S_i$ is \emph{$i$-nonexpanding} if $\tilde d_{M_i}(u)\le \varepsilon^{i-1}D$ and \emph{$i$-expanding} otherwise.
    Note that if $u$ is $i$-nonexpanding, then every vertex $u'\in \tilde N_{M_i}(u)$ satisfies $d(u,u')\le \varepsilon^{i-1}D$.
    On the other hand, if $u$ is $i$-expanding, any vertex $u'\notin \tilde N_{M_i}(u)$ satisfies $d(u,u') > \varepsilon^{i-1}D$.

    Let $u_k=v_\ell$.
    For $i=k,\dots,2$, let $u_{i-1} = f_i(u_i)$.
    Then $u_i\in S_i$ for $i=1,\dots,k$ by definition of $f_i$, and there always exists an edge from $u_i$ to $u_{i-1}$ for $i=2,\dots,k$.
    Clearly $u_1$ is $1$-expanding, and the length of all shortest paths is bounded above by the diameter $D$, so $\tilde d_{M_1}(u)\le D$.
    Let $j$ be the largest index such that $u_j$ is $j$-expanding.
    Then $u_i$ is $i$-nonexpanding for $j<i\le k$, so by the triangle inequality
    \begin{align}
      \label{eq:app-hopset-1}
      d(v_\ell,u_j)
      \le \sum_{i=j+1}^{k} d(u_i,u_{i-1}) 
      = \sum_{i=j+1}^{k} d(u_i, f(u_i)) 
      \le \sum_{i=j+1}^{k} \varepsilon^{i-1}D
      < 2\varepsilon^{j}D.
    \end{align}
    Note this inequality even holds for $j=k$.

    First suppose $v_L\in \tilde N_{M_j}(u_j)$.
    Then the path $u_k,\dots,u_j,v_L$ is a path of at most $k$ edges from $v_\ell=u_k$ to $v_L$ and has length at most
    \begin{align}
      \label{eq:app-hopset-2}
      d(v_\ell,u_j) + d(u_j,v_{L})
      &\le 2d(v_\ell,u_j) + d(v_\ell,v_{L}) 
      < 2\varepsilon D + d(v_\ell,v_{L}).
    \end{align}
    In the first inequality, we used the triangle inequality, and in the last inequality we used \eqref{eq:app-hopset-1} and that $j\ge 1$.

    Now suppose $v_L\notin \tilde N_{M_j}(u_j)$.
    Let $\ell'$ denote the largest integer such that $v_{\ell'-1}\in \tilde N_{M_j}(u_j)$.
    Since we assume $v_L\notin \tilde N_{M_j}(u_j)$, we have $\ell'\le L$.
    By \eqref{eq:app-hopset-1}, we have $d(v_\ell,u_j) = d(u_k,u_j) < \varepsilon^{j-1}D$, so $v_\ell\in \tilde N_{M_j}(u_j)$ and thus $\ell'>\ell$.
    Since $u_j$ is $j$-expanding, and $v_{\ell'+1}\notin \tilde N_{M_j}(u_j)$, we have $d(u_j, v_{\ell'})>\varepsilon^{j-1}D$.
    Thus, by the triangle inequality and \eqref{eq:app-hopset-1}, we have
    \begin{align}
      \label{eq:app-hopset-3}
      d(v_\ell,v_{\ell'}) > d(u_j,v_{\ell'}) - d(u_j,v_\ell) > (\varepsilon^{j-1} - 2\varepsilon^j)D.
    \end{align}
    and in particular $d(v_\ell,v_{\ell'}) > 2\varepsilon^k D$.
    On the other hand, the path $u_k,u_{k-1},\dots,u_j, v_{\ell'-1},v_{\ell'}$ is an at most $k+1$ edge path from $v_\ell$ to $v_{\ell'}$, and the length is at most
    \begin{align}
      d(v_\ell,u_j) &+ d(u_j,v_{\ell'-1}) + d(v_{\ell'-1},v_{\ell'})  \nonumber\\
      &\le 2d(u_k,u_j) + d(v_\ell,v_{\ell'-1})  + d(v_{\ell'-1},v_{\ell'})   & \text{(triangle-ineq.)}\nonumber\\
      &= 2d(u_k,u_j) + d(v_\ell,v_{\ell'})  \nonumber\\
      &< 2\varepsilon^jD + d(v_\ell,v_{\ell'})  & \text{(by \eqref{eq:app-hopset-1})} \nonumber\\
      &< \left(1 + \frac{2\varepsilon^j}{\varepsilon^{j-1}-2\varepsilon^j}\right) d(v_\ell,v_{\ell'})  & \text{(by \eqref{eq:app-hopset-3})}\nonumber\\
      &< (1 + 4\varepsilon) d(v_\ell,v_{\ell'}) & \text{($\varepsilon$ sufficiently small)}
    \end{align}
    as desired.
  \end{proof}
  Now we finish the proof.
  By repeatedly applying Claim~\ref{cl:app-hopset}, we obtain indices $0=\ell_0,\ell_1,\dots,\ell_p=\ell$, such that for $0\le p'<p-1$,  there exists an at-most-$k+1$ edge path from $v_{\ell_{p'}}$ to $v_{\ell_{p'+1}}$ of length at most $(1+\varepsilon)d(v_{\ell_{p'}}, v_{\ell_{p'+1}})$ and an at-most-$k+1$ edge path from $v_{\ell_{p-1}}$ to $v_{\ell_p}$ of length at most $d(v_{\ell_{p-1}},v_{\ell_p}) + \varepsilon D$.
  Furthermore, we have $p < \varepsilon^{-k}$ because $d(v_{\ell_{p'}},v_{\ell_{p'+1}}) \ge \varepsilon^k D$ for all $p'< p-1$.
  In total, this gives a path containing $(k+1)\cdot 2\varepsilon^{-k}\le O_{\delta,\varepsilon}(1)$ vertices and of total length at most
  \begin{align}
    \sum_{p'=0}^{p-2} (1+4\varepsilon)d(v_{\ell_{p'}}, v_{\ell_{p'+1}})
     + (d(v_{\ell_{p-1}},v_{\ell_p}) + 2\varepsilon D)
     &= 
    d(v,v') + \left(\sum_{p'=0}^{p-2} 4\varepsilon d(v_{\ell_{p'}}, v_{\ell_{p'+1}})\right) + 2\varepsilon D \nonumber\\
    &< d(v,v') + 6\varepsilon D.
  \end{align}
  Thus, $E'$ is indeed a $(O_{\delta,\varepsilon}(1), 6\varepsilon)$-additive hopset, as desired.
\end{proof}

\subsection{Directed hopset}
Lemma~\ref{lem:hopset-dir} is implicit in \cite{CaoFR20} and is confirmed by the authors \cite{CaoPrivate}.
Here we explain how Lemma~\ref{lem:hopset-dir} is implicit in \cite{CaoFR20}.
In \cite{CaoFR19}, building off of \cite{JLS19}, a parallel algorithm is given for single-source reachability and single-source approximate shortest paths on directed weighted graphs with $\tilde O(m)$ work and $n^{1/2+o(1)}$ span.
To achieve single-source approximate reachability, they find a set of shortcuts in $\tilde O(m)$ time that reduces the diameter to $n^{1/2+o(1)}$.
To achieve single-source approximate shortest paths, they use a construction of a $(n^{1/2+o(1)},\varepsilon)$ hopset on directed weighted graphs in $\tilde O_\varepsilon(m)$ time.
The hopsets are in fact a generalization of the shortcuts.

In \cite[Theorem 1.1]{CaoFR20}, the authors generalize this algorithm to obtain a tradeoff between work and span, solving single-source reachability and approximate single-source shortest paths in a directed weighted graph in $\tilde O(m\rho^4)$ (in fact, $\tilde O(m\rho^2+n\rho^4)$) work and $n^{1/2+o(1)}/\rho$ span.
When $\rho=1$, these algorithms are exactly the same as the one in \cite{CaoFR19}.
For larger $\rho$, they follow the same approach.
For single-source reachability, they construct shortcuts to reduce the diameter to $n^{1/2+o(1)}/\rho$ in time $\tilde O(m\rho^2)$.  
For single-source approximate shortest path, they construct $(n^{1/2+o(1)}/\rho,\varepsilon)$ hopset on directed weighted graphs in time $\tilde O_\varepsilon(m\rho^4)$.
However, as the paper is only an extended abstract, they only state the shortcut construction \cite[Theorem 4.1]{CaoFR20}, and not the full hopset construction.
The use of such a hopset construction was confirmed by the authors~\cite{CaoPrivate}.

To obtain the addition guarantee that, with probability 1, the edge additions preserve all shortest paths, we note that the algorithms in \cite{CaoFR19,CaoFR20} only adds an edge from $u$ to $v$ of weight $d_G(u,v)$, so the edge does not shorten the shortest path from $u$ to $v$ and thus any other vertex. Additionally, since edges are only added, the length of shortest paths never increase.

\section{Fine grained reductions}
\label{app:fgc}

\subsection{Applying \cite{CarmosinoGIMPS16} for promise problems.}

Lemma~\ref{lem:nseth} follows from the following property, in the same way that \cite[Theorem 5.1]{CarmosinoGIMPS16} follows from \cite[Lemma 3.5]{CarmosinoGIMPS16}.
This property states that fine grained reductions translate $\nconplus$ savings for promise problems to $\ncon$.
\begin{lemma}[Analogue of Lemma 3.5 of \cite{CarmosinoGIMPS16} for promise problems]
Let $(\Pi',T')\le_{FGR}(\Pi,T)$, where $\Pi,\Pi'$ are promise problems, and suppose $\Pi\in \nconplus[T(n)^{1-\varepsilon}]$ for some $\varepsilon>0$. Then there exists a $\delta>0$ such that $\Pi'\in\ncon[T'(n)^{1-\delta}]$.
\label{lem:nseth-c5}
\end{lemma}
\begin{proof}[Proof]
We proceed as in the proof in \cite{CarmosinoGIMPS16}.
We take $\mathcal{M}^{\Pi}$ to be the deterministic oracle that achieves a fine-grained reduction from $\Pi=(\Pi_{YES},\Pi_{NO})$ to $\Pi'=(\Pi'_{YES},\Pi'_{NO})$, and construct a nondeterministic machine $\mathcal{M}'$ for deciding $\Pi'$ to show that $\Pi'\in \ntime[T'(n)^{1-\delta}]$.
Let $\mathcal{A}_N,\mathcal{A}_{coN}$ be the algorithms for $\Pi$ in time $T$ given by Definition~\ref{def:nconplus}.
As in \cite{CarmosinoGIMPS16}, $\mathcal{M}'$ guesses a table of queries to $\Pi$ along with their answers.
$\mathcal{M}'$ then uses the nondeterministic algorithm $\mathcal{A}_N$ for $\Pi$ to verify the queries with ``YES'' answer, and the nondeterministic algorithm $\mathcal{A}_{coN}$ for $\Pi$ to verify the queries with ``NO'' answer.
We then simulate $\mathcal{M}$, looking up the answers to queries in the table.
If a query does not appear in the table, we output ``NO''.

We now check the correctness of the nondeterministic Turing machine $\mathcal{M}'$.
If the input to $\Pi'$ is in $\Pi_{NO}'$, then all nondeterministic choices result in an output of ``NO'' for the same reason as in \cite{CarmosinoGIMPS16}: unless the query table is incorrectly labels a query in $\Pi_{YES}$ with ``NO'' or labels a $\Pi_{NO}$ query with ``YES'', in which case the nondeterministic verifiers for $\Pi$ cause us to output ``NO'', $\mathcal{M}'$ will always output ``NO'' on an input in $\Pi_{NO}'$ by the definition of $\mathcal{M}'$.

We now check that for all inputs to $\Pi'$ is in $\Pi_{YES}'$, there exist nondeterministic choices such that $\mathcal{M}'$ outputs ``YES''.
We put all queries made by $\mathcal{M}'$ in the table, with their corresponding answers.
For all queries $q$ not in the promise $\Pi_{YES}\cup \Pi_{NO}$, if there exist nondeterministic choices for $\mathcal{A}_{N}$ to output ``YES'', then give query $q$ answer ``YES'', and if there exist nondeterministic choices for $\mathcal{A}_{coN}$ to output ``NO'', then give query $q$ answer ``NO''.
By property 3 of the definition of $\nconplus$, one of these two possibilities must exist for all queries $q$ not in the promise.
Now, there exist nondeterministic choices such that $\mathcal{M}'$ never outputs ``NO'' before simulating $\mathcal{M}$, because all queries in the promise $\Pi_{YES}\cup\Pi_{NO}$ can be correctly verified by properties 1 and 2 of $\nconplus$, and, by construction of the table, all queries not in the promise can also be verified.
Thus the simulation of $\mathcal{M}$ succeeds and we correctly output ``YES''.

As in \cite{CarmosinoGIMPS16}, the simulation takes time at most $(T')^{1-\delta}$, and guessing and verifying the query table takes $O((T')^{1-\delta})$ time.
Therefore $\Pi'\in \ntime[T'(n)^{1-\delta}]$.

Similarly, we can design a nondeterministic machine for inputs in $\Pi_{NO}'$ to conclude that $\Pi'\in\ncon[T'(n)^{1-\delta}]$.
\end{proof}

\begin{proof}[Proof of Lemma~\ref{lem:nunseth}]
  The proof is the same as in \cite{CarmosinoGIMPS16}, except that when we nondeterministically simulate the oracle queries to the promise problem $\Pi$, we simulate as in Lemma~\ref{lem:nseth-c5} rather than as in \cite[Lemma C.5]{CarmosinoGIMPS16}.

  Suppose for contradiction that $(\Pi,T^{1+\delta})$ is SETH-hard so that there is a randomized reduction from \textsc{CNFSAF} with time $2^n$ to $\Pi$ with time $T^{1+\delta}$, and let $\mathcal{M}^\Pi$ be the randomized oracle machine.
  Then there exists some $\varepsilon>0$ such that $\TIME[\mathcal{M}]\le 2^{n(1-\varepsilon)}$ and such that the query lengths satisfy $\sum_{q\in\tilde Q(\mathcal{M},x)}^{} T(|q|)\le 2^{n(1-\varepsilon)}$ (here we switched $\varepsilon$ and $\delta$ from the definition in Section~\ref{sec:prelims}).

  Let $m<n^k$ be the length in bits of a description of a k-SAT formula on $n$ inputs.
  By repeating $\mathcal{M}^\Pi$ $O(m)$ times and taking the majority answer, we can make the error probability less than $2^{-m}$.
  Since there are at most $2^m$ possible inputs, there is one random tape such that $\mathcal{M}$ has no errors when all the oracle queries to $\Pi$ are correct.
  Since $\mathcal{M}$ runs in total time $2^{(1-\varepsilon)n}$, this tape will have length at most $O(m2^{(1-\varepsilon)n})$.
  Then simulate oracle queries nondeterministically, using this nondeterministically chosen random tape, using the same simulation as Lemma~\ref{lem:nseth-c5}.
  This simulation has time complexity $O(m2^{(1-\varepsilon)n})$, giving a nondeterministic circuit with total size $O(m2^{(1-\varepsilon)n})$.
\end{proof}

\subsection{Reducing approximate Diameter to the promise problem}
Let $\rho\in[1,\infty]$, possibly as some function of $m$ and $n$.
For $\alpha\ge 1$, let $f_{\alpha,\rho}$ be the function problem $f$ of giving an $\alpha$-approximation of Diameter on graphs satisfying $W_{max}/W_{min}\le \rho$ (all graphs if $\rho=\infty$).
Let $\Pi_{\alpha,\rho}$ be the gap-problem $\alpha/1$-Diameter on graphs satisfying $W_{max}/W_{min}\le \rho$.
\begin{lemma}
For all $\rho\in[1,\infty]$, all constant $\alpha\ge 1$, all constant $\beta>0$, and all time complexities $T$, there is a fine grained reduction from $(f_{\alpha+\beta,\rho},T)$ to $(\Pi_{\alpha,\rho},T)$.
\label{lem:gap-to-approx-0}
\end{lemma}
\begin{proof}
Suppose we are given an oracle to $\alpha/1$-Diameter on graphs with $W_{max}/W_{min}\le \rho$.
If we had oracles to $\alpha D/D$-Diameter for all $D$, we can compute $f$ on graphs with $W_{max}/W_{min}\le \rho$ by simply running a binary search on the same graph to accuracy $\beta$, querying $\alpha D/D$-Diameter for different values of $D$, to obtain an interval $[x,(\alpha+\beta)x]$ containing the correct answer.
However, we can replace all $\alpha D/D$-Diameter queries with $\alpha/1$-Diameter queries by simply re-weighting the graph.
This preserves the weight condition $W_{max}/W_{min}\le \rho$.
This gives a Turing reduction $\mathcal{M}^{\Pi_{\alpha,\rho}}$ running in time $O(\log(1/\beta)) = O_\beta(1)$.
Let $\tilde Q(\mathcal{M},x)$ denote the set of queries made by $\mathcal{M}$ to the oracle on an input $x$ of length $n$. 
There are $O_\beta(1)$ queries, so, for all $\varepsilon>0$ the query lengths obey the following time bound.
\begin{align}
  \sum_{q\in \tilde Q(\mathcal{M},x)}^{} (T(|q|))^{1-\varepsilon} \le O_\beta(T(n)^{1-\varepsilon}).
\end{align}
Hence, this is a fine grained reduction from $(f_{\alpha+\beta,\rho},T)$ to $(\Pi_{\alpha,\rho},T)$.
\end{proof}
We now can prove Lemma~\ref{lem:gap-to-approx}.
\begin{proof}[Proof of Lemma~\ref{lem:gap-to-approx}]
First we consider deterministic reductions. Lemma~\ref{lem:gap-to-approx-0} proves that $(f_{\alpha+\beta,\rho},T)\le_{FGR}(\Pi_{\alpha,\rho}, T)$.
By composition of fine-grained reductions\footnote{Here, we need \cite[Lemma 3.7]{CarmosinoGIMPS16} to hold for promise problems, but the proof is the same so we omit it.}, if $(\textsc{CNFSAT},2^n)\le_{FGR} (f_{\alpha+\beta,\rho},T)$, then we also have $(\textsc{CNFSAT},2^n)\le_{FGR} (\Pi_{\alpha,\rho}, T)$.
Hence, taking the contrapositive, we have that if $(\textsc{CNFSAT},2^n)\not\le_{FGR} (\Pi_{\alpha,\rho}, T)$, then $(\textsc{CNFSAT},2^n) \not \le_{FGR} (f_{\alpha+\beta,\rho},T)$, which is exactly what Lemma~\ref{lem:gap-to-approx} asks to prove.

To show that non-SETH-hardness of the promise problem implies non-SETH-hardness of the approximation problem for \emph{randomized algorithms}, the proof is the same, but we need to show that randomized fine-grained reductions are closed under composition. We prove that these reductions are closed under composition in the scenario that we need.
\begin{lemma}
  Let $(A,T_A)\le_{FGR,r}(B,T_B)$ and $(B,T_B)\le_{FGR,r}(C,T_C)$ where $B$ is an approximate function problem. Then $(A,T_A)\le_{FGR,r}(C,T_C)$.
\label{lem:fgc-comp}
\end{lemma}
\begin{proof}
The proof is the same as \cite[Lemma C.7]{CarmosinoGIMPS16}, except that we simulate the oracle machines multiple times to amplify the failure probabilities.
Specifically, let $\mathcal{M}_{AB}^B$ be the (probabilistic) machine that achieves a fine-grained reduction from $(A,T_A)$ to $(B,T_B)$, and let $\mathcal{M}_{BC}^C$ be the machine that achieves a fine-grained reduction from $(B,T_B)$ to $(C,T_C)$.
We construct a machine $\mathcal{M}_{AC}^C$ that achieves a fine-grained reduction from $(A,T_A)$ to $(C,T_C)$.

Let $\mathcal{M}_{AC}^C$ on an input to $A$ of length $n_A$ simulate $\mathcal{M}_{AB}^B$ \emph{3 times}, and for each oracle call to $B$, $\mathcal{M}_{AC}^C$ simulates $\mathcal{M}_{BC}^C$ \emph{$O(\log T_A(n_A))$ times} and takes the \emph{median} output.
Problem $B$ has a correct output if it is in some range $[OPT,\alpha\cdot OPT]$, and by definition of $\mathcal{M}_{BC}^C$ at least 2/3 of the outputs to $\mathcal{M}_{BC}^C$ lie in that range $[OPT,\alpha\cdot OPT]$ in expectation.
Thus, with probability $\frac{1}{100}T_A(n_A)^{-2}$ (if the number of repeats of $\mathcal{M}_{BC}^C$ is enough), the median output of $\mathcal{M}_{BC}^C$ is in the correct range $[OPT,\alpha\cdot OPT]$.
Since there are at most $T_A(n_A)$ oracle calls to $B$, with probability at least $1 - \frac{3}{100}T_A(n_A)^{-1} > 0.97$, all of the oracle calls to $B$ have a correct output.
Furthermore, given that the oracle calls to $B$ have a correct output, one call to $\mathcal{M}_{AB}^B$ gives the incorrect output to $A$ with probability at most $\frac{1}{3}$, so the probability that at least two of the three calls to $\mathcal{M}_{AB}^B$ give the incorrect output is at most $(1/3)^3+3\cdot(2/3)\cdot(1/3)^2 < 0.26$.
Thus, the probability that $\mathcal{M}_{AC}^C$ gives an incorrect output is at most $0.26 + 0.03 < \frac{1}{3}$, as desired.

To prove that $\mathcal{M}_{AC}^C$ satisfies the required time and query length bounds so that it indeed achieves a fine grained reduction from $(A,T_A)$ to $(C,T_C)$, we follow the same proof as in \cite{CarmosinoGIMPS16}.
The only difference is that the runtime of $\mathcal{M}_{AC}^C$ and the number of queries is multiplied by $O(\log T_A(n_A))$, which has a negligible effect on the overall analysis.
\end{proof}
Applying Lemma~\ref{lem:fgc-comp} completes the proof of Lemma~\ref{lem:gap-to-approx} for randomized reductions.
\end{proof}


\begin{thebibliography}{DWVW19}

\bibitem[ABMR11]{ABMR11}
Eyad Alkassar,  Sascha B\"ohme, Kurt Mehlhorn, and Christine Rizkallah. 
\newblock Verification of certifying computations.
\newblock In {\em International Conference on Computer Aided Verification}, pages 67--82, 2011.

\bibitem[ACIM99]{AingworthCIM99}
Donald Aingworth, Chandra Chekuri, Piotr Indyk, and Rajeev Motwani.
\newblock Fast estimation of diameter and shortest paths (without matrix
  multiplication).
\newblock {\em {SIAM} J. Comput.}, 28(4):1167--1181, 1999.

\bibitem[AKT20]{AKT20}
Amir Abboud, Robert Krauthgamer, and Ohad Trabelsi
\newblock New algorithms and lower bounds for all-pairs max-flow in undirected graphs.
\newblock In {\em Proceedings of the Fourteenth Annual {ACM-SIAM} Symposium on Discrete Algorithms {SODA}, 2020}, pages 48--61, 2020.

\bibitem[AKT20b]{AKT20b}
Amir Abboud, Robert Krauthgamer, and Ohad Trabelsi
\newblock Subcubic Algorithms for Gomory-Hu Tree in Unweighted Graph. arXiv preprint arXiv:2012.10281.


\bibitem[AWW16]{AbboudWW16}
Amir Abboud, Virginia~Vassilevska Williams, and Joshua~R. Wang.
\newblock Approximation and fixed parameter subquadratic algorithms for radius
  and diameter in sparse graphs.
\newblock In {\em Proceedings of the Twenty-Seventh Annual {ACM-SIAM} Symposium
  on Discrete Algorithms, {SODA} 2016, Arlington, VA, USA, January 10-12,
  2016}, pages 377--391, 2016.

\bibitem[BRS{\etalchar{+}}18]{BackursRSWW18}
Arturs Backurs, Liam Roditty, Gilad Segal, Virginia~Vassilevska Williams, and
  Nicole Wein.
\newblock Towards tight approximation bounds for graph diameter and
  eccentricities.
\newblock In {\em Proceedings of the 50th Annual {ACM} {SIGACT} Symposium on
  Theory of Computing, {STOC} 2018, Los Angeles, CA, USA, June 25-29, 2018},
  pages 267--280, 2018.


\bibitem[BBST07]{BBST07}
Boaz Ben-Moshe, Binay Bhattacharya, Qiaosheng Shi, and Arie Tamir. "Efficient algorithms for center problems in cactus networks." \emph{Theoretical Computer Science}, 378(3): 237--252, 2007. 

\bibitem[BN19]{BentertN19}
Matthias Bentert, and Andr\'{e} Nichterlein. Parameterized complexity of diameter. In \emph{International Conference on Algorithms and Complexity}, pp. 50-61. Springer, Cham, 2019.


\bibitem[Bon20]{Bonnet20}
Edouard Bonnet.
\newblock Inapproximability of Diameter in super-linear time: Beyond the 5/3 ratio.
\newblock In {\em Symposium on Theoretical Aspects of Computer Science, STACS 2021}.

\bibitem[BCH{\etalchar{+}}15]{BorassiCHKMT15}
Michele Borassi, Pierluigi Crescenzi, Michel Habib, Walter~A. Kosters, Andrea
  Marino, and Frank~W. Takes.
\newblock Fast diameter and radius bfs-based computation in (weakly connected)
  real-world graphs: With an application to the six degrees of separation
  games.
\newblock {\em Theor. Comput. Sci.}, 586:59--80, 2015.

\bibitem[BE05]{BE05}
Ulrik Brandes and Thomas Erlebach.
\newblock {\em Network Analysis: Methodological Foundations.}
\newblock Springer-Verlag, 2005.

\bibitem[BHM19]{BringmanHM18}
Karl Bringmann, Thore Husfeldt, and Måns Magnusson. Multivariate Analysis of Orthogonal Range Searching and Graph Distances. \emph{Algorithmica}, 2020. 

\bibitem[CDHP01]{CDHP01}
Derek G. Corneil, Feodor F. Dragan, Michel Habib, and Christophe Paul. Diameter determination on restricted graph families. \emph{Discrete Applied Mathematics} 113(2-3): 143--166, 2001.

\bibitem[CGR16]{CairoGR16}
Massimo Cairo, Roberto Grossi, and Romeo Rizzi.
\newblock New bounds for approximating extremal distances in undirected graphs.
\newblock In {\em Proceedings of the Twenty-Seventh Annual {ACM-SIAM} Symposium
  on Discrete Algorithms, {SODA} 2016, Arlington, VA, USA, January 10-12,
  2016}, pages 363--376, 2016.

\bibitem[CFR]{CaoPrivate}
Nairen Cao, Jeremy T. Fineman, and Katina Russell.
\newblock Private communication.

\bibitem[CFR19]{CaoFR19}
Nairen Cao, Jeremy T. Fineman, and Katina Russell.
\newblock Efficient Construction of Directed Hopsets and Parallel Approximate Shortest Paths.
\newblock In {\em Proceedings of the 52nd Annual {ACM} {SIGACT} Symposium on
  Theory of Computing, {STOC} 2020},
  pages 336--349, 2020.


\bibitem[CFR20]{CaoFR20}
Nairen Cao, Jeremy T. Fineman, and Katina Russell.
\newblock Brief Announcement: Improved Work Span Tradeoff for Single Source Reachability and Approximate Shortest Paths.
\newblock In {\em Proceedings of the 32nd ACM Symposium on Parallelism in Algorithms and Architectures (SPAA), July, 2020}, pages 511--513.

\bibitem[CGI{\etalchar{+}}16]{CarmosinoGIMPS16}
Marco~L. Carmosino, Jiawei Gao, Russell Impagliazzo, Ivan Mihajlin, Ramamohan Paturi, and Stefan Schneider. 
\newblock Nondeterministic extensions of the strong exponential time hypothesis and consequences for non-reducibility. 
\newblock In {\em Proceedings of the 2016 ACM Conference on Innovations in Theoretical Computer Science}, page 261--270, 2016.

\bibitem[CLR{\etalchar{+}}14]{ChechikLRSTW14}
Shiri Chechik, Daniel~H. Larkin, Liam Roditty, Grant Schoenebeck, Robert~Endre
  Tarjan, and Virginia~Vassilevska Williams.
\newblock Better approximation algorithms for the graph diameter.
\newblock In {\em Proceedings of the Twenty-Fifth Annual {ACM-SIAM} Symposium
  on Discrete Algorithms, {SODA} 2014, Portland, Oregon, USA, January 5-7,
  2014}, pages 1041--1052, 2014.


\bibitem[CDV02]{CDV02}
Victor Chepoi, Feodor Dragan, and Yann Vaxès. Center and diameter problems in plane triangulations and quadrangulations. In \emph{Proc. SODA}, pp. 346--355, 2002.

\bibitem[Coh00]{Cohen00}
Edith Cohen.
\newblock Polylog-time and near-linear work approximation scheme for undirected shortest paths.
\newblock \emph{J. ACM}, 47(1):132--166, 2000.

\bibitem[CDP19]{CDP19}
David Coudert, Guillaume Ducoffe, and Alexandru Popa. Fully polynomial FPT algorithms for some classes of bounded clique-width graphs. \emph{ACM Transactions on Algorithms (TALG)}, 15(3): 1--57, 2019.


\bibitem[CGS15]{CyganGS15}
Marek Cygan, Harold N. Gabow, and Piotr Sankowski. 
\newblock Algorithmic applications of baur-strassen’s theorem: Shortest cycles, diameter, and matchings. 
\newblock \emph{J. ACM}, 62(4):28:1--28:30, 2015. 


\bibitem[DWV{\etalchar{+}}19]{DalirrooyfardW019}
Mina Dalirrooyfard, Virginia~Vassilevska Williams, Nikhil Vyas, Nicole Wein,
  Yinzhan Xu, and Yuancheng Yu.
\newblock Approximation algorithms for min-distance problems.
\newblock In {\em 46th International Colloquium on Automata, Languages, and
  Programming, {ICALP} 2019, July 9-12, 2019, Patras, Greece}, pages
  46:1--46:14, 2019.

\bibitem[DWVW19]{DalirrooyfardW019a}
Mina Dalirrooyfard, Virginia~Vassilevska Williams, Nikhil Vyas, and Nicole
  Wein.
\newblock Tight approximation algorithms for bichromatic graph diameter and
  related problems.
\newblock In {\em 46th International Colloquium on Automata, Languages, and
  Programming, {ICALP} 2019, July 9-12, 2019, Patras, Greece}, pages
  47:1--47:15, 2019.

\bibitem[DW20]{DalirrooyfardW2020}
Mina Dalirrooyfard and Nicole Wein.
\newblock Tight Conditional Lower Bounds for Approximating Diameter in Directed Graphs. \newblock In {\em Symposium on Theory of Computing, STOC 2021}, to appear.

\bibitem[Dam16]{Damaschke16}
Peter, Damaschke. Computing giant graph diameters. In \emph{International Workshop on Combinatorial Algorithms}, pp. 373--384. Springer, Cham, 2016.

\bibitem[Duc18]{Ducoffe18}
Guillaume Ducoffe,
A New Application of Orthogonal Range Searching for Computing Giant Graph Diameters. In \emph{2nd Symposium on Simplicity in Algorithms (SOSA 2019)}. Schloss Dagstuhl-Leibniz-Zentrum fuer Informatik, 2019.


\bibitem[DHV20]{DHV20}
Guillaume Ducoffe, Michel Habib, and Laurent Viennot. Diameter computation on H-minor free graphs and graphs of bounded (distance) VC-dimension. In \emph{Proceedings of the Fourteenth Annual ACM-SIAM Symposium on Discrete Algorithms}, pp. 1905--1922. Society for Industrial and Applied Mathematics, 2020.

\bibitem[Epp00]{Eppstein00}
David Eppstein,
Diameter and treewidth in minor-closed graph families. \emph{Algorithmica}, 27(3-4):275–291, 2000.

\bibitem[FP80]{FarleyP80}
Arthur M. Farley, and Andrzej Proskurowski. Computation of the center and diameter of outerplanar graphs. \emph{Discrete Applied Mathematics} 2(3):185--191, 1980.
Harvard

\bibitem[GKM\etalchar{+}18]{GKMSW18}
Pawl Gawrychowski, Haim Kaplan, Shay Mozes, Micha Sharir, and Oren Weimann. Voronoi diagrams on planar graphs, and computing the diameter in deterministic $\tilde O(n^{5/3})$ time. In \emph{Proceedings of the Twenty-Ninth Annual ACM-SIAM Symposium on Discrete Algorithms}, pp. 495--514. Society for Industrial and Applied Mathematics, 2018.

\bibitem[Gol05]{Goldreich05}
Oded Goldreich,
\newblock On promise problems: A survey. 
\newblock \emph{Theoretical computer science.} Springer, Berlin, Heidelberg, 254--290, 2006.

\bibitem[IPZ01]{ImpagliazzoPZ01}
Russell Impagliazzo, Ramamohan Paturi, and Francis Zane.
\newblock Which problems have strongly exponential complexity?
\newblock {\em J. Comput. Syst. Sci.}, 63(4):512--530, 2001.

\bibitem[JLS19]{JLS19}
Arun Jambulapati, Yang Liu, and Aaron Sidford.
\newblock Parallel Reachability in Almost Linear Work and Square Root Depth.
\newblock \emph{2019 IEEE 60th Annual Symposium on Foundations of Computer Science (FOCS).} IEEE, 2019.


\bibitem[Kun18]{Kun18}
Marvin K\"unnemann, On Nondeterministic Derandomization of Freivalds' Algorithm: Consequences, Avenues and Algorithmic Progress. 
\newblock In {\em 26th Annual European Symposium on Algorithms (ESA 2018)}, 2018.

\bibitem[Le14]{Le14}
Fran\c{c}ois Le Gall. 
Powers of tensors and fast matrix multiplication. 
In \emph{International Symposium on Symbolic and Algebraic Computation, ISSAC ’14, Kobe, Japan, July 23-25, 2014}, pages 296--303, 2014.

\bibitem[LWCW16]{LinWCW16}
Ting-Chun Lin, Mei-Jin Wu, Wei-Jie Chen, and Bang-Ye Wu.
\newblock Computing the diameters of huge social networks.
\newblock In {\em 2016 International Computer Symposium (ICS)}, pages 6--11.
  IEEE, 2016.


\bibitem[MMNS11]{MMNS11}
Ross M. McConnell, Kurt Mehlhorn, Stefan N\"aher, and Pascal Schweitzer. 
\newblock Certifying algorithms.
\newblock {\em Computer Science Review}, 5(2):119--161, 2011.

\bibitem[PRT12]{PelegRT12}
David Peleg, Liam Roditty, and Elad Tal.
\newblock Distributed algorithms for network diameter and girth.
\newblock In {\em Automata, Languages, and Programming - 39th International
  Colloquium, {ICALP} 2012, Warwick, UK, July 9-13, 2012, Proceedings, Part
  {II}}, pages 660--672, 2012.

\bibitem[Pet04]{Pettie04}
Seth Pettie. 
\newblock A new approach to all-pairs shortest paths on real-weighted graphs. 
\newblock \emph{Theor. Comput. Sci.}, 312(1):47--74, 2004. 

\bibitem[PR05]{PettieR05}
Seth Pettie and Vijaya Ramachandran. 
\newblock A shortest path algorithm for real-weighted undirected graphs. 
\newblock \emph{SIAM J. Comput.}, 34(6):1398--1431, 2005. 



\bibitem[RW13]{RodittyW13}
Liam Roditty and Virginia~Vassilevska Williams.
\newblock Fast approximation algorithms for the diameter and radius of sparse
  graphs.
\newblock In {\em Symposium on Theory of Computing Conference, STOC'13, Palo
  Alto, CA, USA, June 1-4, 2013}, pages 515--524, 2013.

\bibitem[RW19]{RubinsteinW19}
Aviad Rubinstein and Virginia~Vassilevska Williams.
\newblock {SETH} vs approximation.
\newblock {\em {SIGACT} News}, 50(4):57--76, 2019.

\bibitem[Sto10]{Stothers10}
Andrew Stothers. 
On the complexity of matrix multiplication. \emph{Ph.D. Thesis, U. Edinburgh,} 2010.

\bibitem[WS98]{WattsS98}
D. J. Watts and S. H. Strogatz. 
Collective dynamics of `small-world' networks. \emph{Nature}, 393:440--442, 1998.


\bibitem[Wil05]{Williams05}
Ryan Williams.
\newblock A new algorithm for optimal 2-constraint satisfaction and its
  implications.
\newblock {\em Theor. Comput. Sci.}, 348(2-3):357--365, 2005.

\bibitem[Wil12]{Williams12}
Virginia Vassilevska Williams. 
\newblock Multiplying matrices faster than coppersmith-winograd. 
\newblock In {\em Proceedings of the forty-fourth annual ACM symposium on Theory of computing}, pages 887--898. ACM, 2012.


\bibitem[Wil14]{Williams14}
Ryan Williams.
\newblock Faster all-pairs shortest paths via circuit complexity. 
\newblock In {\em Symposium on Theory of Computing, STOC 2014, New York, NY, USA, May 31 - June 03, 2014}, pages 664–673, 2014.

\bibitem[Wil16]{Williams16}
Richard Ryan Williams. 
\newblock Strong ETH breaks with merlin and arthur: Short non-interactive proofs of batch evaluation. In {\em 31st Conference on Computational Complexity, CCC 2016, May 29 to June 1, 2016}, Tokyo, Japan, pages 2:1–-2:17, 2016.

\bibitem[Wil18]{Williams18}
Virginia~Vassilevska Williams.
\newblock On some fine-grained questions in algorithms and complexity.
\newblock In {\em Proceedings of the ICM}, volume~3. World Scientific, 2018.

\end{thebibliography}
\end{document}